\documentclass[12pt]{article}
\usepackage{enumerate}
\usepackage{natbib}
\usepackage{authblk}
\usepackage{url}            % simple URL typesetting
\usepackage{amsfonts}       % blackboard math symbols
\usepackage{nicefrac}       % compact symbols for 1/2, etc.
\usepackage{microtype}      % microtypography

\usepackage{wrapfig}
\usepackage{graphicx}
\usepackage{subcaption}
\usepackage{booktabs} % for professional tables
\usepackage{amsmath}
\usepackage{amsthm}
\usepackage{amssymb}
\usepackage{mathtools}
\usepackage{algorithm,algorithmicx,algpseudocode}
\usepackage{thmtools}
\usepackage{thm-restate}
\usepackage{hyperref}
\usepackage{setspace}
\hypersetup{hidelinks}
\usepackage{geometry}
\usepackage{amsmath}
\allowdisplaybreaks[4]
\usepackage{float}
\usepackage{placeins}
\usepackage{mathrsfs}
\usepackage{amsfonts}   
\usepackage{xcolor}

\usepackage{times}
\usepackage{bm}
\usepackage{colonequals}
\usepackage{graphicx}
\usepackage{caption}
\usepackage{enumerate}
\usepackage{cmll}
\usepackage{geometry}
\usepackage{amsthm}
\usepackage{subcaption}
\usepackage{threeparttable}
\usepackage{booktabs}
\usepackage{multirow}
\usepackage{makecell}

\makeatletter

\makeatletter
\newcommand{\bz}{\boldsymbol{z}}
\newcommand{\bx}{\boldsymbol{x}}

\newcommand{\bbf}{\boldsymbol{f}}
\newcommand{\be}{\boldsymbol{e}}
\newcommand{\bX}{\boldsymbol{X}}
\newcommand{\bZ}{\boldsymbol{Z}}
\newcommand{\bW}{\boldsymbol{W}}

\newcommand{\bepsilon}{\boldsymbol{\epsilon}}
\newcommand{\bxi}{\boldsymbol{\xi}}
\newcommand{\bmu}{\boldsymbol{\mu}}

\newcommand{\bzeta}{\boldsymbol{\zeta}}
\newcommand{\bGamma}{\boldsymbol{\Gamma}}
\newcommand{\bbeta}{\boldsymbol{\beta}}

\newcommand{\bs}{\boldsymbol{s}}
\newcommand{\bu}{\boldsymbol{u}}

\newcommand{\bg}{\boldsymbol{g}}
\newcommand{\by}{\boldsymbol{y}}
\newcommand{\bY}{\boldsymbol{Y}}

\newcommand{\bzero}{\boldsymbol{0}}

\newcommand{\beps}{\boldsymbol{\epsilon}}
\newcommand{\btheta}{\boldsymbol{\theta}}
\newcommand{\bkappa}{\boldsymbol{\kappa}}

\newcommand{\BK}{\mathbf{K}}

\newcommand{\bI}{\boldsymbol{I}}

\newcommand{\bV}{\boldsymbol{V}}
\newcommand{\bbP}{\mathbb{P}}

\newcommand{\bbR}{\mathbb{R}}
\newcommand{\bbC}{\mathbb{C}}

\newcommand{\bv}{\boldsymbol{v}}

\newcommand{\sW}{\mathsf{W}}

\newtheorem{assumption}{\textbf{Assumption}}

\newtheorem{corollary}{\textbf{Corollary}}
\newtheorem{lemma}{\textbf{Lemma}}
\newtheorem{theorem}{\textbf{Theorem}}
\newtheorem{proposition}{\textbf{Proposition}}
\newtheorem{remark}{\textbf{Remark}}

\newcommand{\mE}{{\mathbb{E}}}

\newcommand{\TV}{\mathsf{TV}}

\newcommand{\cD}{\mathcal{D}}

\newcommand{\cL}{\mathcal{L}}

\newcommand{\cC}{\mathcal{C}}
\newcommand{\cN}{\mathcal{N}}
\newcommand{\cB}{\mathcal{B}}

\newcommand{\cT}{\mathcal{T}}

\newcommand{\cK}{\mathcal{K}}

\newcommand{\cH}{\mathcal{H}}

\newcommand{\cO}{\mathcal{O}}

\newcommand{\sD}{\mathscr{D}}

\newcommand{\tX}{\widetilde{\bX}}

\newcommand{\tZ}{\widetilde{\bZ}}

\newcommand{\tr}{\mathrm{tr}}

\newcommand{\bbSig}{\mathbf{\Sigma}}
\newcommand{\bbOmega}{\mathbf{\Omega}}

\newcommand{\bbH}{\mathbf{H}}

\newcommand{\bcK}{\boldsymbol{\cK}}

\newcommand{\hZ}{\widehat{\bZ}}

\newcommand{\hX}{\widehat{\bX}}

\newcommand{\htheta}{\widehat{\boldsymbol{\theta}}}
\newcommand{\hbeta}{\widehat{\boldsymbol{\beta}}}

\def\T{{\rm T}}

\usepackage{titlesec}

\titlespacing*{\section}
  {0pt}    % Indentación a la izquierda
  {10pt}   % Espacio antes del título
  {5pt}    % Espacio después del título

\titlespacing*{\subsection}
  {0pt}
  {8pt}    % Espacio antes del título
  {4pt}    % Espacio después del títul

\newtheorem*{assumption*}{Assumption}

\makeatother

\usepackage{easyReview}

\usepackage{lipsum}

\newcommand\blfootnote[1]{%
	\begingroup
	\renewcommand\thefootnote{}\footnote{#1}%
	\addtocounter{footnote}{-1}%
	\endgroup
}

	\usepackage{xcolor}

\begin{document}
\pagenumbering{gobble}
\setstretch{1}
\title{\bfseries Denoising Data with Measurement Error Using a Reproducing Kernel-Based Diffusion Model}
\author[a]{Mingyang Yi}
\author[b]{Marcos Matabuena}
\author[c]{Zhi-Ming Ma}
\author[d]{Ruoyu Wang$^*$}
\affil[a]{School of Information, Renmin University of China, Beijing, China}
\affil[b]{Mohamed bin Zayed University of Artificial Intelligence, Masdar City, Abu Dhabi, UAE}
\affil[c]{Academy of Mathematics and Systems Science, Chinese Academy of Sciences, Beijing, China}
\affil[d]{Department of Statistics and Data Science, Tsinghua University, Beijing, China}
\date{}
\maketitle
\blfootnote{$*$ Corresponding author. Email: ruoyuwang@mail.tsinghua.edu.cn}
\bigskip

\begin{abstract}
		Recent advancements in measurement technologies have facilitated the acquisition of high-resolution data in fields such as engineering, biology, and medicine. However, these observations are frequently corrupted by measurement noise. To address this challenge, we propose a denoising framework that employs diffusion models to recover the underlying error-free data distribution, thereby enabling reliable downstream analysis. Central to our framework is a novel Reproducing Kernel Hilbert Space (RKHS)-based method that trains the diffusion model using only error-contaminated data. This method admits a closed-form solution and achieves a fast convergence rate in terms of estimation error. Theoretically, we validate our approach by deriving an upper bound on the total variation distance between the distributions of the generated denoised data and the true error-free data. Extensive simulations demonstrate the superior empirical performance of our method compared to state-of-the-art baselines. Finally, we illustrate the framework's practical utility in digital health and epidemiology by applying it to Continuous Glucose Monitor (CGM) data and the association analysis between blood lipids and the coronary heart disease.
	\end{abstract}
	\vspace{-0.2in}
	\noindent%
	{\it Keywords:} Digital health, Error-contaminated data, Generative model, Machine learning, Score matching.
	\pagenumbering{arabic}
	\setstretch{1.9}
	\section{Introduction}\label{sec:introduction}
	In the current era of artificial intelligence, machine learning has become central to applied data analysis. Concurrently, advancements in measurement systems now enable the acquisition of high-resolution, large-scale datasets across diverse scientific domains, including engineering, biology, and medicine. However, these measurements are frequently compromised by significant errors, which can fundamentally undermine scientific discoveries and predictive accuracy \citep{carroll2006measurement}. Obtaining error-free measurements remains a pervasive challenge across modern research landscapes, ranging from epidemiology \citep{doi:10.1126/science.aal3618}, clinical trials \citep{doi:10.1191/1740774504cn057oa}, and agriculture \citep{desiere2018land}, to econometrics \citep{thompson2007overview}, computer vision \citep{brodley1999identifying,xiao2015learning}, and natural language processing \citep{garg2021towards,havrilla2024understanding}. Crucially, many state-of-the-art machine learning methods are not designed to accommodate such measurement errors \citep{li2021learning,song2022learning}; ignoring them can severely affect decision-making processes. For instance, in medical contexts, inaccurate diagnoses or patient misclassifications stemming from noisy data can lead to severe consequences from clinical, social, and economic perspectives.
	
	\vspace{-0.1in}
	Methods that account for potential measurement errors and produce valid statistical inferences have predominantly arisen in statistical science applications, such as in medicine, and with a historical focus on signal processing. Examples of these methods include the observed likelihood approach \citep{zucker2005pseudo}, the conditional score method \citep{stefanski1987conditional}, the corrected estimating equation method \citep{nakamura1990corrected, wang2000expected, wang2012corrected}, and simulation-extrapolation techniques \citep{carroll1996asymptotics}. These classical methods address the measurement error by constructing asymptotically unbiased estimators from error-contaminated data. However, deriving such estimators is nontrivial, often model-specific, and may require expertise beyond that of the typical investigator. In addition, the above methods are based on asymptotic theories for classical estimators such as $M$-estimators. It is not straightforward to extend these methods to accommodate common machine learning methods whose asymptotic properties can be complicated and understudied.
	
	\vspace{-0.1in}
	Another line of research considers the deconvolution problem, which aims to recover the distribution of error-free data from error-contaminated observations. Most existing deconvolution methods focus on the univariate case \citep{zhang1990fourier, fan1991optimal, cordy1997deconvolution, delaigle2008using, meister2010deconvolution, lounici2011global, guan2021fast, belomestny2021density}. A handful of approaches address the multivariate setting, typically via normal mixture models \citep{bovy2011extreme, sarkar2018bayesian} or kernel smoothing \citep{masry1993strong, lepski2019oracle}. However, normal mixture models can be highly sensitive to misspecification, whereas kernel smoothing methods tend to converge slowly, at polynomial rates in $1/\log n$ where $n$ is the sample size for deconvolution problems with normal error \citep{fan1991optimal}, and suffer from the curse of dimensionality in multivariate scenarios \citep{masry1993strong}. Consequently, these limitations hinder the application of classical non-parametric statistical methods to more complex models across diverse research areas and emerging scientific problems. 
	For example, modern data collection practices have enabled high-dimensional data, such as medical imaging, and other emerging data structures \citep{loh2012high, jiang2023high}. These developments pose additional challenges for classical methods. As a result, researchers are increasingly turning to machine learning approaches, such as neural networks \citep{lecun2015deep} and kernel-based statistical learning in reproducing kernel Hilbert spaces (RKHS) \citep{de2005learning}, which can empirically outperform classical nonparametric techniques such as kernel smoothing in high-dimensional settings.
	
	In the emerging field of machine learning, learning with error-contaminated data has received significant attention, especially for deep learning models \citep{lecun2015deep}. In contrast to statistics, which focuses on statistical inference based on error-contaminated data, researchers in the AI community are primarily concerned with training models (usually neural networks) on error-contaminated data, aiming for accurate predictions based solely on error-free data. Methods in this field often focus on constructing ``robust loss'' functions, which enhance model robustness when trained on error-contaminated data \citep{song2022learning}. However, these methods mainly apply to scenarios with low-level errors and are restricted to classification problems. 
	
	\vspace{-0.2in}
	The combination of statistical inferential thinking with machine learning offers the opportunity to develop novel methods for tackling the measurement error problem. 
	In this paper, we explore this interface through diffusion models, a powerful class of generative models originating from machine learning \citep{hyvarinen2005estimation,vincent2011connection,ho2020denoising,song2020score}, which have been shown to outperform classical generative models in multiple practical applications \citep{song2020denoising,ho2020denoising,rombach2022high,song2021maximum}. Specifically, we focus on normal measurement error, which is commonly adopted in the measurement error literature \citep{yi2017statistical}. In contrast to the aforementioned standard methods, we directly generate denoised data whose distribution is close to the error-free distribution, so that downstream analysis can be conducted based on the denoised data. The idea of a diffusion model is to transform normal noise into a target distribution by gradually removing noise through a diffusion process. In our setting, this trajectory can be viewed schematically as ``error-contaminated data'' $\to$ ``data with less measurement error'' $\to$ ``error-free data'', making diffusion models a natural fit for our problem. However, applying the diffusion technique is not straightforward in this context. Specifically, the diffusion model requires estimating the ``score function'' \citep{hyvarinen2005estimation} of the distribution of intermediate error-contaminated data \citep{song2020denoising,ho2020denoising}, which is the drift term \citep{oksendal2013stochastic} in the SDE characterizing the denoising process. Unfortunately, this estimation is nontrivial because the intermediate error-contaminated data are unobserved in practice. To overcome this challenge, we propose a \textbf{complex-valued score-matching} loss function that estimates the score function of the intermediate error-contaminated data based solely on the observed error-contaminated data. Moreover, to address the implementation challenges associated with complex numbers, we propose a novel RKHS-based score-matching method. When implemented with a Gaussian kernel, all the complex numbers involved in the complex-valued score-matching loss can be integrated out, enabling the estimated score function to be obtained without directly operating on complex numbers. Furthermore, the resulting score estimator is available in closed form, offering significant computational convenience. Figure \ref{fig: illustration} provides an illustration of the proposed method and a comparison with the standard diffusion generation procedure.
	
	\begin{figure}[t!]
		\centering
		\vspace{-0.1in}
		\subcaptionbox{Standard diffusion model to generate synthetic data}{\includegraphics[width=0.95\linewidth]{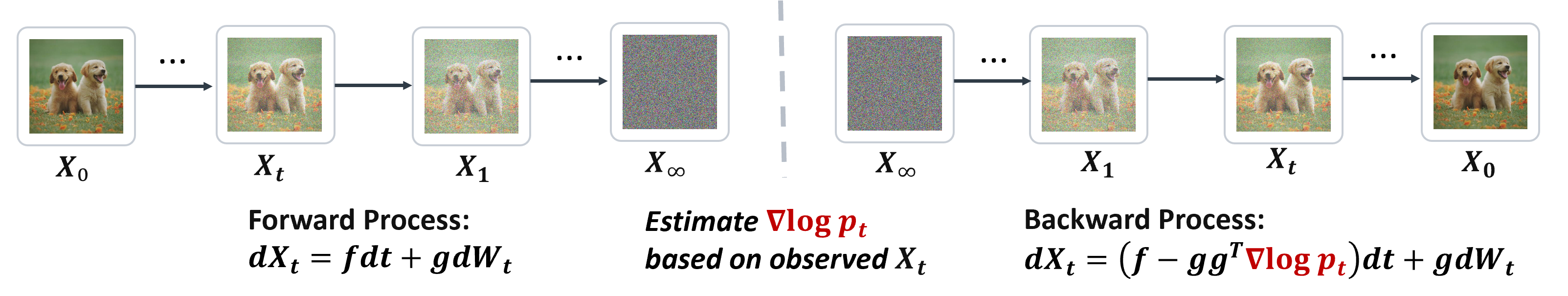}
			\vspace{-10pt}}\\
		\vspace{10pt}
		\subcaptionbox{Proposed diffusion-based method to denoise error-contaminated data}{\includegraphics[width=0.95\linewidth]{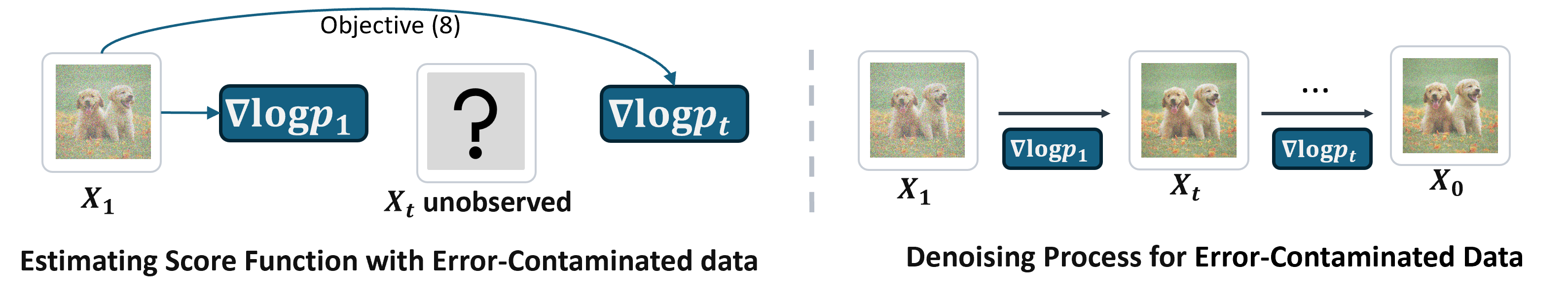}
			\vspace{-10pt}}
		\caption{Comparison between the proposed method and standard diffusion generation procedure.}
		\label{fig: illustration}
		\vspace{-0.2in}
	\end{figure}
	
	\par
	When using the denoised data to conduct downstream analysis, the accuracy of the downstream analysis depends on the gap between the distributions of the denoised data and the error-free data. We characterize this gap by the Total Variation (TV) divergence between the two distributions under proper regularity conditions. 
Concretely, with $n$ observed data, we bound the TV distance by the combined discretization, estimation, and approximation errors in Corollary \ref{cor:convergence rate}.
 Numerical results in simulation studies show the promising empirical performance of our method. An application to the digital health trial data collected by the Juvenile Diabetes Research Foundation (JDRF) Continuous Glucose Monitoring (CGM) Study Group \citep{doi:10.1056/NEJMoa0805017,juvenile2009effect} illustrates the practical value in clinical applications.
	
	% \subsection*{\bf Paper Contributions}
	\vspace{-0.1in}
	
	The contributions of this paper are outlined below:
	
	\vspace{-0.1in}
	\begin{enumerate} \item \textbf{Flexible and Scalable Diffusion-Based Method:}
		We propose a flexible and scalable diffusion-based denoising framework that generates samples from an estimated latent distribution. Then, statistical and machine learning methods can be applied to them for downstream tasks, e.g., predictive or unsupervised learning tasks. Our method does not assume that the density of the error-free data follows a given parametric form and is applicable to large-scale multivariate datasets in a non-parametric manner. 
		
		\vspace{-0.1in}
		\item \textbf{Novel Non-parametric RKHS Framework:}  
		To achieve tractable and efficient estimation with favorable convergence properties, we propose a non-parametric diffusion framework in an RKHS setting with a Gaussian kernel, based on a complex score-matching loss. This approach enables us to estimate the score function of intermediate error-contaminated data using only the observed fully error-contaminated data. 
		
		\item \textbf{Provable Convergence Rate}: 
		Under the regularity and growth conditions of Corollary \ref{cor:convergence rate}, the TV distance between the denoised and error-free distributions is \newline$\cO_P\{n^{-\rho/[4(\rho+\gamma)]}(\log n)^{\rho/(\rho+\gamma)}\}$ for fixed dimension. The exponent approaches $1/4$ when the approximation-decay exponent $\rho$ is large relative to the moment-growth exponent $\gamma$. This polynomial rate is faster than the rates of classical kernel-smoothing-based deconvolution distribution estimators \citep{fan1991optimal, butucea2008sharp}. Such an acceleration relies on the additional approximation and moment-growth conditions in Corollary \ref{cor:convergence rate}, which are not included in analyses of traditional deconvolution methods \citep{fan1991optimal}. 
		
		% Please verify that this rate is correct for your intended analysis.

		\item \textbf{Application to Real Data in Digital Health:}  
		We demonstrate the flexibility and applicability of the proposed framework through a real-world clinical problem in the digital health domain, specifically examining the impact of measurement error on decision-making for CGM. Our sensitivity analysis provides a novel opportunity to reevaluate diabetes clinical trials by incorporating the measurement error of continuous glucose monitors and drawing valid conclusions about the effects of new interventions and drugs.
	\end{enumerate}
	
	% \subsection*{\bf Outline}
	The rest of this paper is organized as follows. Section \ref{sec:notations} provides some preliminaries. In Section \ref{sec:DEP}, we propose the RKHS-based diffusion model that generates denoised samples using error-contaminated data. Our theoretical analysis is mainly in Section \ref{sec:theoretical analysis}, where we characterize the sampling error measured by total variation distance between the generated denoised data and the error-free ones. Section \ref{sec:simulation} presents simulation results. Section \ref{sec:real-data} presents applications to continuous glucose monitoring and the association between lipid measurements and coronary heart disease. Section \ref{sec:discussion} concludes the paper. The Supplementary Material contains the proofs, additional background, related work, assumptions, and numerical experiments.

	\section{Notation and Preliminaries on Diffusion Models}\label{sec:notations}
	\paragraph{Notation.} We use $\mE$ to denote expectation, with a subscript indicating the variables with respect to which the expectation is taken when needed. The subscript is omitted when the expectation is taken over all random variables. Constants related to the dimension $d$ of the data and the data-generating process may be hidden in the order notations $\cO$ and $\cO_{P}$, while we preserve important constants that can be critically relevant to the finite-sample performance. %Specifically, for any random vectors $\bV_{1}$, $\bV_{2}$ and function $f$, $\mE_{\bV_{1}}\{f(\bV_{1}, \bV_{2})\} = \mE\{f(\bV_{1}, \bv)\}\big|_{\bv = \bV_{2}}$. 
	For a random vector $\bX$, $P_{\bX}$ is its probability distribution. 
	%An overline denotes componentwise complex conjugation, $\top$ denotes transpose without conjugation, and $\bz^{*}=\overline{\bz}^{\top}$ denotes conjugate transpose, with $\|\bz\|^{2}=\bz^{*}\bz$. 
	We write $\bI_d$ for the $d\times d$ identity matrix and omit the dimension when it is clear.
	%The $L^{2}(P_{\bX})$ is the $L^{2}$-space equipped with norm $\| f\|_{P_{\bX}} \mE[\| f(\bX), g(\bX)\|]$ for $f\in L^{2}(P_{\bX})$. 
	Let $\cH_{\cK}$ denote an RKHS induced by a positive semi-definite kernel $\cK(\cdot, \cdot)$ defined on the product space $\bbC^{d}$ of the complex plane $\bbC$. Throughout this paper, we use the Gaussian kernel as the kernel function. For any $f\in \cH_{\cK}$ and $\bx\in \bbC^{d}$, the kernel satisfies the reproducing property $\langle\cK(\cdot, \bx), f\rangle_{\cH_{\cK}} = f(\bx)$ where $\langle\cdot, \cdot\rangle_{\cH_{\cK}}$ is conjugate-linear in its first argument and linear in its second. We also define $\cH_{\cK}^{d} = \cH_{\cK}\times \cdots\times \cH_{\cK}$ to be the RKHS of $d$-component functions with inner product defined as the summation of component-wise inner products. Let $L_{\bX_t}^{2}=L^{2}(P_{\bX_t})$ denote the space of square-integrable functions, with $\|f\|_{L_{\bX_t}^{2}}^{2}=\mE_{\bX_t}[|f(\bX_t)|^{2}]<\infty$. For vector-valued functions, we use $(L^{2}_{\bX_t})^{d}$, with squared norm equal to the sum of componentwise squared norms; $\|\cdot\|_{L^2_{\bX_t}}$ denotes the vector-valued norm.  
	%Suppose we are interested in a functional $\psi(P_{\bX_{0}})$ of the distribution $P_{\bX_{0}}$ of the error-free data. 
	\par
	In this paper, the goal is to sample denoised data from the target distribution when only error-contaminated data are observed. Let $\{\bX_{0}^{(i)}\}_{i=1}^{n}$ be independent and identically distributed samples from the target distribution $P_{\bX_{0}}$ of the $d$-dimensional error-free data. We only observe the error-contaminated dataset $\{\bX_{1}^{(i)}\}_{i=1}^{n}$, with $\bX_{1}^{(i)} = \bX_{0}^{(i)} + \bepsilon^{(i)}$ for $i=1,\cdots, n$, where $\bX_{0}^{(i)}\perp\bepsilon^{(i)}$, and $\{\bepsilon^{(i)}\}_{i=1}^{n}$ are independent random noises following $\cN(\bzero, \bbSig)$, independent of the latent sample. For a generic observation, we write $\beps=\bX_1-\bX_0$. To focus on the main idea, we mainly consider the case where the measurement error is normal with zero mean and known variance. For normal measurement error with unknown variance, we provide a discussion in Supplementary Material \ref{app:unknown-sigma}, where we estimate the noise variance from error-contaminated data and then plug the estimate into our method. 
	% Alternatively, researchers may specify a range of $\bbSig$ and view $\bbSig$ as a sensitivity parameter and use the proposed method to evaluate the sensitivity of the analysis result to measurement errors.
	
	Moreover, we define the level-controlled error-contaminated data $\bX_{t} = \bX_{0} + \bbSig^{\frac{1}{2}}\bW_{t}$, and denote its distribution and density by $P_{\bX_{t}}$ and $p_{t}$, respectively, where $\bW_{t}\in\bbR^{d}$ is a standard Brownian motion independent of $\bX_0$. Note that $\bX_{t}$ need not have a direct real-world interpretation, and we do not require an actual Brownian motion connecting the error-free and error-contaminated observations. The level-controlled error-contaminated data are introduced only as a conceptual device for understanding the diffusion model.
	\par
	\vspace{-0.3in}
	\paragraph{Preliminaries on Diffusion Models.} The diffusion model was first proposed in \citep{sohl2015deep} as a data generation technique in the AI community. Theoretically, the technique is related to multi-step variational auto-encoders \citep{kingma2013auto} and Langevin dynamics sampling \citep{hyvarinen2005estimation,vincent2011connection,chewi2023log}. Recently, with the rapid development of deep learning, \citet{ho2020denoising} and \citet{song2020score} revisited this technique and successfully applied it to high-resolution image generation tasks. In this paper, the diffusion model framework is built upon continuous-time SDEs \citep{oksendal2013stochastic} as in \citet{song2020score}. That is, the diffusion process $\bX_{t}\in\bbR^{d}$ ($t\in[0, 1]$) satisfies  
	\begin{equation}\label{eq:sde}
		d\bX_{t} = \textbf{f}(\bX_{t}, t)dt + \textbf{g}(t)d\bW_{t}, 
	\end{equation}
	starting from $\bX_{0}\sim P_{\bX_{0}}$ with coefficient functions $\textbf{f}:\bbR^{d}\times[0,1]\to\bbR^{d}$ and $\textbf{g}:[0,1]\to\bbR^{d\times d}$, and Brownian motion $\bW_{t}$. Under proper regularity conditions \citep{anderson1982reverse}, the reverse-time process $\tX_{t}$ defined by
	\begin{equation}\label{eq:reverse sde}
		d\tX_{t} = -\left(\textbf{f}(\tX_{t}, 1-t) - \textbf{g}(1-t)\textbf{g}(1-t)^{\top}\nabla\log{p_{1 - t}(\tX_{t})}\right)dt + \textbf{g}(1-t)d\widetilde{\bW}_{t},
	\end{equation}
	driven by Brownian motion $\widetilde{\bW}_t$ in reverse time and started from $\tX_{0}\sim P_{\bX_{1}}$, satisfies $P_{\tX_{t}} = P_{\bX_{1 - t}}$ for all $t$.  
	
	The goal of a diffusion model is to generate data from $P_{\bX_{0}}$ given observations $\{\bX_{0}^{(i)}\}_{i=1}^{n}$ from it. By properly choosing $\textbf{f}$ and $\textbf{g}$ in \eqref{eq:sde}, the terminal distribution $P_{\bX_{1}}$ is usually easy to sample from, e.g., it is often normal \citep{song2020score,liu2022flow}. Then, starting from an easily sampled $\tX_{0}\sim P_{\bX_{1}}$, solving \eqref{eq:reverse sde} can produce the observation $\tX_{1}$, which follows the target distribution $P_{\bX_{0}}$. To do so, one needs to estimate the ``score function'' $\nabla\log{p_{t}}$ of $\bX_{t}$ by minimizing the following score-matching objective:  
	\begin{equation*}
		\min_{\btheta_{t}} \left\{\mE\left[\left\|\bs_{t}(\bX_{t}; \btheta_{t}) - \nabla\log{p_{t}(\bX_{t})}\right\|^{2}\right]\right\}; \qquad t\in[0, 1],
	\end{equation*}   
	where $\bs_{t}(\cdot; \btheta_{t})$ is the specified model for the score function and is real-valued on $\bbR^d$. We use its analytic extension when evaluating it at complex arguments. 
	However, the objective function above is infeasible, because the target $\nabla\log{p_{t}(\bx)}$ is unknown. To address this, \citet{hyvarinen2005estimation} and \citet{vincent2011connection} propose the following feasible equivalent objective function: %\footnote{In some certain cases, e.g., $\textbf{f} = \textbf{0}$ in \eqref{eq:sde} \citep{ho2020denoising,song2020score}, the objective can be substituted with a simpler one $\min_{\btheta_{t}}\mE_{\bX_{0}}[\mE_{\bX_{t}}[\|\bs_{t}(\bX_{t}; \btheta_{t}) - \nabla\log{p_{t\mid 0}(\bX_{t}\mid \bX_{0})}\|^{2}]]$, whereas the conditional score function $\nabla\log{p_{t\mid 0}(\bX_{t}\mid \bX_{0})}$ has explicit formulation.}  
	\begin{equation}\label{eq:objective}
		\min_{\btheta_{t}} \left\{\mE\left[\left\|\bs_{t}(\bX_{t}; \btheta_{t})\right\|^{2}\right] + 2\tr\left\{\mE\left[\nabla\bs_{t}(\bX_{t}; \btheta_{t})\right]\right\}\right\}; \qquad t\in[0, 1]. 
	\end{equation}
	In a data-generation task, observations from $P_{\bX_{t}}$ can be obtained by implementing the forward SDE \eqref{eq:sde} starting from data $\{\bX_{0}^{(i)}\}_{i=1}^{n}$. Unfortunately, this is unsuitable for the measurement error problem studied in this paper, because the error-contaminated data $\bX_{1}$, rather than the error-free $\bX_{0}$, are observed. 
    % We present our solution to this problem later in Section \ref{sec:DEP}.   
	% \blue{
		% \begin{remark}\label{eq: }
			%     The SDE described here is the variance-exploding SDE \citep{song2020score} in the context of diffusion models. In diffusion models, alternative SDEs such as the variance-preserving SDE can be adopted for generative purposes \citep{song2020score}. Our framework adopts the variance-exploding SDE because it fits naturally into the error-contaminated data denoising. Specifically, the error-contaminated data and the error-free data can be described as two states of the process associated with the variance-exploding SDE. This is not true for the variance-preserving SDE because the error-contaminated data has larger variance than the error-free data under additive non-differential measurement errors.
			% \end{remark}
		% }

	%\paragraph{Preliminaries to Reproducing Kernel Hilbert Spaces.}
	% It worth noting that $\bX_{0}\sim P_{\bX_{0}}$ and $\bX_{1}\sim P_{\bX_{1}} = P_{\bX_{0}}*Q$, where $*$ represents the convolution operator. Assume throughout the paper that $\bX_{t}$ is independent of the data for $t \in [0, 1]$.
	\vspace{-0.07in}
	\section{The Estimation and Denoising Procedure}\label{sec:DEP}
	\vspace{-0.07in}
	
	We use the diffusion model \citep{song2020score,ho2020denoising} to generate denoised data that closely mirror the distribution of error-free data. Diffusion models have achieved strong empirical performance across a wide range of tasks, especially in generative modeling. As discussed in Section \ref{sec:introduction}, our idea to generate approximately error-free data is to implement the reverse SDE \eqref{eq:reverse sde} starting from the observed error-contaminated data $\{\bX_{1}^{(i)}\}_{i=1}^{n}$. To do so, we need to estimate the ``score function'' \citep{vincent2011connection} by solving problem \eqref{eq:objective}. Unfortunately, solving \eqref{eq:objective} requires data from $P_{\bX_{t}}$, which are unobserved in practice. Next, we present our solution to this problem. 
	
	\vspace{-0.07in}
	\subsection{Estimating the Score Function}\label{subsec:Estimating Score Function}
	\vspace{-0.07in}
	
	By construction, the level-controlled error-contaminated data $\bX_{t} = \bX_{0} + \bbSig^{\frac{1}{2}}\bW_{t}$ satisfies  
	\vspace{-0.1in}
	\begin{equation}\label{eq:sode}
		d\bX_{t} = \bbSig^{\frac{1}{2}}d\bW_{t}, \qquad t\in[0, 1],
	\end{equation}
	which is a special case of \eqref{eq:sde} with $\textbf{f} = \textbf{0}$ and $\textbf{g} = \bbSig^{\frac{1}{2}}$, i.e., the variance exploding SDE in \citep{song2020score}.Under proper regularity conditions (Supplementary Material \ref{app:regularity conditions}), the process $\tX_{t}$ satisfying the \emph{reverse SDE} 
	\begin{equation}\label{eq:reverse time sde}
    % \vspace{-0.4in}
		d\tX_{t} = \bbSig\nabla\log{p_{1 - t}(\tX_{t})}dt + \bbSig^{\frac{1}{2}} d\widetilde{\bW}_{t},	
	\end{equation}
	% \vspace{-0.1in}
	has the same distribution as $\bX_{1 - t}$ in \eqref{eq:sode} for each $t$. Thus, implementing the reverse SDE \eqref{eq:reverse time sde} from $\tX_{0}\sim P_{\bX_{1}}$ yields $\tX_{1}\sim P_{\bX_{0}}$. See Proposition \ref{pro:density} in Supplementary Material \ref{app:discussion on reverse time sde} or  \cite{anderson1982reverse} for more details.  
	\par
	To implement the reverse SDE \eqref{eq:reverse time sde} in practice, we must estimate the score function $\nabla\log{p_{t}(\cdot)}$ for each $t$. However, the standard objective \eqref{eq:objective} cannot be estimated directly because we do not have data from $P_{\bX_{t}}$. To resolve this, we use the following useful lemma from \cite{stefanski1995simulation}. With this lemma, we show that the target \eqref{eq:objective} can be properly estimated without samples from $P_{\bX_{t}}$. To simplify the presentation, we assume throughout that expectation and series summation are exchangeable for the Taylor series of an analytic function.
	\vspace{-0.1in}
	\begin{lemma}\label{lem: complex debias}
		Let $\bxi_{1}$ and $\bxi_{2}$ be mean-zero independent and identically distributed normal vectors. For any analytic function $f(\cdot)$ on $\bbC^{d}$, it holds that $\mE_{\bxi_{1}, \bxi_{2}}[f(\bx + \bxi_{1} + \mathrm{i}\bxi_{2})] = f(\bx)$, where $\mathrm{i} = \sqrt{-1}$ is the imaginary unit.
	\end{lemma}
	\vspace{-0.1in}
	
	Next, we provide some intuition for Lemma \ref{lem: complex debias}. We consider the univariate case where $\xi_{1}$ and $\xi_{2}$ are standard normal variables. Then, for any positive integer $r$,  $\mE[(\xi_{1} + \mathrm{i}\xi_{2})^{r}] = 0$ and hence $\mE[(x + \xi_{1} + \mathrm{i}\xi_{2})^{r}] = x^{r}$ according to the binomial theorem and straightforward calculation. This suggests that Lemma \ref{lem: complex debias} holds for all polynomials. The generalization of this lemma to general analytic functions is straightforward, since analytic functions can be expressed as Taylor series under regularity conditions. 
	The following corollary based on Lemma \ref{lem: complex debias} paves the way to estimate the score function of $\bX_{t}$ based solely on the error-contaminated data $\bX_{1}$. 
	
	\vspace{-0.1in}
	\begin{corollary}\label{coro:debias}
		Let $\bxi \sim \cN(\bzero, \bbSig)$ independent of $\bX_{0}$ and $\beps$. If $\bs_t(\cdot;\btheta_t)$ is analytic on $\bbC^d$, then we have  
\begingroup\small
\begin{equation*}
			\begin{aligned}
				& \mE_{\bX_{t}}\left[\left\|\bs_{t}(\bX_{t}; \btheta_{t})\right\|^{2}\right] + 2\tr\left\{\mE_{\bX_{t}}\left[\nabla\bs_{t}(\bX_{t}; \btheta_{t})\right]\right\} = \\
				& \mE_{\bX_{1}, \bxi}\left[\bs_t(\bX_{1}+\mathrm{i}\sqrt{1-t}\bxi;\btheta_t)^\top\bs_t(\bX_{1}+\mathrm{i}\sqrt{1-t}\bxi;\btheta_t)\right] + 2\tr\left\{\mE_{\bX_{1}, \bxi}\left[\nabla\bs_{t}(\bX_{1} + \mathrm{i}\sqrt{1 - t}\bxi; \btheta_{t})\right]\right\}, 
			\end{aligned}
		\end{equation*}
\endgroup
\vspace{-0.3in}
	\end{corollary}
	We write $\bz^{\T}\bz$ explicitly for $\bz\in\bbC^d$ because it is analytic while the squared norm is not. The proof of this corollary is in Supplementary Material \ref{app:proof in estimating score function}. According to Corollary \ref{coro:debias}, we can estimate objective \eqref{eq:objective} by minimizing: 
\begingroup\small
\begin{equation}\label{eq:new complex objective}
		\mE_{\bX_{1}, \bxi}\left[\bs_t(\bX_{1}+\mathrm{i}\sqrt{1-t}\bxi;\btheta_t)^\top\bs_t(\bX_{1}+\mathrm{i}\sqrt{1-t}\bxi;\btheta_t)\right] + 2\tr\left\{\mE_{\bX_{1}, \bxi}\left[\nabla\bs_{t}(\bX_{1} + \mathrm{i}\sqrt{1 - t}\bxi; \btheta_{t})\right]\right\},
	\end{equation}
\endgroup
	without requiring observations from $P_{\bX_{t}}$. Ideally, with objective \eqref{eq:new complex objective}, we can start from $\bX_{1}^{(i)}$ and substitute the ground-truth score $\nabla\log{p_{t}}(\cdot)$ with its approximation $\bs_{t}(\cdot;\btheta_{t})$ in \eqref{eq:reverse sde} to generate the desired denoised data. 
	
	% Next, we describe this process in detail.   
	
	\begin{remark}
		In this paper, we consider the case where the measurement errors $\{\bepsilon^{(i)}\}_{i=1}^{n}$ follow a Gaussian distribution. This is motivated by the prevalence of Gaussian errors in real data \citep{harchol2023introduction} and especially in empirical Bayes analysis \citep{efron2009empirical}. A natural question is whether the current framework can be extended to more general error distributions. Unfortunately, Lemma \ref{lem: complex debias}, which enables the estimation of the score function based on observed data, relies heavily on the Gaussian assumption. On the other hand, the idea of ``diffusion-based denoising'' can, in principle, be extended to other error distributions by considering flow-based models \citep{liu2022flow,lipman2022flow}. However, to make such procedures feasible, new methods beyond our Corollary \ref{coro:debias} are needed to resolve the problem of estimating the score of the intermediate data distribution.
	\end{remark}
	
	\vspace{-0.2in}
	\subsection{Kernel-Based Score Function}\label{sec:kernel-based function}
	In this section, we consider the modeling strategy for the score function.
	Note that in \eqref{eq:new complex objective}: (i) the model $\bs_{t}(\cdot;\btheta_{t})$ is supposed to be analytic in $\bbC^{d}$, a requirement that is violated by many commonly used nonparametric basis functions such as B-splines; (ii) computing the objective function \eqref{eq:new complex objective} requires the Jacobian of the model and integration with respect to a complex normal vector, which may cause a significant computational burden and numerical instability during minimization. To address these issues, one needs to specify a model $\bs_{t}(\cdot;\btheta_{t})$ that is: (i) analytic; (ii) capable of approximating nonparametric smooth functions; and (iii) permits the explicit computation of the expectation with respect to $\mathrm{i}\sqrt{1 - t}\bxi$. These properties ensure the nonparametric validity of the score estimator and the computational efficiency of the loss-function evaluation at the sample level. 
	To this end, we develop a method that satisfies all the above requirements based on the RKHS of $d$-component functions (see Supplementary Material \ref{app:rkhs} for more details) $\cH_{\cK}^{d}$ associated with the positive semi-definite Gaussian kernel: 
	\begin{equation}\label{eq:Gaussian kernel}
		\cK(\bx_{1}, \bx_{2}) = \exp\left\{-(\bx_{1} - \bar{\bx}_{2})^{\top}\bbH(\bx_{1} - \bar{\bx}_{2})\right\}
	\end{equation}
	defined on $\bbC^{d}\times \bbC^{d}$, where $\bar{\bx}_{2}$ is the conjugate of $\bx_{2}$, and $\bbH$ is a real symmetric positive-definite matrix. 
	
	Specifically, we model the score function as a linear combination of Gaussian kernels. Such a model is flexible because the Gaussian kernel is universal in the sense that its linear combinations can approximate any continuous function to arbitrary accuracy on any compact set \citep{micchelli2006universal}.
To model the score, we assume, for any $b > 0$ and $t\in[0,1]$, there is a real vector-valued function $\bbf_{t,b}$ satisfying 
	\begin{equation*}
		\nabla\log{p_{t}}(\bx) = \mE_{\bX_{t}}[\cK(\bx, \bX_{t})\bbf_{t,b}(\bX_{t})] + \bkappa_{t,b}(\bx),
	\end{equation*}
	with
$\mE_{\bX_t}[\|\bbf_{t,b}(\bX_t)\|^2]\leq b$,
$\sup_{t\in[0,1]}\mE_{\bX_t}[\|\bkappa_{t,b}(\bX_t)\|^2]\leq\kappa(b)$
for some function $\kappa(\cdot)$ such that $\lim_{b\to\infty}\kappa(b) = 0$. That says, $\nabla\log{p_{t}}(\bx)$ can be well approximated by a function of the form $\mE_{\bX_{t}}[\cK(\bx, \bX_{t})\bbf_{t,b}(\bX_{t})]$ with $\kappa(b)$ being a function characterizing the relationship between the approximation error and the norm constraint $b$. Moreover, to make this formulation computable, noting that $\bX_{1} = \bX_{t} + \bbSig^{\frac{1}{2}}(\bW_{1} - \bW_{t})$, $ \bbSig^{\frac{1}{2}}(\bW_{1} - \bW_{t}) \sim \cN(\bzero, (1 - t)\bbSig)$  and $ \bbSig^{\frac{1}{2}}(\bW_{1} - \bW_{t}) \perp \bX_{t}$, we further have $\nabla\log{p_{t}}(\bx) = \mE_{\bX_t,\bX_1,\bxi}[\cK(\bx, \bX_{1} + \mathrm{i}\sqrt{1 - t}\bxi)\bbf_{t,b}(\bX_{t})] + \bkappa_{t,b}(\bx)$ by combining Lemma \ref{lem: complex debias} and the discussion above, where the expectation is over the joint forward-process pair $(\bX_t,\bX_1)$ and an independent $\bxi\sim\cN(\bzero,\bbSig)$. Thus, the score function $\nabla\log{p_{t}}$ can be approximated by
	\begin{equation}\label{eq: approximate mean}
		\nabla\log{p_{t}}(\bx)
		\approx\frac{1}{m}\sum_{i = 1}^{m}\bbf_{t,b}(\bX_{t}^{(i)})\mE_{\bxi}[\cK(\bx, \bX_{1}^{(i)} + \mathrm{i}\sqrt{1 - t}\bxi)]
		= \frac{1}{m}\sum_{i = 1}^{m}\bbf_{t,b}(\bX_{t}^{(i)})\cK_{t}(\bx, \bX_{1}^{(i)}),
	\end{equation}
	whose basis functions depend only on the first $m<n$ observed points, with $m\ll n$ in practice. The unobserved coefficients $\bbf_{t,b}(\bX_t^{(i)})/m$ motivate the free parameters fitted below. Here 
	\begin{equation}\label{eq:basis s}
		\cK_{t}(\bx_{1}, \bx_{2}) = \mE_{\bxi}[\cK(\bx_{1}, \bx_{2} + \mathrm{i}\sqrt{1 - t}\bxi)] = \sqrt{|\bbOmega_{t}|^{-1}|\bbOmega|}
		\exp\left\{ - (\bx_{1} - \bx_{2})^{\top}\bbH_{t}(\bx_{1} - \bx_{2})\right\}
	\end{equation}
	for any real vectors $\bx_{1}$ and $\bx_{2}$, $\bbOmega = (2\bbSig)^{-1}$, $\bbOmega_{t} = \bbOmega - (1 - t)\bbH$ and $\bbH_{t} = \{\bbH^{-1} - (1 - t)\bbOmega^{-1}\}^{-1}$. This formula requires $\bbOmega_t\succ0$. For complex arguments, $\cK_t$ is extended by replacing $\bx_2$ with $\overline{\bx_2}$ in the displayed Gaussian formula, so it is analytic in its first argument. The derivation of \eqref{eq:basis s} is in Supplementary Material \ref{app:proofs in kernel based function}. We use $m$ samples instead of all $n$ samples to estimate $\nabla\log{p_{t}}(\bx)$, as in standard kernel ridge regression \citep{zhang2015divide,zhou2020nonparametric}, because the resulting $nd$ basis functions would make the computation cumbersome in practice. 
	%To mitigate this, we propose to construct the $\bs_{t}(\bX_{t};\btheta_{t})$ with the basis functions induced by $m$ a subsample to approximate the true score function, which accelerates the computation process. 
	This trick can be viewed as an extension of the random sketch technique adopted in kernel ridge regression which improves computational efficiency while maintaining the optimal convergence rate \citep{yang2017randomized}. The choice of $m$ will be discussed in Section \ref{sec:generalization analysis}.

	\vspace{-0.1in}\begin{remark}\label{remark:representation}
For fixed $t$ and $b$, the approximating function $g=\nabla\log p_t-\bkappa_{t,b}$ has the integral representation $g=\cT_{\cK, t}\bbf_{t,b}$ where the operator $\cT_{\cK,t}\bbf=\mE_{\bX_t}[\cK(\bx,\bX_t)\bbf(\bX_t)]$ for any squared integrable $\bbf$. For simplicity, consider $d=1$. By Mercer's theorem, the positive-definite kernel operator $\cT_{\cK, t}$ admits eigenpairs $\{(\lambda_j,\phi_j)\}_{j\ge1}$. A function $g=\sum_j\alpha_j\phi_j$ belongs to $\cH_{\cK}$ if and only if $\sum_j\alpha_j^2/\lambda_j<\infty$, which implies $\cH_{\cK}=\mathrm{Range}(\cT_{\cK, t}^{1/2})$. The representation $g=\cT_{\cK, t}f_{t,b}$ with $\|f_{t,b}\|_{L^2_{\bX_t}}^2\leq b$ instead gives $\sum_j\alpha_j^2/\lambda_j^2\leq b$. This restriction applies to the approximating function, while $\kappa(b)$ bounds its squared $L^2_{\bX_t}$ distance from the score, uniformly in $t$. Allowing $b\to\infty$ relaxes the restriction; it does not require the score itself to belong to $\mathrm{Range}(\cT_{\cK, t})$ or to be analytic.
	\end{remark}
	
	Next, we study the approximation error of the kernel-based model. According to \eqref{eq: approximate mean}, we propose to approximate $\nabla\log{p_{t}}$ by elements in $\cH_{\cK_{t}, m}^{d} = \cH_{\cK_{t}, m}\times \cdots\times \cH_{\cK_{t}, m}$ where $\cH_{\cK_{t}, m}\subset \cH_{\cK_t,\bbR}$ is the real linear span of $\{\cK_{t}(\cdot, \bX_{1}^{(i)})\}_{i=1}^{m}$. We use $\btheta_{t, l} = (\theta_{t,l}^{(1)},\dots, \theta_{t, l}^{(m)})^{\top}\in\bbR^{m}$ and $\btheta_{t} = (\btheta_{t, 1}, \dots, \btheta_{t, d})\in\bbR^{m\times d}$ to represent the parameters in our kernel-based model $\bs_{t}(\bx; \btheta_{t})$, which is defined as:
	\begin{equation}\label{eq: kernel model}
		\bs_{t}(\bx; \btheta_{t}) = (s_{t1}(\bx; \btheta_{t}), \dots, s_{td}(\bx; \btheta_{t}))^{\top}\text{;}\enspace s_{tl}(\bx; \btheta_{t}) = \sum_{i = 1}^{m}\theta_{t, l}^{(i)}\cK_{t}(\bx, \bX_{1}^{(i)}),\qquad l = 1,\dots, d
	\end{equation}
	
	\vspace{-0.1in}
	\begin{remark}
		A natural alternative strategy is to model $\bs_{t}(\cdot;\btheta_{t})$ as a U-net neural network \citep{ronneberger2015u}, as is common practice in diffusion model training \citep{huang2023scalelong}. However, commonly used activation functions need not admit entire extensions to the complex plane: ReLU is not real analytic, whereas sigmoid is real analytic on the real line but its complex extension has poles \citep{lecun2015deep}. Neural networks using these activations therefore generally fail to meet the analyticity requirement of Corollary \ref{coro:debias}. 
	\end{remark}

	\vspace{-0.1in}\begin{proposition}\label{prop:approximation error}
Under Assumptions \ref{ass:lip continuity} and \ref{ass:subexp} in Supplementary Material, fix $b>0$ and $t\in[0,1]$, define
\begin{equation*}
    \delta_{0m}(\eta;b)=2\left(\frac{H_0(b)}{m}+\frac{\sigma_0(b)}{\sqrt m}\right)\log\left(\frac{2}{\eta}\right),
\end{equation*}
where $\sigma_0(b)$ and $H_0(b)$ are the constants in Assumption \ref{ass:subexp}. Then,
\begin{equation*}
\begin{aligned}
    &\inf_{\bs_t(\cdot;\btheta_t)\in\cH_{\cK_t,m}^{d}}
    \mE_{\bX_t}\left[\left\|\bs_t(\bX_t;\btheta_t)-\nabla\log p_t(\bX_t)\right\|^2\right]\leq 2\delta_{0m}(\eta;b)^2+\frac{2\sigma_0(b)^2}{m}+2\kappa(b)
\end{aligned}
\end{equation*}
holds with probability at least $1-\eta$ ($0<\eta<1$), where the randomness comes from $\cH_{\cK_t,m}^{d}$.
	\end{proposition}
For fixed $b$ and confidence level, this proposition bounds the approximation error by $\cO(\sigma_0(b)^2/m+H_0(b)^2/m^2+\kappa(b))$. The residual term $\kappa(b)$ vanishes as the approximation budget tends to infinity, while the concentration constants may grow with $b$.
	Based on the formulation of $\bs_{t}(\cdot; \btheta_{t})$, the following proposition provides an explicit unbiased estimator of the desired objective \eqref{eq:new complex objective} to be minimized.
	\vspace{-0.2in}\begin{proposition}\label{pro:explicit formulation}
		Condition on the first $m$ observations used as kernel centers, fix $\btheta_t$, and choose $\bbH$ independently of the remaining $n-m$ observations. For $\bs_t$ in \eqref{eq: kernel model}, provided the matrices $\bbOmega_t$, $\bbOmega_t^{(1)}$, and $\bbOmega_t^{(2)}$ below are positive definite, an unbiased estimator of \eqref{eq:new complex objective} is 
		\begin{equation}\label{eq:quadratic programming}
			\begin{aligned}
				\tr\left\{\btheta_{t}^{\top}\BK_{t}^{(1)} + \btheta_{t}^{\top}\BK_{t}^{(2)}\btheta_{t}\right\},
			\end{aligned}
		\end{equation}
		where $\BK_{t}^{(1)}$ is an $m\times d$ matrix whose $(i, l)$-th element is $2(n - m)^{-1}\sum_{k = m + 1}^{n} \cK_{t}^{(1, l)}(\bX_{1}^{(k)}; \bX_{1}^{(i)})$, $\BK_{t}^{(2)}$ is an $m\times m$ matrix whose $(i,j)$-th element is $(n - m)^{-1}\sum_{k = m + 1}^{n}\cK_{t}^{(2)}(\bX_{1}^{(k)}; \bX_{1}^{(i)}, \bX_{1}^{(j)})$,  
		\begin{equation*}
			\begin{aligned}
				\cK_{t}^{(1, l)}(\bx; \bx_{1}) 
				& = -2\frac{|\bbOmega|}{\sqrt{|\bbOmega_{t}||\bbOmega_{t}^{(1)}|}}\be_{l}^{\top}\bbH_{t}^{(1)}(\bx - \bx_{1})\exp\left\{ - (\bx - \bx_{1})^{\top}\bbH_{t}^{(1)}(\bx - \bx_{1})\right\},
			\end{aligned}
		\end{equation*}
		and 
		\begin{equation*}
			\begin{aligned}
				\cK_{t}^{(2)}(\bx; \bx_{1}, \bx_{2}) 
				& = \frac{|\bbOmega|^{\frac{3}{2}}}{|\bbOmega_{t}||\bbOmega_{t}^{(2)}|^{\frac{1}{2}}}\exp\left\{ - (\bx_{1} + \bx_{2} - 2\bx)^{\top}\bbH_{t}^{(2)}(\bx_{1} + \bx_{2} - 2\bx)\right\}\\
				&\quad \times 
				\exp\left\{ - (\bx_{1} - \bx)^{\top}\bbH_{t}(\bx_{1} - \bx)\right\} \times \exp\left\{ - (\bx_{2} - \bx)^{\top}\bbH_{t}(\bx_{2} - \bx)\right\},
			\end{aligned}
		\end{equation*}
		for any $\bx,\bx_1,\bx_2\in\bbR^d$. Here $\be_{l}$ is the $l$-th basis vector in the $d$-dimensional space, $\bbOmega_{t}^{(1)} = \bbOmega - (1 - t)\bbH_{t}$,
		$\bbH_{t}^{(1)} = (1 - t)\bbH_{t}\bbOmega_{t}^{(1)-1}\bbH_{t} + \bbH_{t}$, $\bbOmega_{t}^{(2)} = \bbOmega - 2(1 - t)\bbH_{t}$ and $\bbH_{t}^{(2)} = (1 - t)\bbH_{t}\bbOmega_{t}^{(2)-1}\bbH_{t}$. 
	\end{proposition}
	
	The proof of this proposition is in Supplementary Material \ref{app:proofs in Section dep}. Notably, the expected objective function \eqref{eq:new complex objective} involves expectations with respect to $\bX_{1}$ and $\bxi$, whereas due to the formulation of our $\bs_{t}(\cdot; \btheta_{t})$, the expectation with respect to $\bxi$ can be explicitly computed without requiring a Monte Carlo approximation. See the proof of Proposition \ref{pro:explicit formulation} for more details. Moreover, as in the literature \citep{de2005learning,smale2003estimating,wainwright2019high}, we add a regularization term $\lambda\|\bs_{t}(\cdot; \btheta_{t})\|_{\cH_{\cK}^{d}}^{2}$ to our training objective \eqref{eq:quadratic programming}, where $\lambda > 0$ is a regularization parameter, and $\|\bs_{t}(\cdot; \btheta_{t})\|_{\cH_{\cK}^{d}}$ is the norm of $\bs_{t}(\cdot; \btheta_{t})$ in the RKHS of $d$-component functions. As in the existing literature, such a regularization term is critical to avoid overfitting the kernel-based model \citep{smale2003estimating}. Using the $\lambda$ discussed in Theorem \ref{thm:gen error}, we obtain the \emph{final objective function} for each $t$,
	\begin{equation*}
		\min_{\btheta_{t}}\tr\left\{\btheta_{t}^{\top}\BK_{t}^{(1)} + \btheta_{t}^{\top}\BK_{t}^{(2)}\btheta_{t} + \lambda \btheta_{t}^{\top}\BK_{t}^{(0)}\btheta_{t}\right\},
	\end{equation*}
	where $\BK_{t}^{(0)}\in \bbR^{m\times m}$ with $(i,j)$-th element as the inner product between $\cK_{t}(\cdot, \bX_{1}^{(i)})$ and $\cK_{t}(\cdot, \bX_{1}^{(j)})$ in $\cH_{\cK}$,
	\begingroup\small
\begin{equation*}
		\begin{aligned}
			\cK_{t}^{(0)}\left(\bX_{1}^{(i)}, \bX_{1}^{(j)}\right)&\eqqcolon\left\langle\cK_{t}(\cdot, \bX_{1}^{(i)}), \cK_{t}(\cdot, \bX_{1}^{(j)})\right\rangle_{\cH_{\cK}}
			=
			\mE_{\bxi, \bxi^{\prime}}\left[\cK\left(\bX_{1}^{(i)} + \mathrm{i}\sqrt{1 - t}\bxi, \bX_{1}^{(j)} + \mathrm{i}\sqrt{1 - t}\bxi^{\prime}\right)\right]\\
			& = \sqrt{|\bbOmega_{t}^{(0)}|^{-1}|\bbOmega^{(0)}|}
			\exp\left\{ - \left(\bX_{1}^{(i)} - \bX_{1}^{(j)}\right)^{\top}\bbH_{t}^{(0)}\left(\bX_{1}^{(i)} - \bX_{1}^{(j)}\right)\right\}.
		\end{aligned}
	\end{equation*}
\endgroup
	Here $\bbOmega^{(0)} = (4\bbSig)^{-1}$, $\bbOmega_{t}^{(0)} = \bbOmega^{(0)} - (1 - t)\bbH$, $\bbH_{t}^{(0)} = \left\{\bbH^{-1} - (1 - t)\bbOmega^{(0)-1}\right\}^{-1}$, and $\bxi$, $\bxi^{\prime}$ are independent random vectors that follow $\cN(\bzero, \bbSig)$. Notably, one should choose $\bbH$ as a function of the variance $\bbSig$ to ensure that $\bbOmega_t$, $\bbOmega_t^{(0)}$, $\bbOmega_t^{(1)}$ and $\bbOmega_t^{(2)}$ are all positive definite, which imposes restrictions on $\bbH$. A rule-of-thumb choice for $\bbH$ that fulfills the requirement is $\bbH = (8\bbSig + \epsilon\bI)^{-1}$ for some small $\epsilon>0$.
	\par
	We compute the coefficient matrix using  
	\begin{equation}\label{eq:closed form solution}
		\htheta_{t} = - \frac{1}{2}(\BK_{t}^{(2)} + \lambda \BK_{t}^{(0)})^{+}\BK_{t}^{(1)},        
	\end{equation}
	where ``$+$'' denotes the Moore--Penrose generalized inverse.
	On the concentration event in Theorem \ref{thm:gen error}, its condition on $\lambda$ makes the regularized quadratic form positive definite on the function span, so \eqref{eq:closed form solution} gives a minimizer. Outside this event, the pseudoinverse formula alone does not ensure minimization. Using this estimated score model, we can generate denoised data by solving the reverse SDE \eqref{eq:reverse time sde}, as elaborated below. 
	
	\vspace{-0.2in}
	\begin{remark}\label{remark:hermite}
		In the proposed method, we estimate the score function in \eqref{eq:objective} by introducing imaginary noise in Corollary \ref{coro:debias}. Building upon this, we model the score function using functions within a Reproducing Kernel Hilbert Space (RKHS). Generally speaking, by Tweedie's formula \citep{efron2011tweedie}, the central intuition is to represent $\bs_{t}(\bX_{t}; \btheta_{t})$ by the conditional expectation $\mE[\bg_{t}(\bX_{1})\mid \bX_{t}]$ for some $\bg_{t}$. For illustration, we consider the univariate case ($d = 1$) with the standard deviation $\sigma$ of the measurement error. For $0\leq t<1$, using the joint representation $X_1=X_t+\sigma\sqrt{1-t}\xi$ with $\xi$ independent of $X_t$, we have
		\begin{equation*}
			\mE\left[He_{j}\left(\frac{X_{1}}{\sigma\sqrt{1 - t}}\right)\mid X_{t}\right] = \mE_{\xi}\left[He_{j}\left(\frac{X_{t} + \sigma\sqrt{1 - t}\xi}{\sigma\sqrt{1 - t}}\right)\mid X_{t}\right] = \sigma^{-j}(1 - t)^{-\frac{j}{2}}X_{t}^{j},    
		\end{equation*}
		where $\xi\sim \cN(0, 1)$, $He_{j}$ denotes the Hermite polynomial of degree $j$. Leveraging the above formulation and by approximating $\bs_{t}$ with a Hermite polynomial, we can obtain an unbiased estimate for the objective function \eqref{eq:new complex objective}. Notably, this specific approximation does not require the introduction of imaginary noise. However, polynomials suffer from a severe curse of dimensionality issue and are numerically unstable due to the presence of high-order terms \citep{de1978practical}. Moreover, the use of Hermite polynomials would require estimating the high-order moments of $X_{1}$, which is challenging and highly sensitive to potential outliers. These problems motivate us to adopt the RKHS modeling strategy with imaginary noise.
	\end{remark}

	\setlength{\columnsep}{0.35in} % 控制正文和 wrapfigure 的水平距离
	\begin{wrapfigure}{r}{8.0cm}
		\centering
		\vspace{-0.2in}
		\includegraphics[scale=.19]{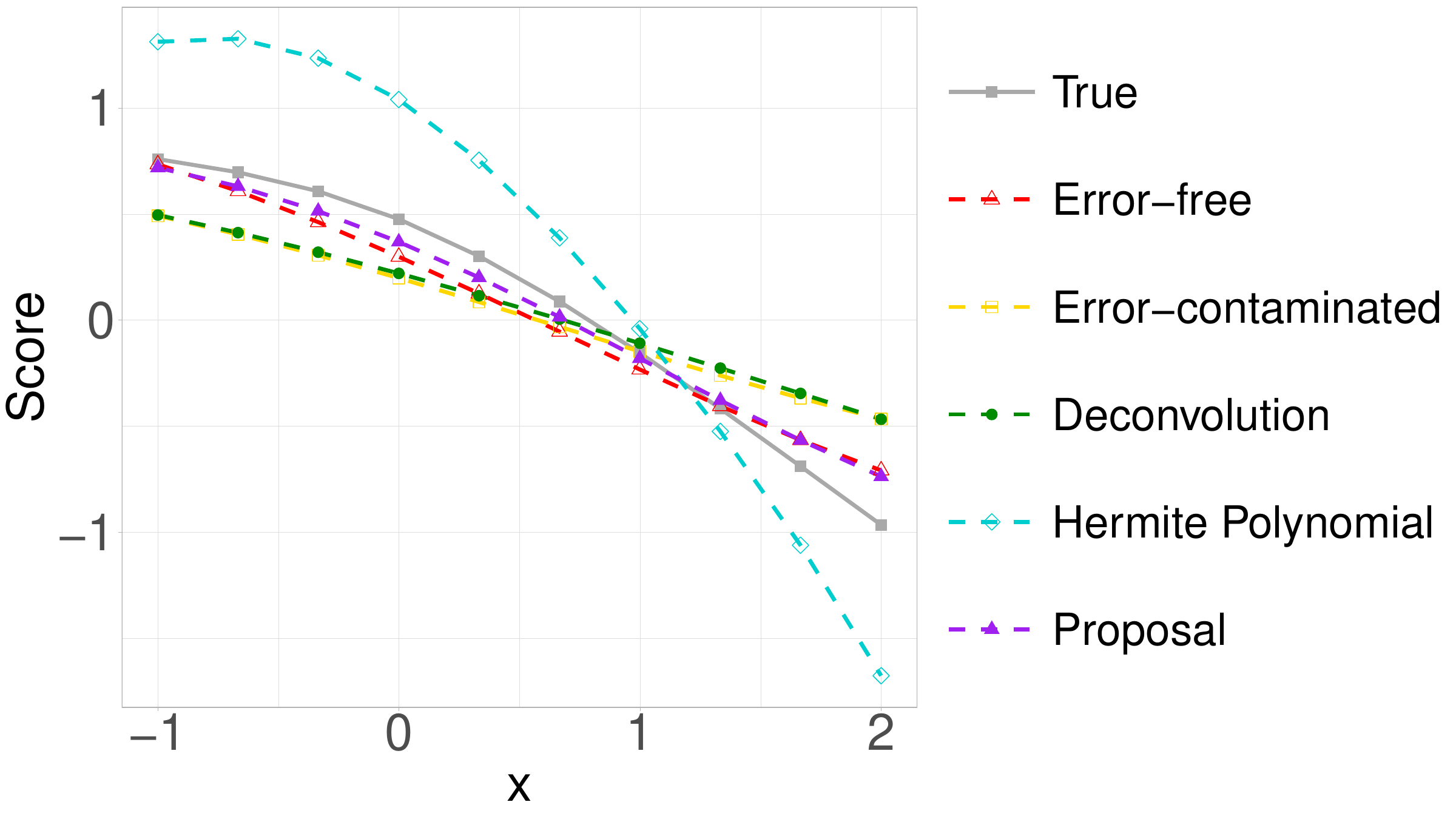}
		\captionsetup{
			margin={0.8cm,0.1cm},
			font=small
		}
		\caption{The score functions fitted using different methods in a simulation example.}
		\vspace{-0.2in}
		\label{fig:score}
	\end{wrapfigure}
	Finally, we use a simulation to provide an intuitive illustration of the effectiveness of the proposed method. Let $n = 1000$, $\epsilon \sim \cN(0, 2)$, $X_{0} \sim 0.5 \cN(1, 1) + 0.5 \cN(0, 2.25)$, and $X_{1} = X_{0} + \epsilon$. We set $H = 1 / (8\times 2)$, $m = \lfloor n^{1/2}\rfloor = 31$, and $\lambda = n^{-1/2}$ 
as suggested by Corollary \ref{cor:convergence rate}, up to logarithmic factors
. Then, we compare the estimated score of our method with estimates obtained under four other baseline methods: 1) and 2) the score functions obtained by RKHS-based score matching \citep{zhou2020nonparametric} under the error-free and error-contaminated data; 3) the deconvolution score function obtained from the density estimated by the R package \texttt{deconv} \citep{wang2011deconvolution} using the classical kernel smoothing-based deconvolution method \citep{fan1991optimal}; and 4) the score estimate leveraging the Hermite polynomial as introduced in Remark \ref{remark:hermite}. As shown in Figure \ref{fig:score}, the proposed method produces a score function close to the ground truth, and is consistent with that obtained from error-free data, which verifies the effectiveness of our method. The mean squared errors over $200$ simulations are $0.11$, $0.41$, $0.48$, $21.47$, and $0.19$ for error-free RKHS, error-contaminated RKHS, the deconvolution-based method, the Hermite polynomial-based method, and the proposed method, respectively. We further discuss applying our method to empirical Bayes estimation in Supplementary Material.
	
	\subsection{Sampling with the Empirical Reverse SDE}\label{sec:sampling with reverse-sde}
	As mentioned in the previous section, the ground truth score function $\nabla\log{p_{t}}$ is approximated by $\bs_{t}(\cdot; \htheta_{t})$. Then, starting from an error-contaminated observation $\tX_{0} \sim P_{\bX_{1}}$ and substituting $\nabla\log{p_{1 - t}(\tX_{t})}$ in \eqref{eq:reverse time sde} with $\bs_{1 - t}(\tX_{t}; \htheta_{1 - t})$, we can generate a denoised observation through an empirical version of the reverse SDE \eqref{eq:reverse time sde}. To do so, we solve \eqref{eq:reverse time sde} with a numerical method. In this paper, we use $\hX_{t}$ to represent the empirical version of $\tX_{t}$ generated from the following standard first-order Euler--Maruyama \citep{platen1999introduction} sampling method  
	\begin{equation}\label{eq:euler solving sde}
		\hX_{(k + 1)/K} = \hX_{k/K} + \frac{1}{K}\bbSig\bs_{1 - k/K}\left(\hX_{k/K}; \htheta_{1 - k/K}\right) + \sqrt{\frac{1}{K}}\bbSig^{\frac{1}{2}}\beps_{k},
	\end{equation} 
	with $\hX_{0} = \bX_{1}$ as the error-contaminated observation, where $K\geq2$ is the number of intervals in the full time grid, the $\beps_k\sim\cN(\bzero,\bI_d)$ are independent across steps and independent of the initialization and fitted scores, and $0\leq k\leq K-2$. Thus, early stopping uses $K-1$ updates of size $1/K$. Here, the discrete sampling process stops at $\hX_{1 - 1/K}$, which approximates $\bX_{1/K}$ as in \citep{chen2023probability}. To simplify the notations, we define $\hX_{1} = \hX_{1 - 1/K}$. The early stop is to account for the possible discontinuity of $\nabla_{\bx}\log{p_{t}}(\bx)$ w.r.t. $t$ when $t\to 0$, which can make sampling near the endpoint problematic. However, since $\bX_{1/K}\approx \bX_{0}$ for sufficiently large $K$, approximating it is enough to approximate the error-free data. In the practical Algorithm \ref{alg:alg1}, we initialize \eqref{eq:euler solving sde} at each of the $n$ observed points $\{\bX_1^{(i)}\}_{i=1}^n$ to generate $n$ denoised observations, using independent Gaussian increments for each trajectory. The constructed independence for the sampling is described in Section \ref{sec:kl divergence error}.

	\vspace{-0.07in}
	\section{Theoretical Analysis}\label{sec:theoretical analysis}
	\vspace{-0.07in}
	Above, we present our method for generating denoised data. However, there is a gap between the distribution of the generated $\hX_{1}$ and the desired $P_{\tX_{1}} = P_{\bX_{0}}$ for three reasons: 1) the estimated score $\bs_{t}(\cdot; \htheta_{t})$ differs from the ground truth score $\nabla\log{p_{t}}$ due to estimation error; 2) the reverse SDE is solved numerically by \eqref{eq:euler solving sde}, which leads to discretization error; 3) the gap between $P_{\bX_{0}}$ and $P_{\bX_{1/K}}$. Next, we theoretically characterize the gap between $P_{\hX_{1}}$ and the desired $P_{\bX_{0}}$.
	\subsection{Discrepancy between the Denoised Data and the Error-free Data}\label{sec:kl divergence error}
	In this section, we establish an upper bound for the sampling error under some mild assumptions. To do so, we first assume the kernel-based score model $\bs_{t}(\cdot; \htheta_{t})$ approximates the ground truth score $\nabla\log{p_{t}}(\bx)$ well (we analyze this approximation error later in Section \ref{sec:generalization analysis}). Our result is presented in the following theorem, which characterizes the gap between the distribution of the generated denoised data $\hX_{1}$ and the desired $\bX_{0}$.  
	\vspace{-0.1in}
	\begin{theorem}\label{thm:sampling error}
		Suppose that $\hX_{0} \sim P_{\bX_{1}}$ is independent of the fitted score functions $\bs_t(\cdot;\htheta_t)$, and that $\hX_{1}$ is generated from the initial point $\hX_{0}$ following the diffusion updates in \eqref{eq:euler solving sde}.
		For $K\geq2$, under Assumptions \ref{ass:regularity}, \ref{ass:lip continuity} and \ref{ass:subexp} in Supplementary Material, assume that for all $t\in \{1/K,\cdots, 1\}$,  
		\begin{equation}\label{eq:small error}
			\mE_{\bX_{t}}\left[\left\|\bs_{t}(\bX_{t}; \htheta_{t}) - \nabla\log{p_{t}}(\bX_{t})\right\|^{2}\right] \leq \delta.
		\end{equation}
		Then	$\TV\left(P_{\bX_{0}}, P_{\hX_{1}}\right) \leq \cO\left(\sqrt{d/K + \delta}\right)$.
	\end{theorem}
	We prove this theorem in Supplementary Material \ref{app:proofs in section sampling with reverse-sde}. 
	Theorem \ref{thm:sampling error} indicates that when the score function is well approximated ($\delta\to 0$), with a large number of sampling steps ($K\to \infty$), the denoised data distribution $P_{\hX_{1}}$ approximates the target $P_{\bX_{0}}$ well. Thus, the generated denoised data may be useful for downstream tasks such as estimation, prediction, and hypothesis testing. Here, the terms related to the step size $1 / K$ and estimation error $\delta$ correspond respectively to the ``discretization error'' and ``approximation error'' discussed above. 
	% Intuitively, they appear in \eqref{eq:error bound} because the error between the generated data \eqref{eq:euler solving sde} and the ground truth \eqref{eq:reverse sde} has two origins: 1) the gap between the estimated score function $\bs_{t}(\cdot;\btheta_{t})$ and the ground truth $\nabla\log{p_{t}}$; 2) the discretization error introduced by the numerical solver \eqref{eq:euler solving sde}.
	\par
%	Notably, by invoking Pinsker's inequality and $\tr(\bbSig)=\cO(d)$, we can obtain the following bound for the total variation distance \citep{wasserman2006all,wainwright2019high}
%	\begin{equation*}
%		\begin{aligned}
%			\TV\left(P_{\bX_{0}}, P_{\hX_{1}}\right)
%			& \leq \sqrt{\frac{1}{2}D_{KL}\left(P_{\bX_{0}} \parallel P_{\hX_{1}}\right)}
%			& = \cO\left(\sqrt{\frac{d}{K} + \delta}\right).
%		\end{aligned}
%	\end{equation*} 
	% where $\TV\left(P_{\bX_{0}}, P_{\hX_{1}}\right)$ is the total variation distance defined as
	% \begin{equation*}
		% \TV\left(P_{\bX_{0}}, P_{\hX_{1}}\right) = \sup_{\text{$\cA$ is measurable}} |P(\bX_{0} \in \cA) - P(\hX_{1} \in \cA)|.
		% \end{equation*}
	
	In Theorem \ref{thm:sampling error}, we assume independence between the initial point $\hX_{0}$ and the entire fitted score model.\footnote{Here $\mE_{\bX_t}$ integrates over a fresh evaluation point, and $P_{\hX_1}$ denotes the law conditional on the fitted score functions. Without independence, \eqref{eq:small error} alone does not justify the theorem.} Kernel centers and any data-dependent bandwidth must also be selected without the generation observations. This can be guaranteed by a leave-one-out training procedure; i.e., for $i = 1,\dots, n$, one trains the $i$-th score model using all data except $\bX_{1}^{(i)}$ and then takes $\bX_{1}^{(i)}$ as the initial point to generate denoised data using the $i$-th score model. However, such training suffers from high computational complexity because it requires training the score models $n$ times. Alternatively, one can split $\{\bX_{1}^{(i)}\}_{i = m + 1}^{n}$ into two disjoint sets $\cD_{\rm train}$ and $\cD_{\rm gen}$ with roughly equal sizes, then retain the first $m$ observations as kernel centers, replace the empirical averages in $\BK_t^{(1)}$ and $\BK_t^{(2)}$ by averages over $\cD_{\rm train}$ with denominator $|\cD_{\rm train}|$, and use $\cD_{\rm gen}$ only for generation. This also ensures independence in Theorem \ref{thm:sampling error}, at the cost of reducing the denoised-data sample size. For each fitted model, Theorem \ref{thm:gen error} justifies \eqref{eq:small error} with its effective training-sample size, including the kernel centers, and with $\bbH$ fixed independently of that training sample as assumed in Section \ref{sec:generalization analysis}. For the split above, this size is $m+|\cD_{\rm train}|$.
	\par
	However, in our experiments, we found that the procedure without leave-one-out or data splitting, i.e., using $\{\bX_{1}^{(i)}\}_{i=1}^{n}$ to train the score model and initialize the SDE, performs reasonably well. We conjecture that the independence condition in Theorem \ref{thm:sampling error} can be relaxed to certain weak dependence conditions. We did not pursue this because the analysis is rather involved. In practice, we recommend implementing the proposed method without leave-one-out or data splitting.
	\par
	% In Theorem \ref{thm:sampling error}, we assume the expected estimation error between our model $\bs_{t}(\cdot; \htheta_{t})$ and the ground truth $\nabla\log{p_{t}}(\cdot)$ is controlled by a small number $\delta$. Next, we establish the high probability bound for the estimation error. 
	
	% However, it worth noting that increasing sampling steps $K$ requires plenty of computation efforts, especially when the model $\bs_{t}_{\btheta}(\cdot, \cdot)$ is a neural network. Thus, practically there is a trade-off between the sampling quality and computational efforts. 
	
	\subsection{Error Analysis for the Score-Function Model}\label{sec:generalization analysis}
	From Theorem \ref{thm:sampling error}, part of the sampling error related to the proposed method originates from the gap between the trained model $\bs_{t}(\cdot; \htheta_{t})$ and the ground truth score function $\nabla\log{p_{t}}$. Next, we characterize this gap. 
	%In this section, our goal is to formally characterize such a gap. Let $\bar{\btheta}_{t} = \mathop{\arg\min}_{\btheta_{t}}\mE_{\bX_{t}}\left[\left\|\bs_{t}(\bX_{t}; \btheta_{t}) - \nabla\log{p_{t}}(\bX_{t})\right\|^{2}\right]$. Note that
	% \[
	% \begin{aligned}
		%     &\mE_{\bX_{t}}\left[\left\|\bs_{t}(\bX_{t}; \htheta_{t}) - \nabla\log{p_{t}}(\bX_{t})\right\|^{2}\right]
		%     \\
		%     & =
		%     \mE_{\bX_{t}}\left[\left\|\bs_{t}(\bX_{t}; \htheta_{t}) - \nabla\log{p_{t}}(\bX_{t})\right\|^{2}\right] - \mE_{\bX_{t}}\left[\left\|\bs_{t}(\bX_{t}; \bar{\btheta}_{t}) - \nabla\log{p_{t}}(\bX_{t})\right\|^{2}\right]\\
		%     &\quad +
		%     \mE_{\bX_{t}}\left[\left\|\bs_{t}(\bX_{t}; \bar{\btheta}_{t}) - \nabla\log{p_{t}}(\bX_{t})\right\|^{2}\right]
		%     \\& =
		%     \underbrace{\mE_{\bX_{t}}[\|\bs_{t}(\bX_{t};  \htheta_{t}) - \bs_{t}(\bX_{t}; \bar{\btheta}_{t})\|^{2}]}_{\cE_{\rm est}} +
		%     \underbrace{\mE_{\bX_{t}}\left[\left\|\bs_{t}(\bX_{t}; \bar{\btheta}_{t}) - \nabla\log{p_{t}}(\bX_{t})\right\|^{2}\right]}_{\cE_{\rm app}},
		% \end{aligned}
	% \]
	% where $\cE_{\rm est}$ and $\cE_{\rm app}$ are the estimation and approximation error, respectively.
	% Proposition \ref{prop:approximation error} establishes the bound for the approximation error $\cE_{\rm app}$. It remains to establish the bound for $\cE_{\rm est}$. 
To this end, we first introduce some assumptions and definitions. Recall the approximation functions $\bbf_{t,b}$ and remainders $\bkappa_{t,b}$ in Section \ref{sec:kernel-based function}. Throughout the following bounds, $b>0$ is deterministic and may be designed as depended on $n$, $\bbH$ is fixed independently of the training observations, and $\mE_{\bX_t}$ integrates over a fresh evaluation point.
\begin{assumption}\label{ass:moment norm}
For every $b>0$, the same functions $\bbf_{t,b}$ satisfy, uniformly over $t\in[0,1]$,
\begin{equation*}
\mE_{\bX_t}\left[\exp\left\{|\bbOmega_t^{(0)}|^{-1/2}|\bbOmega^{(0)}|^{1/2}\|\bbf_{t,b}(\bX_t)\|^2\right\}\right]\leq\exp\{\cB(b)\},
\end{equation*}
where $\cB(b)<\infty$ may depend on $b$ and $d$.
\end{assumption}
	This assumption and the constant $\cB(b)$, which depends on $d$, are critical in our bound for score function estimation error. In worst case, the $\cB(b)$ may scale exponentially with $d$. For example,
	if each component of $\bbf_{t,b}$ is bounded, then we have $\sup_{\bx}|\bbOmega_{t}^{(0)}|^{-1/2}|\bbOmega^{(0)}|^{1/2}\|\bbf_{t,b}(\bx)\|^{2} = \cO(d |\bbOmega_{t}^{(0)}|^{-1/2}|\bbOmega^{(0)}|^{1/2})$. For the recommended choice $\bbH=(8\bbSig+\epsilon\bI)^{-1}$ with $\epsilon\to 0$, the determinant ratio is bounded by $\{2/(1+t)\}^{d/2}$; equality holds for the limiting choice $\bbH=(8\bbSig)^{-1}$. Then, the $\cB(b)$ scales with $d$ at the order of $\cO(dc^{d})$ for some $c > 1$. 
	
	Despite the worst case, we consider a more general claim. 
Let $0<\eta<1/4$.
	%and $\bGamma_{t} = \mE[\bGamma_{\bX_{1}, t}]$. 
	Similar to Proposition \ref{prop:approximation error}, define
\begin{equation*}
		\delta_{1n}(\eta) = 2\left(\frac{H_{1}}{n - m} + \frac{\sigma_{1}}{\sqrt{n - m}}\right)\log\left(\frac{2}{\eta}\right),
		\enspace
		\delta_{2n}(\eta) = 2\left(\frac{H_{2}}{n - m} + \frac{\sigma_{2}}{\sqrt{n - m}}\right)\log\left(\frac{2}{\eta}\right),
	\end{equation*}
For $K\geq8$ and $0<\eta<1/4$, $\log(2K/\eta)\leq\log(2/\eta)\log K$, so
$\delta_{0m}(\eta/K;b)\leq\delta_{0m}(\eta;b)\log K$,
$\delta_{an}(\eta/K)\leq\delta_{an}(\eta)\log K$ for $a=1,2$.
Define
\begin{equation*}
\begin{aligned}
C_{\eta,m,\lambda}(b)&=\lambda^{-1}\left\{\frac{2\sigma_0(b)^2}{m}+2\delta_{0m}(\eta;b)^2\log^2 K+2\kappa(b)\right\}+\cB(b)+\log(K\eta^{-1}),
\end{aligned}
\end{equation*}
	where $\sigma_{1}$, $\sigma_{2}$, $H_{1}$ and $H_{2}$ are the constants in %in the Bernstein-type concentration condition, 
	Assumption \ref{ass:bernstein} in Supplementary Material \ref{app:regularity conditions}.
	We have the following non-asymptotic upper bound for $\mE_{\bX_{t}}[\|\bs_{t}(\bX_{t}; \htheta_{t}) - \nabla\log{p_{t}}(\bX_{t})\|^{2}]$. 
	\begin{theorem}\label{thm:gen error}
Fix $b>0$, $0<\eta<1/4$, and $K\geq8$, and suppose $\lambda>2\delta_{2n}(\eta)\log K$.
Under Assumption \ref{ass:moment norm} and Assumptions \ref{ass:lip continuity}, \ref{ass:subexp}, and \ref{ass:bernstein} in Supplementary Material,
\begin{equation*}
\begin{aligned}
&\mE_{\bX_t}\left[\left\|\bs_t(\bX_t;\htheta_t)-\nabla\log p_t(\bX_t)\right\|^2\right]\\
&\quad\leq 2\lambda\{\cB(b)+\log(K\eta^{-1})\}
    +\frac{4\sigma_0(b)^2}{m}+4\delta_{0m}(\eta;b)^2\log^2 K+4\kappa(b)\\
&\qquad\quad+\frac{4\{\delta_{1n}(\eta)+2\delta_{2n}(\eta)\sqrt{C_{\eta,m,\lambda}(b)}\}^{2}\log^2 K}{\lambda}
\end{aligned}
\end{equation*}
holds simultaneously for all $t\in\{1/K,\ldots,1\}$ with probability at least $1-4\eta$.
	\end{theorem}
Notably, the randomness in the inequality comes from the training sample. The proof of this theorem is in Supplementary Material \ref{app:generalization proof}. Taking $m\asymp n^{1/2}$ and $\lambda\asymp n^{-1/2}\log(2/\eta)\log K$, with the multiplicative constant in $\lambda$ large enough that $\lambda>2\delta_{2n}(\eta)\log K$, gives, on the same event,
\begin{equation*}
\begin{aligned}
&\sup_{t\in\{1/K,\ldots,1\}}\mE_{\bX_t}\left[\left\|\bs_t(\bX_t;\htheta_t)-\nabla\log p_t(\bX_t)\right\|^2\right]\\
&\quad\leq\cO\Bigg(\kappa(b)
    +\frac{\{\cB(b)+\log(K/\eta)\}\log(2/\eta)\log K}{\sqrt n}+\left\{\frac{\sigma_0(b)^2}{\sqrt n}+\frac{H_0(b)^2}{n}\right\}\log^2(2/\eta)\log^2 K\Bigg).
\end{aligned}
\end{equation*}
The constant is independent of $b,n,K,\eta$, with the constants in Assumption \ref{ass:bernstein} and the data dimension held fixed. Combining this bound with Theorem \ref{thm:sampling error} gives the following result, which characterizes the distributional error.  
\begin{corollary}\label{cor:convergence rate}
Fix $b>0$, $0<\eta<1/4$, and $K\geq8$. Suppose $m\asymp n^{1/2}$, $\lambda\asymp n^{-1/2}\log(2/\eta)\log K$ with $\lambda>2\delta_{2n}(\eta)\log K$, and $\hX_1$ is generated as in Theorem \ref{thm:sampling error}. Under Assumption \ref{ass:moment norm} and the assumptions in Supplementary Material \ref{app:regularity conditions},
\begin{equation*}
\begin{aligned}
\TV(P_{\bX_0},P_{\hX_1})\leq\cO\Bigg(&\frac{d}{K}+\kappa(b)
    +\frac{\{\cB(b)+\log(K/\eta)\}\log(2/\eta)\log K}{\sqrt n}\\
    &+\left\{\frac{\sigma_0(b)^2}{\sqrt n}+\frac{H_0(b)^2}{n}\right\}\log^2(2/\eta)\log^2 K\Bigg)^{1/2}
\end{aligned}
\end{equation*}
holds with probability at least $1-4\eta$.
Assume, in addition, there exist fixed constants $\rho>0$, $\gamma>0$, and $C>0$, independent of $b,n,K,\eta$, such that, for every $b\geq1$,
\begin{equation}\label{eq:b growth conditions}
\kappa(b)\leq Cb^{-\rho},
\qquad
\cB(b)+\sigma_0(b)^2+H_0(b)^2\leq Cb^{\gamma},
\end{equation}
then, for fixed $d$ and $\eta$ and sufficiently large $n$ such that $\log^2(2/\eta)\log^2 K\leq\sqrt n$, choose
\begin{equation}\label{eq:optimal b}
b=\left\{\frac{\sqrt n}{\log^2(2/\eta)\log^2 K}\right\}^{1/(\rho+\gamma)}\ \text{and} \ K=\max\left\{8,\left\lceil\left(\frac{\sqrt n}{\log^2 n}\right)^{\rho/(\rho+\gamma)}\right\rceil\right\}.
\end{equation}
Then,
\[
\TV(P_{\bX_0},P_{\hX_1})
\leq\cO\left(n^{-\rho/\{4(\rho+\gamma)\}}(\log n)^{\rho/(\rho+\gamma)}\right)
\]
with probability at least $1-4\eta$.
\end{corollary}
The exponent $\rho$ describes the decay of the squared approximation error, while $\gamma$ bounds the growth of the moment and concentration constants. The TV bound also controls the uniform difference between the distribution functions of $\hX_1$ and $\bX_0$. For any fixed $\epsilon$ satisfying $\gamma/\{4(\rho+\gamma)\}<\epsilon<1/4$, the rate in Corollary~\ref{cor:convergence rate} is $\cO_P\{(\log n/n)^{1/4-\epsilon}\}$; thus $\epsilon$ can be small when $\rho$ is large relative to $\gamma$.
 This rate is faster than those achieved by classical kernel-smoothing-based deconvolution methods, which are polynomial in $1/\log n$ under H\"older-smooth densities of error-free data \citep{fan1991optimal} and of order $n^{-c/\sqrt{\log n}}$, up to logarithmic factors, for some $c>0$ under analytic densities \citep{butucea2008sharp}. Under suitable regularity conditions, convergence of the distribution function further implies convergence of the corresponding quantiles \citep{van2000asymptotic}. Notably, these improved rates are obtained under stronger assumptions on the density of the error-free data. Specifically, the rate relies on the quantitative approximation and moment-growth conditions in \eqref{eq:b growth conditions}; Assumption~\ref{ass:subexp} alone does not require the true score to be analytic. Such conditions are not required by the deconvolution methods.

	\section{Simulations}\label{sec:simulation}
	The goal of this section is to evaluate the proposed method's finite-sample performance in generating samples from the estimated latent distribution when both covariates and responses are observed with measurement error. The simulation design is deliberately model-free: the data-generating mechanism is known to us, but the denoising methods only use error-contaminated observations. 
	% This allows us to assess both distributional recovery and downstream prediction.
	\begin{figure}[t!]
		\vspace{-0.1in}
		\centering
		\begin{subfigure}{\textwidth}
			\centering
			\includegraphics[width=0.9\textwidth]{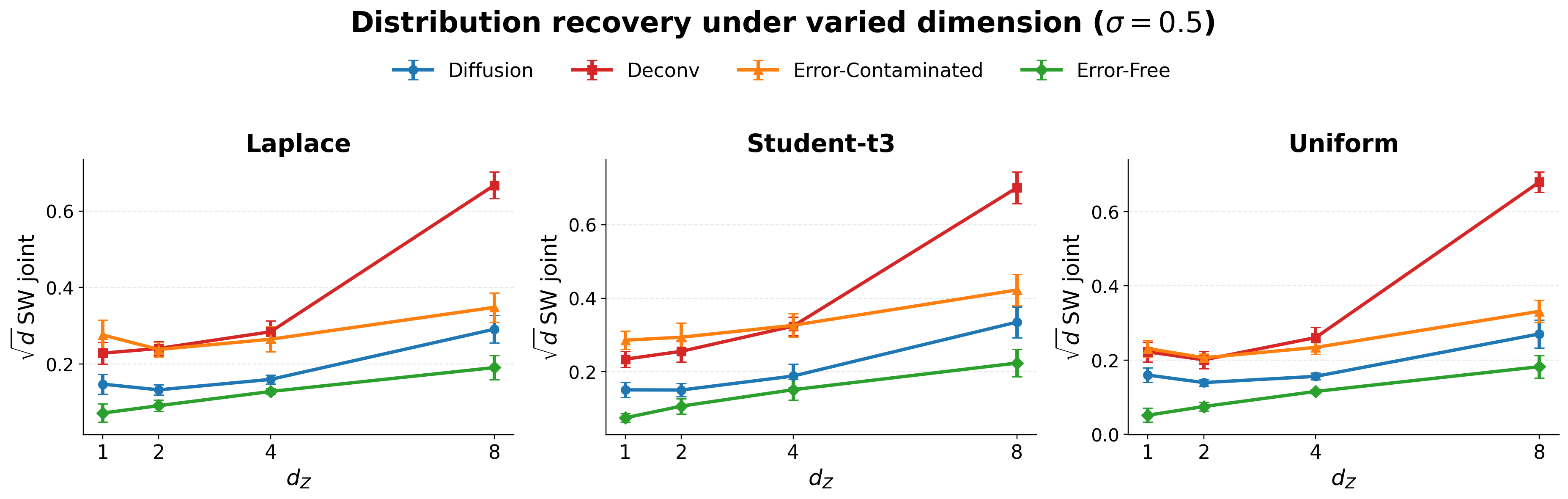}
			\vspace{-0.1in}
			\caption{Scaled joint sliced Wasserstein distance $\mathrm{SW}_{1}^{\mathrm{sc}}$ of denoised data.}
			\label{fig: sim dimension sw}
		\end{subfigure}
		\begin{subfigure}{\textwidth}
			\centering
			\includegraphics[width=0.9\textwidth]{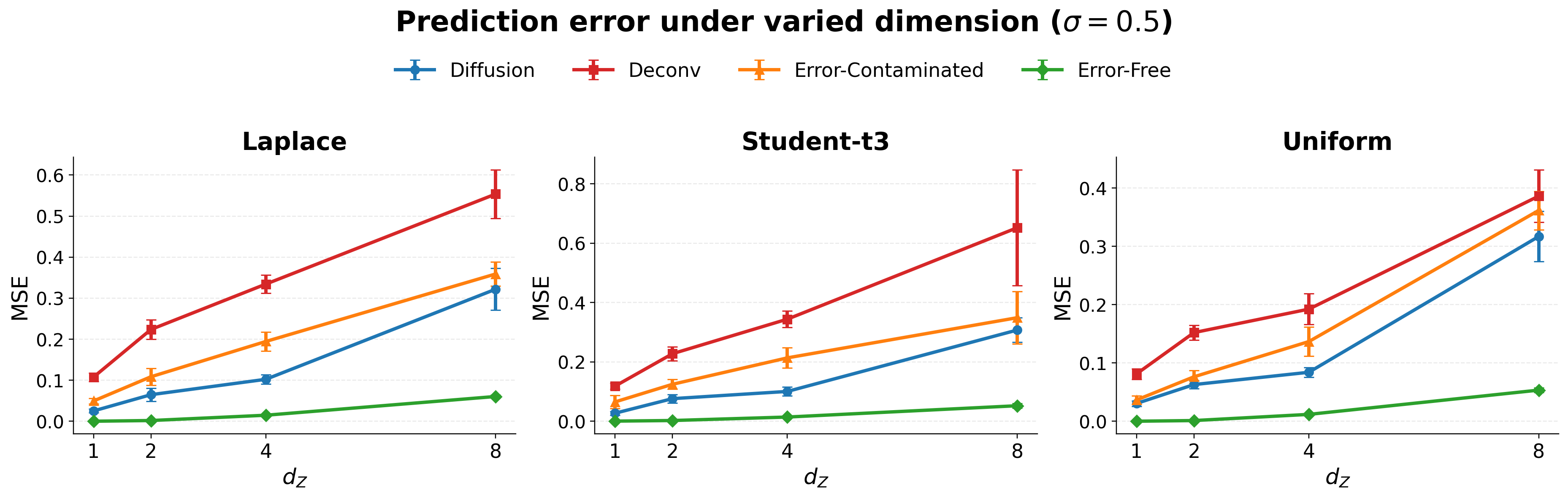}
			\vspace{-0.1in}
			\caption{MSE compared with the ground truth model.}
			\label{fig: sim dimension mse}
		\end{subfigure}
		\caption{Simulations with varying dimension $d_{\bZ}$ with $\sigma=0.5$. The three panels in each subfigure correspond to different distributions of $\bZ_{0}$. Error bars show one standard deviation over $10$ independent runs.}
		\label{fig: sim dimension}
		\vspace{-0.2in}
	\end{figure}
	First, we generate the error-free $\bZ_{0}$ in dimension $d_{\bZ}$ from three non-Gaussian product distributions whose marginal variances are standardized to one.
	\begin{itemize}
		\item \textit{Laplace}: independent Laplace marginals with location zero and scale $1/\sqrt{2}$;
		\item \textit{Student-$t$}: independent Student-$t_{3}$ marginals scaled by $1/\sqrt{3}$;
		\item \textit{Uniform}: independent Uniform marginals on $[-\sqrt{3},\sqrt{3}]$.
	\end{itemize}
	Next, we generate the response $Y_{0} = f_{\bbeta_{\mathrm{truth}}}(\bZ_{0}),$
	% \begin{equation*}
		% 	Y_{0} = f_{\bbeta_{\mathrm{truth}}}(\bZ_{0}),
		% \end{equation*}
	where $f_{\bbeta_{\mathrm{truth}}}$ is a three-layer neural network with He initialization \citep{he2015delving}. Finally, we generate the error-contaminated observation $\bX_{1}=(\bZ_{1},Y_{1})$ by adding independent Gaussian measurement error with covariance $\bbSig = \sigma^{2}\bI$ to both $\bZ_{0}$ and $Y_{0}$.
	\vspace{-0.2in}
	\paragraph{Setup.}
	For the error-contaminated observation $\bX_{1}=(\bZ_{1},Y_{1})$, we generate $n=1000$ samples. For the \texttt{Diffusion} denoising method, we use $m=64$ kernel centers to estimate the score function \eqref{eq: kernel model}, set the regularization coefficient to $\lambda=n^{-1/2}$, and use a full reverse-time grid with $K=50$ intervals. In the Gaussian kernel \eqref{eq:Gaussian kernel}, we find that the anisotropic bandwidth matrix $\bbH = 0.5 \widehat{\bbSig}_{\bX_{1}}^{-1}$ 
	% \[
	% \bbH = 0.5 \widehat{\bbSig}_{\bX_{1}}^{-1}, \qquad h=0.5,
	% \]
	works best in practice, where $\widehat{\bbSig}_{\bX_{1}}$ is the empirical covariance matrix of the error-contaminated joint vector, compared to the conservative choice $\bbH=(8\bbSig+\epsilon\bI)^{-1}$ with $\epsilon>0$ given above. This empirical bandwidth is data-dependent, whereas the score-estimation theory assumes a fixed bandwidth. 
	% We implement the reverse-time dynamics using the probability-flow ODE version of Algorithm \ref{alg:alg1}.
	\vspace{-0.2in}
	\paragraph{Evaluation.}
	We generate an independent test sample of $n_{\rm test}=1000$ observations from the same error-free distribution. To measure distributional recovery, we report $\mathrm{SW}_{1}^{\mathrm{sc}}=\sqrt{d}\,\mathrm{SW}_{1}$, where $\mathrm{SW}_{1}$ is the unscaled first-order sliced Wasserstein distance and $d=d_{\bZ}+1$ is the joint dimension (the factor $\sqrt{d}$ compensates for the scale reduction induced by one-dimensional projection) \citep{villani2008optimal} between the generated joint sample $(\hZ_{0},\hat Y_{0})$ (the components of the terminal reverse-time output $\hX_1$) and the independent error-free test sample $(\tZ_{0},f_{\bbeta_{\mathrm{truth}}}(\tZ_{0}))$; the empirical estimator is described in Supplementary Material \ref{app:sw-distance}. To further evaluate the denoising quality, we report a mean squared error (MSE) $n_{\rm test}^{-1}\sum_{i=1}^{n_{\rm test}}
	\left\{ f_{\bbeta_{\mathrm{truth}}}(\tZ_{0}^{(i)})
	- f_{\hbeta}(\tZ_{0}^{(i)}) \right\}^{2},$ 
	% \begin{equation*}
		% 	\frac{1}{n}\sum_{i=1}^{n}
		% 	\left\{ f_{\bbeta_{\mathrm{truth}}}(\tZ_{0}^{(i)})
		% 	- f_{\hbeta}(\tZ_{0}^{(i)}) \right\}^{2},
		% \end{equation*}
	where $f_{\hbeta}$ is trained using the generated or baseline data. 
	\begin{figure}[t!]
		\vspace{-0.1in}
		\centering
		\begin{subfigure}{\textwidth}
			\centering
			\includegraphics[width=0.9\textwidth]{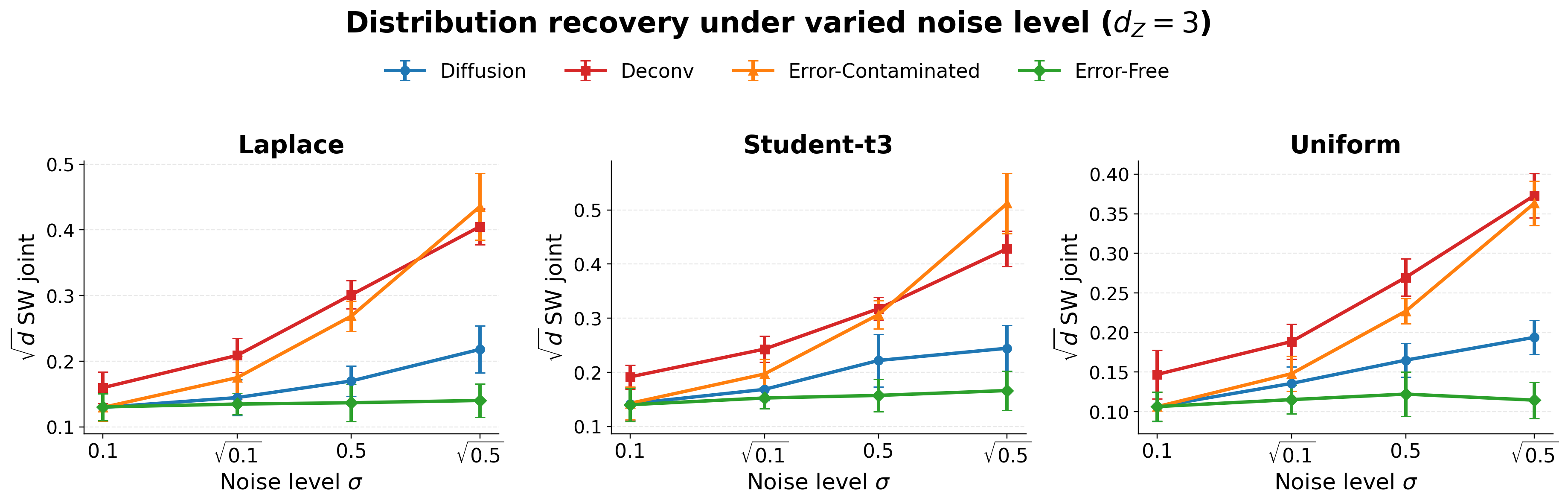}
			\caption{Scaled joint sliced Wasserstein distance $\mathrm{SW}_{1}^{\mathrm{sc}}$ of denoised data.}
			\label{fig: sim sigma sw}
		\end{subfigure}
		\begin{subfigure}{\textwidth}
			\centering
			\includegraphics[width=0.9\textwidth]{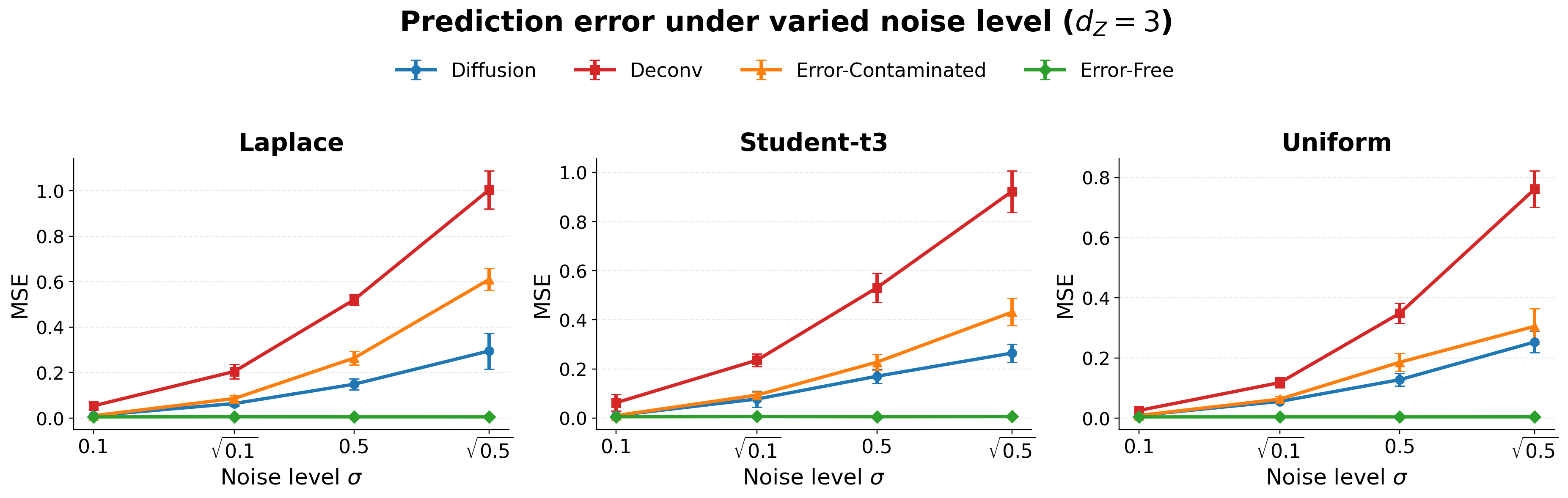}
			\caption{MSE compared with the ground truth model.}
			\label{fig: sim sigma mse}
		\end{subfigure}
		\caption{Results with varying noise level $\sigma$ with $d_{\bZ}=3$. The panels in each subfigure correspond to different distributions of $\bZ_{0}$. Error bars show one standard deviation over $10$ independent runs.}
		\label{fig: sim sigma}
		\vspace{-0.2in}
	\end{figure}
	\vspace{-0.1in}
	
	We compare our proposed \texttt{Diffusion} method with three baseline methods: \texttt{Error-Free}, an oracle method using the error-free sample; \texttt{Error-Contaminated}, which ignores measurement error and uses the error-contaminated sample directly; and \texttt{Deconv}, a deconvolution baseline based on the kernel deconvolution idea of \cite{fan1991optimal}. The \texttt{Deconv} implementation deconvolves the marginal distribution of each coordinate of $\bZ$ and uses local kernel regression to recover $Y$ conditionally on the deconvolved covariates. This implicitly uses the prior knowledge that the coordinates of $\bZ_{1}$ are mutually independent. By doing so, we alleviate the curse of dimensionality for the \texttt{Deconv} method, since the original version performs poorly in our experiments. Another diffusion based method \texttt{Recursive-GSURE} \citep{daras2024ambient} use constructed denoised data to estimate the score function of the intermediate noisy data.

	We summarize the results in Figures \ref{fig: sim dimension} and \ref{fig: sim sigma}, where each curve reports the mean over $10$ independent repetitions and the error bars show one standard deviation. Across the three distributions of $\bZ_{0}$, Figures \ref{fig: sim dimension sw}--\ref{fig: sim sigma mse} show that \texttt{Diffusion} consistently improves upon the deconvolution and error-contaminated baselines in both joint SW and MSE. At very low noise levels, the error-contaminated baseline can be competitive because the observed sample is already close to the error-free sample. Moreover, although we decouple the coordinate correlations of $\bZ_{1}$ in the deconvolution baseline, unlike our \texttt{Diffusion} method, it still works poorly in high-dimensional regimes. The oracle error-free method provides the empirical lower benchmark; its SW values are nonzero due to finite-sample variability between independent training and testing samples. The results show that our \texttt{Diffusion} method works well for all light-tailed, heavy-tailed, and bounded distributions $\bZ_{0}$. 
	
	% \vspace{-0.05in}
	We further evaluate these methods under varying noise levels $\sigma$. The results summarized in Figure \ref{fig: sim sigma} further verify the effectiveness of our proposed \texttt{Diffusion} method. 
	
	\vspace{-0.2in}
	\paragraph{\textbf{Ablation Study}} For our \texttt{Diffusion} method, we further investigate the impact of the regularization parameter \( \lambda \), the number of kernel centers \(m\), the bandwidth scale \(h\), and the number of discretization steps $K$ in Supplementary Material \ref{app:ablation study}. The results indicate that our method remains stable across variations of hyperparameters. We also discuss the unknown-$\sigma$ setting in Supplementary Material \ref{app:unknown-sigma}. 
	
	\section{Real Data Analysis}\label{sec:real-data}
	\vspace{-0.1in}
	\subsection{Impact of Glucose Monitor Measurement Error in Diabetes Clinical Trials}\label{sec:application}
	\vspace{-0.1in}
	\textit{Continuous Glucose Monitors} (CGMs) are minimally invasive technologies used to record glucose levels over prolonged periods, such as days and weeks. Since their adoption nearly 20 years ago in clinical practice, the benefits of CGMs in managing type 1 diabetes and evaluating the efficacy of new drugs in diabetes trials have become clear. More recently, CGMs have also found novel applications, for example, in human nutrition to support personalized diet prescriptions. 
	% Moreover, CGMs are increasingly popular in epidemiological studies involving healthy populations.
	
	Despite their utility, CGMs are subject to significant measurement errors, which can affect the conclusions drawn from the data, such as those in clinical trials. Although the measurement error of CGM technologies has decreased over time, e.g., the mean absolute relative difference has dropped from approximately 20\% to about 10\%, it is  enough to cause misleading conclusions about an individual's true metabolic status and its evolution.
	
	From a technological perspective, understanding the limitations and potential of CGMs in relation to measurement error through novel statistical methods is a major challenge, particularly for regulatory approvals (e.g., from the FDA) of predictive AI algorithms or for defining novel digital biomarkers. Our denoising diffusion framework provides an effective approach to performing sensitivity analyses on how measurement error affects the estimation of possibly complicated digital health metrics and, ultimately, clinical decision-making. The analysis goals of this section are to:
	i)
	Use denoising diffusion algorithms to estimate individual distributional representations of CGM data under measurement error;
	ii) Incorporate denoising steps to enable sensitivity analyses that evaluate new clinical interventions in clinical trials, taking into account the underlying noise level. We focus on assessing statistical inference when comparing two different interventions.
	
	\vspace{-0.2in}
	\paragraph{\textbf{Data Structure and Collection.}} 
	Consider \(n\) individuals, each with \(n_i\) glucose observations collected at times \(s^{(i,j)}\), giving  
	\(\bigl\{\bigl(s^{(i,j)}, G^{(i,j)}\bigr)\bigr\}_{j=1}^{n_i}\).
	Let \(X^{(i)}(s)\) be the unobserved glucose process for the \(i\)-th individual, where \(s \in [0,\mathcal{S}^{(i)}]\). 
	The observed measurements satisfy
	$ G^{(i,j)}
	\;=\;
	X^{(i)}\!\bigl(s^{(i,j)}\bigr)
	\;+\;
	\epsilon^{(i,j)}
	$,
	where \(\epsilon^{(i,j)} \sim \mathcal{N}(0,\sigma_{0}^2)\) are i.i.d.\ across all \(i,j\). Although this may be simplistic for real continuous glucose monitoring (CGM) data, it provides a tractable framework for assessing measurement error sensitivity.
	
	From a population perspective, we focus on the marginal distribution of \(X^{(i)}\). Define
	\[
	F^{(i)}(g)
	\;=\;
	\frac{1}{\mathcal{S}^{(i)}}
	\int_{0}^{\mathcal{S}^{(i)}}
	\mathbf{1}\!\bigl\{X^{(i)}(u)\le g\bigr\}\,\mathrm{d}u,
	\quad
	g \in [40,400].
	\]
	Various metrics, such as time in hypoglycemia, time in range (TIR), and time in hyperglycemia, can be derived from \(F^{(i)}\). 
	We often categorize glucose readings as \(G^{(i,j)}\le 70\) (hypoglycemia), \(70< G^{(i,j)}\le 180\) (TIR), or \(G^{(i,j)}> 180\) (hyperglycemia). In particular, the hypoglycemia, TIR, and hyperglycemia proportions are
	$\mathrm{Hypo}^{(i)} = F^{(i)}(70)$, $\mathrm{TIR}^{(i)} = F^{(i)}(180) - F^{(i)}(70)$ and $\mathrm{Hyper}^{(i)} = 1 - F^{(i)}(180)$, respectively.
	For estimation, denoised observations \(\bigl\{G^{(\sigma,i,j)}\bigr\}_{i, j}\) can be used, where \(\sigma=0\) corresponds to the raw measurements \((G^{(0,i,j)}=G^{(i,j)})\). In addition to discrete summaries, we also consider the functional ``glucodensity'' profile \citep{doi:10.1177/0962280221998064}, a density estimator that, in practice, can be estimated via kernel-based methods.
	
	Our goal is to quantify how measurement error affects these population-level CGM metrics and to develop estimators that minimize the related discrepancies. We assume that the population cdf \(F^{(i)}\) corresponds to a continuous-time stochastic process \(X^{(i)}\) that is ``slowly varying'' or approximately stationary during the observational window. Biologically---especially in diseases like diabetes---it is reasonable to assume that the metabolic evolution of the disease is not rapid so that the distribution \(F^{(i)}\) remains roughly time-invariant over the study period. Summarizing this distributional behavior via the marginal cdf \(F^{(i)}\) is therefore a good proxy for tracking glucose health patterns and a good method to predict long-term diabetes outcomes. % \citep{matabuena2024glucodensity}.
	
	\vspace{-0.2in}
	% {\noindent\bf\large }
	% \label{sec:data-description}
	\paragraph{\textbf{Data Description.}}
	Our study is motivated by data from the JDRF CGM Study Group \citep{doi:10.1056/NEJMoa0805017,juvenile2009effect}, an early large-scale clinical trial evaluating CGM for managing type 1 diabetes (T1DM). Data were obtained from the multi-center trial\footnote{JDRF CGM RCT, NCT00406133; \url{https://public.jaeb.org/datasets/diabetes}}, wherein 451 adults and children were randomized to either a CGM-based (treatment) arm or a standard-of-care (control) arm.
	
	For this analysis, we focus on 188 individuals (102 control, 86 treatment) with minimal missing CGM data at the study's start and end of the intervention. We focus on CGM-based measures (time in hypoglycemia, hyperglycemia, or the target range) and how they are affected by measurement error in the underlying glucose process. Note that these CGM metrics are among the current state-of-the-art metrics promoted by the American Diabetes Association.
	
	\vspace{-0.2in}
	\paragraph{\textbf{Density Function Estimation with Measurement Error.}}
	% {\noindent\bf\large Density Function Estimation with Measurement Error}
	Given a noise level \(\sigma\), for each individual \(i = 1, \dots, n\), we use the denoising framework to generate the sample \(G^{(\sigma, i,j)}\) for \(j = 1, \dots, n_i\). We estimate the marginal density of \(\{G^{(\sigma, i,j)}\}_{j = 1}^{n_{i}}\) for each individual using a kernel density estimator as in \cite{matabuena2024glucodensity}. 
	% \[
	% 	\tilde{f}_{i}^{(\sigma)}(x)
	% 	\;=\; 
	% 	\frac{1}{n_i}
	% 	\sum_{j=1}^{n_i} 
	% 		K_{h_i}\!\Bigl(x - G^{(\sigma, i,j)}\Bigr),
	% \]
	% where 
	% \[
	% 	K_{h_i}(u) 
	% 	\;=\; 
	% 	\frac{1}{h_i} 
	% 	K\!\left(\frac{u}{h_i}\right).
	% \]
	% Here, \(h_i > 0\) is the smoothing bandwidth and \(K(\cdot)\) is a chosen kernel function (e.g., Gaussian). The specific form of \(K\) has a relatively minor impact on the efficiency of the estimator, whereas the bandwidth \(h_i\) is crucial for controlling bias and variance. We select the \(h_i\) parameter via Gaussian asymptotic approximations.
	Figure \ref{fig:combined} shows the density estimators for four individuals at noise levels \(\sigma = 0, 10, 15, 20, 25, 30\). The first two individuals belong to the control group, and the last two belong to the treatment group at the beginning of the study. We estimate the corresponding densities at both the beginning and end of the interventions. We observe that different noise levels result in varying shapes of density functions, which can impact the proportion of time spent in clinically relevant target zones defined previously.
	
	To display these conclusions with higher resolution for noise levels \(\sigma = 0, 10, 20\), Supplementary Material \ref{app:figure for diabetes} illustrates the variations in the control and treatment groups (RT-CGM) before and after interventions. For example, when comparing \(\sigma = 0\) versus \(\sigma = 20\) in the boxplot for low glucose concentration (hypoglycemia) and high glucose concentration (hyperglycemia), the median range of variation increases to \(1.5\) and \(0.75\), respectively. Although there are certain original variations, for TIR (normal glucose concentration), the estimate is reduced by approximately \(3\%\) of the time, indicating that the original analysis overestimated the proportion of time that individuals spend within this range. From a clinical point of view, a variation in the hypo- or hyperglycemia proportion of \(0.5\%\) can be clinically relevant, especially for hypoglycemia, when detecting statistically significant differences in randomized clinical trials.
	\begin{figure}[t!]
		\vspace{-0.1in}
		
		\centering
		\begin{subfigure}[b]{0.49\textwidth}
			\centering
			\includegraphics[width=\textwidth]{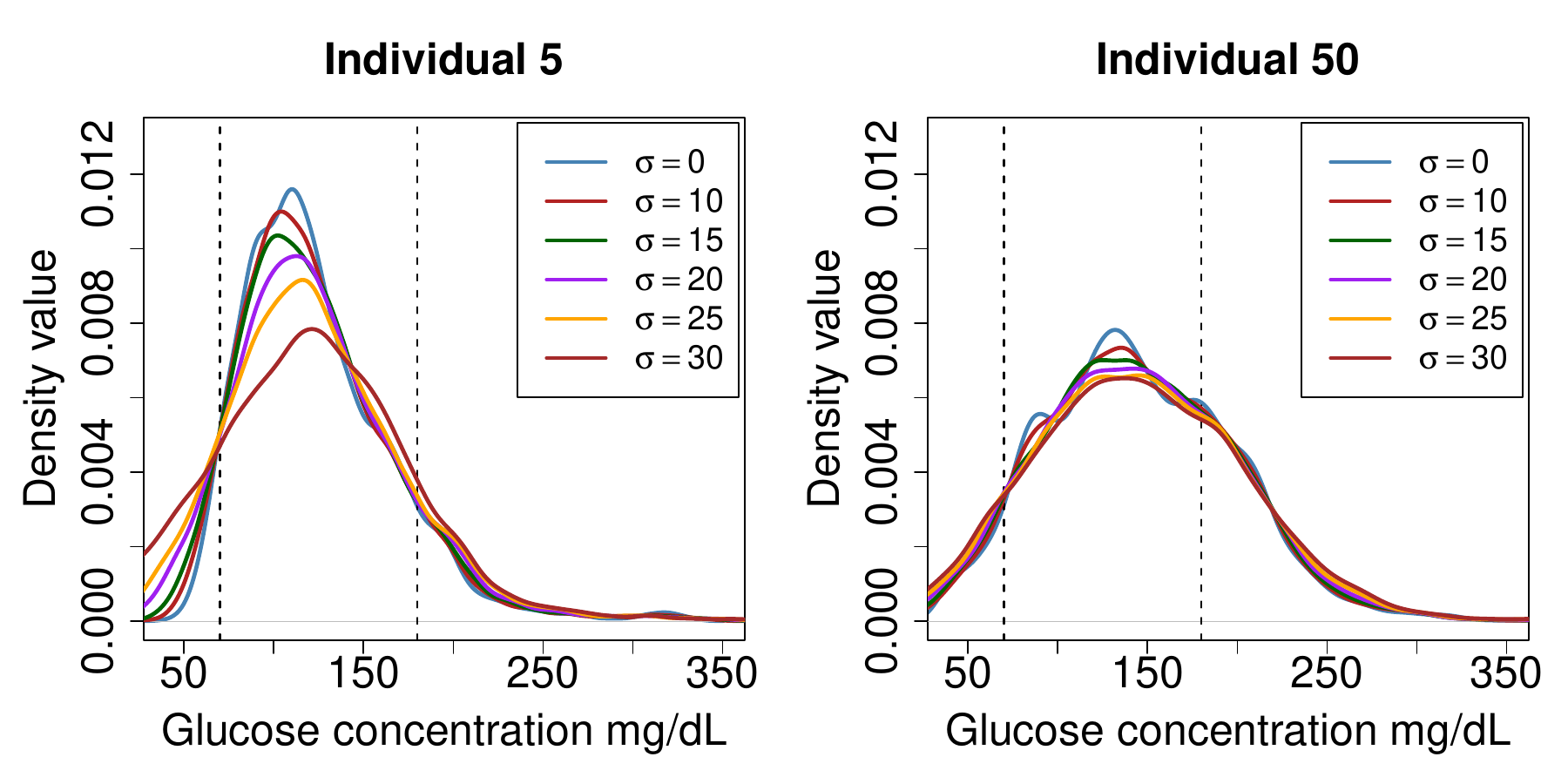}
			\caption{Control group}
			\label{fig:sub1}
		\end{subfigure}
		\hfill
		\begin{subfigure}[b]{0.49\textwidth}
			\centering
			\includegraphics[width=\textwidth]{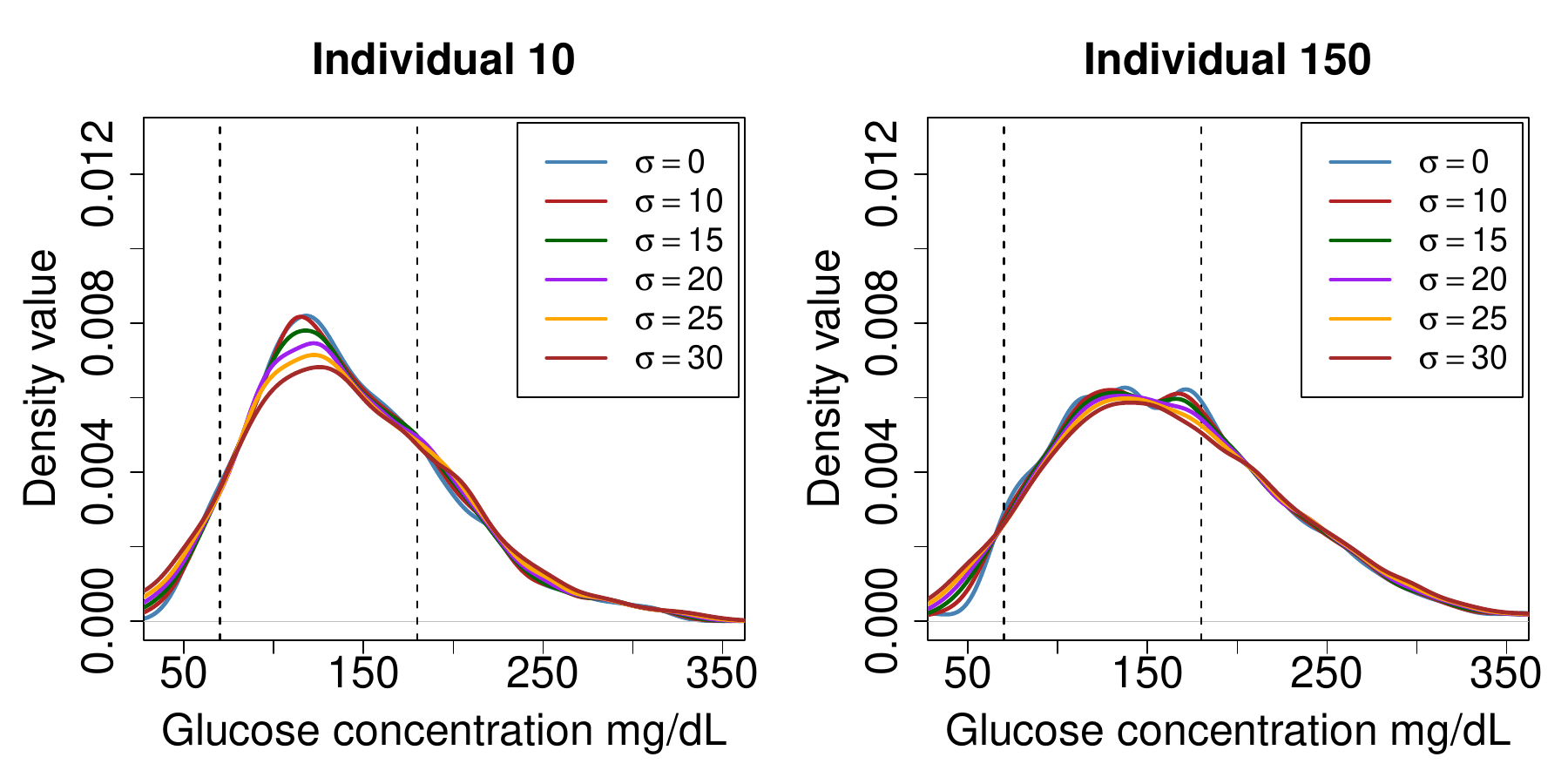}
			\caption{Treatment group}
			\label{fig:sub2}
		\end{subfigure}
		\caption{Density estimators for four individuals at different noise levels \(\sigma = 0, 10, 15, 20, 25, 30\) at the beginning.}
		\label{fig:combined}
		\vspace{-0.2in}
	\end{figure}
	
	% {\noindent\bf\large Impact of Measurement Error in Intervention Assessment}
	\vspace{-0.2in}
	\paragraph{\textbf{Impact of Measurement Error in Intervention Assessment.}}
	Now, we focus on the changes in the control and treatment groups (RT-CGM) after the interventions for the following metrics: hypoglycemia (hypo), hyperglycemia (hyper), and Time in Range (TIR). We consider a boxplot of the differences for each individual between the end and the beginning of the intervention. Figure \ref{fig:ejemplo2} shows these results along with the corresponding p-values derived from a paired t-test for different noise levels \(\sigma = 10, 15, 20, 25, 30\).
	In the case of the hypoglycemia metric for different noise levels, there are no statistically significant differences, similar to the TIR metrics. However, the p-values are highly unstable, and the boxplots indicate important changes in the distribution across different noise levels \(\sigma\). For hyperglycemia (Hyper), the p-values exceed the nominal threshold $0.05$ at \(\sigma = 15\) and \(\sigma = 25\). Thus, whether the p-value crosses the nominal threshold depends on the assumed noise level.
	
	We recognize that there can be an estimation error between the denoised-data distribution and the true error-free distribution even when \(\sigma\) is correctly specified, and the theoretical validity of the t-test p-values may not be fully justified. However, the p-values and the boxplots under different \(\sigma\) values are still useful exploration tools for assessing the impact of measurement error in the clinical study.
	
	\begin{figure}[t!]
		\centering
		\vspace{-0.1in}
		\includegraphics[angle=0, width=1\textwidth]{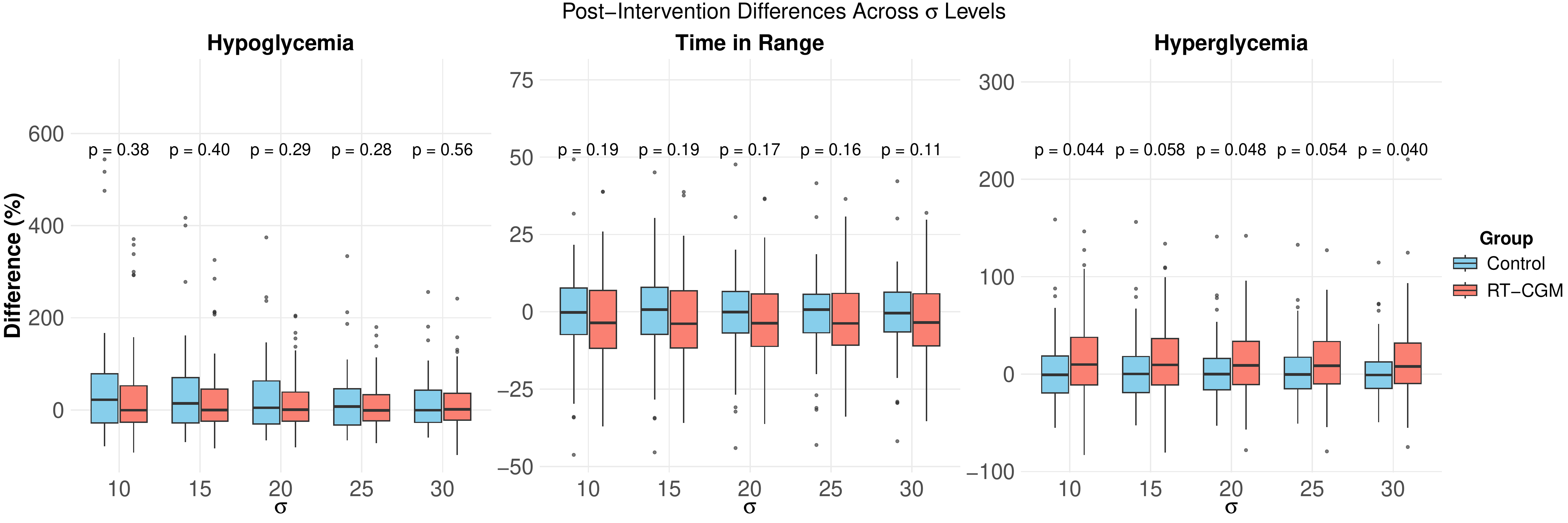}
		\caption{Boxplots of the differences in hypoglycemia (Hypo), hyperglycemia (Hyper), and Time in Range (TIR) metrics for control and treatment groups (RT-CGM) before and after interventions across various noise levels (\(\sigma = 10, 15, 20, 25, 30\)). The figure also shows the corresponding p-values from paired t-tests at each noise level.}
		\vspace{-0.2in}
		\label{fig:ejemplo2}
	\end{figure}
	
	\subsection{Association between High-Density Lipoprotein Cholesterol and Coronary Heart Disease Risk}\label{sec:lipids}
	The preceding analyses focus on CGM-derived digital biomarkers. Next, we consider a classical lipid example to assess how measurement error may affect inference in a standard multivariable regression setting. 
	The association between High-Density Lipoprotein Cholesterol (HDL-C), triglyceride, and Coronary Heart Disease (CHD) risk has long served as a cautionary example in observational epidemiology. HDL-C and triglyceride are correlated lipid measures, and when such lipid exposures are measured imprecisely, conventional adjustment for one exposure while estimating the association of the other can yield unstable and potentially misleading ``independent'' associations. Although observational studies have often reported an inverse association between HDL-C and CHD, evidence from randomized trials and Mendelian randomization analyses has challenged a causal interpretation of this association \citep{daveysmith2020correlation}. A key concern is that the measurement error in HDL-C and triglyceride may distort their mutually adjusted associations with CHD, making it difficult to separate their respective roles using regression \citep{daveysmith2020correlation}. Motivated by this example, we apply the proposed denoising method to examine how sensitive the adjusted HDL-C--CHD association is to possible measurement error in the lipid measurements.
	
	We use data from the National Health and Nutrition Examination Survey (NHANES) for the 2011--2012 and 2013--2014 cycles. After excluding individuals with missing CHD status, HDL-C, or triglyceride measurements, the analysis includes 4975 individuals. We fit a logistic regression model for CHD status with HDL-C and triglyceride as covariates. Both lipid variables are centered and standardized before entering the regression model, so that the corresponding coefficients are interpreted on a standard-deviation scale. To evaluate the sensitivity of the adjusted HDL-C association to possible measurement error, we treat the observed HDL-C and triglyceride values as error-contaminated measurements of the underlying lipid variables and apply the proposed denoising procedure under a sequence of candidate measurement-error variances. Specifically, for the two-dimensional standardized lipid vector, we set $\bbSig = \sigma^{2}\bI$, where $\sigma^{2}=0$ corresponds to the conventional analysis based directly on the observed measurements. For each value of $\sigma^{2}$, we denoise HDL-C and triglyceride jointly within each CHD-status group and then refit the downstream logistic regression using the denoised lipid variables.
	
	\begin{figure}[!t]
		\centering
		\includegraphics[width=0.7\linewidth]{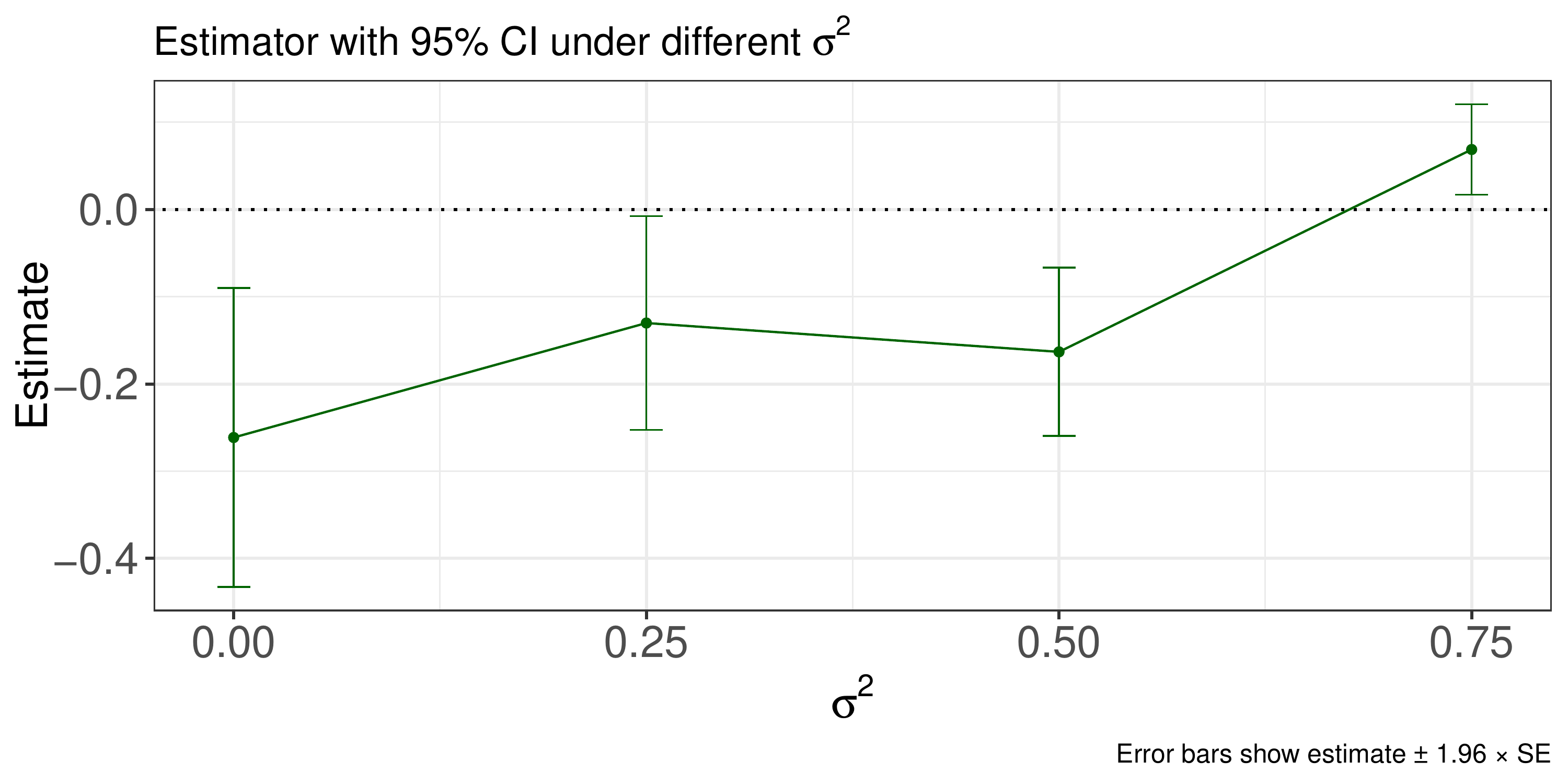}
		\caption{Sensitivity analysis for the adjusted HDL-C coefficient in the logistic regression of CHD status on HDL-C and triglyceride. For each candidate measurement-error variance $\sigma^{2}$ in $\bbSig=\sigma^{2}\bI$, the lipid variables are denoised using the proposed method and the downstream logistic regression is refitted. Points represent the estimated HDL-C coefficients, and vertical bars represent Wald-type 95\% confidence intervals computed as the estimate $\pm 1.96$ standard errors. The horizontal dotted line marks zero.}
		\label{fig:HDL_ME}
	\end{figure}
	
	Figure~\ref{fig:HDL_ME} shows that the conclusion from the conventional analysis is sensitive to the assumed measurement-error level. When $\sigma^{2}=0$, the estimated coefficient of HDL-C is negative and its confidence interval lies below zero, which is consistent with the familiar observational inverse association between HDL-C and CHD risk after adjustment for triglyceride. However, once measurement error is allowed, the estimated association becomes unstable: the magnitude is attenuated for some values of $\sigma^{2}$ and changes sign at $\sigma^{2}=0.75$. Thus, the apparent protective association between HDL-C and CHD is not robust to plausible perturbations of the error structure in the correlated lipid measurements. Similar to the last section, the confidence interval here serve as exploration tools for assessing the impact of measurement error rather than justified valid inference results.

	\section{Discussion}\label{sec:discussion}
	This paper proposes a novel statistical machine-learning approach to address the measurement error problem in both supervised and unsupervised data analysis tasks. Our main methodological contribution is an RKHS-based, complex-valued score matching method for training a diffusion model without requiring access to intermediate error-contaminated data. In addition, we provide a theoretical guarantee on the distributional gap between the denoised data and the true error-free data. Empirical results demonstrate the performance of the approach for multivariate data, including settings outside the theoretical assumptions.
	We also present applications to glucose monitoring and to the association between lipid measurements and coronary heart disease. These examples illustrate the sensitivity of the reported results to measurement-error modeling. The glucose-monitoring example highlights the need, in this AI era, to address not only predictive performance but also technological limitations regarding measurement reliability \citep{bent2020investigating}. Such considerations are particularly relevant for regulatory agencies (e.g., the FDA) seeking to understand both the potential and the limitations of emerging technologies. Our framework can be applied to various clinical scenarios, including biomarker discovery and the validation of novel therapies. Notably, it can generate denoised data that are readily analyzed with standard statistical techniques, thereby enhancing both interpretability and usability. Furthermore, this flexibility enables practitioners to select the most suitable models for their specific needs.
	
	A natural future research direction for our method is to validate finite-sample performance in supporting statistical inference for parameter estimation. For instance, one can construct confidence intervals for parameter estimates by using denoised data, which can provide asymptotically valid inferences under suitable conditions. Moreover, through additional resampling techniques, such as generalizing smoothing bootstrap methods to our RKHS-based denoising setting, finite-sample performance can be further improved. Extensions along these lines, including specialized predictive approaches (e.g., conformal prediction) that offer non-asymptotic guarantees, are critical for ensuring trustworthiness in machine learning applications.
	This paper is an initial step toward tackling measurement error using diffusion models. Specifically, we assume that the measurement error follows a normal distribution with a known variance. As a promising direction for future research, we plan to extend this method to accommodate other known or unknown error distributions using diffusion models.

	% In this paper, we propose to leverage the diffusion model developed in AI to tace the measurement error problem by generating denoised data whose distribution aligns well with the error-free ones. The main methodological contribution is a novel RKHS-based complex value score matching method to train the diffusion model without the access to observations of the intermediate error-contamited data. Moreover, we develop the theoretical guarantee of the distributional gap between the denoised data and the model-free ones. We conduct plenty of empirical studies over various tasks. The results show that the proposed diffusion-based method has promising performance in both simulations and practical digital healthcare settings. 
	% \par
	% We recognize this paper as an initialization to study measurement error with the diffusion model. In this paper, the error is assumed to be normal with a known variance. It is an interesting future direction to extend the method to handle measurement errors of other known or unknown distributions using diffusion models.

\clearpage
\appendix
\phantomsection
\addcontentsline{toc}{section}{Supplementary Material}
\noindent{\bfseries\Large Supplementary Material}

\setcounter{section}{0}
\setcounter{subsection}{0}
\setcounter{subsubsection}{0}
\setcounter{assumption}{0}
\setcounter{theorem}{0}
\setcounter{table}{0}
\setcounter{figure}{0}
\setcounter{example}{0}
\setcounter{proposition}{0}
\setcounter{lemma}{0}
\setcounter{equation}{0}
\setcounter{corollary}{0}
\setcounter{definition}{0}
\setcounter{remark}{0}
\setcounter{algorithm}{0}
\renewcommand{\thesection}{S\arabic{section}}
\renewcommand{\theassumption}{S\arabic{assumption}}
\renewcommand{\thetheorem}{S\arabic{theorem}}
\renewcommand{\thetable}{S\arabic{table}}
\renewcommand{\thefigure}{S\arabic{figure}}
\renewcommand{\theexample}{S\arabic{example}}
\renewcommand{\theproposition}{S\arabic{proposition}}
\renewcommand{\thelemma}{S\arabic{lemma}}
\renewcommand{\theequation}{S\arabic{equation}}
\providecommand{\theHsection}{}
\renewcommand{\theHsection}{supp.\arabic{section}}
\providecommand{\theHsubsection}{}
\renewcommand{\theHsubsection}{\theHsection.\arabic{subsection}}
\providecommand{\theHsubsubsection}{}
\renewcommand{\theHsubsubsection}{\theHsubsection.\arabic{subsubsection}}
\providecommand{\theHassumption}{}
\renewcommand{\theHassumption}{supp.\arabic{assumption}}
\providecommand{\theHtheorem}{}
\renewcommand{\theHtheorem}{supp.\arabic{theorem}}
\providecommand{\theHtable}{}
\renewcommand{\theHtable}{supp.\arabic{table}}
\providecommand{\theHfigure}{}
\renewcommand{\theHfigure}{supp.\arabic{figure}}
\providecommand{\theHexample}{}
\renewcommand{\theHexample}{supp.\arabic{example}}
\providecommand{\theHproposition}{}
\renewcommand{\theHproposition}{supp.\arabic{proposition}}
\providecommand{\theHlemma}{}
\renewcommand{\theHlemma}{supp.\arabic{lemma}}
\providecommand{\theHequation}{}
\renewcommand{\theHequation}{supp.\arabic{equation}}
\providecommand{\theHcorollary}{}
\renewcommand{\theHcorollary}{supp.\arabic{corollary}}
\providecommand{\theHdefinition}{}
\renewcommand{\theHdefinition}{supp.\arabic{definition}}
\providecommand{\theHremark}{}
\renewcommand{\theHremark}{supp.\arabic{remark}}
\providecommand{\theHalgorithm}{}
\renewcommand{\theHalgorithm}{supp.\arabic{algorithm}}

\section{Further Related Work}\label{app:related work}
	\paragraph{Measurement error.}
	
	Data science methods for error-contaminated data have a long history, dating back at least to econometrics (e.g., \cite{frisch1934statistical}; see \cite{annurev:/content/journals/10.1146/annurev-economics-080315-015058} for a review). Research in this area has remained active and has even gained momentum in recent years. One reason for this renewed interest is the push to relax many standard assumptions (such as linearity and independence), which has led to more complex methodological developments. Examples include fruitful results in binary regression \citep{chen2011marginal}, generalized linear models \citep{stefanski1987conditional}, instrumental variable models \citep{jiang2020measurement}, survival analysis \citep{prentice1982covariate, yan2015corrected}, and classification using neural networks \citep{yi2021improved, wang2022out, yi2023breaking}. 
	% Recently, some methods have been developed to address the measurement error problem in the high-dimensional context under specific models such as the linear regression model \citep{loh2012high} and the Poisson regression model \citep{jiang2023high}. 
	However, many of the existing methods \citep{loh2012high,jiang2023high} are highly model-dependent, making them hard to use when the form of the ground truth model is unknown.

	Another standard model-free method is deconvolution, which has been developed to estimate the density of the error-free variable when observations are error-contaminated \citep{fan1991optimal,hall2009deconvolution,belomestny2021density}, without requiring the structure of the model. Within this framework, kernel smoothing techniques are employed to estimate the density of the error-free data. However, a significant limitation hampers existing deconvolution techniques: the curse of dimensionality. This issue severely restricts their utility for tasks involving moderate- to high-dimensional variables.

	\vspace{-0.2in}
	\paragraph{Diffusion Model.} As introduced in Section \ref{sec:introduction}, diffusion models are generative models in AI that have achieved great success in image generation \citep{ho2020denoising,song2020score,rombach2022high,yi2023generalization,yi2024towards,wang2025improved,jiang2026fragile,jiang2026ofa}, video generation \citep{ma2024latte,opensora}, and text generation \citep{lin2023text,li2022diffusion}. Theoretically, this method originated from Langevin dynamics \citep{welling2011bayesian} for sampling from a target distribution by estimating the score function \citep{hyvarinen2005estimation,vincent2011connection}. Later, it was linked to VAEs \citep{kingma2013auto,sohl2015deep,song2020sliced} and widely applied in AI by modeling the score function with suitable deep neural networks \citep{ronneberger2015u,peebles2023scalable}. 
	\par
	In a diffusion model, the key idea is to transform an easily sampled distribution into a target distribution. Therefore, diffusion models are suitable for the measurement error problem studied in this paper when the target distribution is set to the denoised distribution. Despite this connection, to the best of our knowledge, diffusion models have not previously been applied to measurement error problems. We speculate that the main obstacle is correctly estimating the score function of the intermediate data. The closest topic in the diffusion-model literature is learning from ``corrupted data'' \citep{daras2024ambient,kawar2023gsure,daras2024consistent}. These works study learning from corrupted observations under their respective corruption models and training objectives. In particular, \citet{daras2024consistent} combine a generalized Tweedie identity with a consistency constraint to learn the clean distribution and noise levels below the observed noise level without clean training data. However, it lacks a statistical convergence guarantee.     
	
	\paragraph{Learning in Reproducing Kernel Hilbert Spaces.} 
	Statistical learning in Reproducing Kernel Hilbert Spaces (RKHS) has a rich history in both statistical science and machine learning. The foundations of this area can be traced back to non-linear prediction research by Wahba \cite{wahba1990spline}, which introduced the idea of learning complex regression functions without explicitly tuning the underlying smoothing hyperparameters. Examples of such methods include kernel ridge regression and support vector machines \citep{hofmann2008kernel}.

	Another direction focuses on kernel mean embeddings \citep{muandet2017kernel}, a technique that maps probability distributions into an RKHS as random elements of that space. Kernel mean embeddings are widely used for comparing distributions (e.g., in two-sample testing), clustering \citep{matabuena2022kernel}, evaluating generative models, and, more recently, in conformal prediction    \citep{matabuena2024conformal}. Novel RKHS-based methods have also been proposed for density estimation \citep{zhou2020nonparametric,kim2012robust,schuster2020kernel}.
	
	Despite the many ways in which the properties of RKHS can be leveraged for data analysis, to the best of our knowledge, there are currently no methods for efficiently handling measurement error while preserving a probability distribution perspective that exploits diffusion models. The intersection of these two fields has the potential to enable powerful, model-free approaches for addressing error in statistical analysis.

	\section{Regularity Conditions}\label{app:regularity conditions}
	
	\subsection{Standard Regularity Conditions}
	In this paper, our results are built on the following mild assumptions. They are standard ones in the existing diffusion-based literature \citep{song2021maximum}. 
	
	\begin{assumption}\label{ass:regularity}
		We make the following assumptions throughout this paper. 
		\begin{enumerate}
			\item For $t\in[0, 1]$, $p_t>0$ on $\bbR^d$, $p_t\in\cC^2(\bbR^d)$ \footnote{Twice continuously differentiable function.}, and $\mE[\|\bX_{t}\|^{2}] < \infty$. 
			\item $\lambda_{\max}(\bbSig)$ and $\lambda_{\max}(\bbH)$ are bounded from above, $\lambda_{\min}(\bbSig)$ and $\lambda_{\min}(\bbH)$ are bounded away from zero\footnote{Here $\lambda_{\max}(\cdot)$ and $\lambda_{\min}(\cdot)$ are respectively the largest and the smallest eigenvalues of a matrix.}. 
			\item There exists a positive constant $C$ such that for all $\bx$ and $t\in[0, 1]$, $\|\nabla\log{p_{t}}(\bx)\| \leq C(1 + \|\bx\|)$. 
            \item \textbf{Novikov's condition for the discretized sampler.}
Conditional on the fitted score functions, for the early-stopped sampler with $K\geq2$,
\begingroup\small
\[
\mE_{\tX}\left[\exp\left\{\frac12\sum_{k=0}^{K-2}\int_{k/K}^{(k+1)/K}
\left\|\bbSig^{1/2}\left(\nabla\log p_{1-t}(\tX_t)-\bs_{1-k/K}(\tX_{k/K};\htheta_{1-k/K})\right)\right\|^2dt\right\}\right]<\infty.
\]
\endgroup
		\end{enumerate}
	\end{assumption}
	We now discuss these assumptions. The first is a standard moment and density-continuity assumption for $p_{t}$. The second requires the upper and lower eigenvalues of $\bbSig$ and $\bbH$ to be bounded, which can be satisfied by properly choosing $\bbH$ or artificially adding extra Gaussian noise to observed data. The third condition is imposed to ensure that the reverse-time SDE \eqref{eq:reverse time sde} has a solution \citep{oksendal2013stochastic}.\footnote{This assumption is also implied by our Assumption \ref{ass:lip continuity}.} 
The Novikov condition above concerns the discretized drift and is imposed conditional on the fitted score functions. Here $\mE_{\tX}$ integrates over the exact reverse process in \eqref{eq:reverse time sde}, and the fitted score on each interval is evaluated at $\tX_{k/K}$ at forward time $1-k/K$, as required by \eqref{eq:euler solving sde}.
	\par
	
	\subsection{Continuity Condition}\label{app: continuity}
	Next, we give the continuity assumptions used to derive Theorem \ref{thm:sampling error}. 
	\begin{assumption}\label{ass:lip continuity}
		There exists a constant $L<\infty$, uniform in $t\in[0,1]$, satisfying \newline$\lambda_{\max}(\nabla^{2}_{\bx}\log{p_{t}}(\bx)) \leq L$ and, for all $\bx,\by\in\bbR^{d}$, the score and the fitted kernel-based model satisfy 
		\begin{equation*}
			\begin{aligned}
				\left\|\nabla\log p_t(\bx)-\nabla\log p_t(\by)\right\| &\leq L\|\bx-\by\|, \\
				\left\|\bs_t(\bx;\htheta_t)-\bs_t(\by;\htheta_t)\right\| &\leq L\|\bx-\by\|. 
			\end{aligned}
		\end{equation*}
	\end{assumption}
	This condition is widely used in the existing literature on error analysis for diffusion-based methods, e.g., \citep{chen2022sampling,lee2022convergence,chen2023probability,li2024towards}. The bounded spectral norm of $\nabla^{2}\log{p_{t}}(\bx)$ naturally implies the Lipschitz continuity of $\nabla\log p_t(\bx)$. In fact, it implicitly implies that $\bX_{0}$ has support $\bbR^{d}$. However, we found that our method can be applied to the other  $\bX_{0}$ with bounded support in practice, for example the uniform distribution in Section \ref{sec:simulation}. For the fitted model, the same deterministic $L$ is assumed to apply to all realizations of the fitted coefficients under consideration. This requires uniform deterministic bounds on the coefficients and the kernel derivatives; a realization-dependent coefficient bound gives only a realization-dependent Lipschitz constant.  
	\subsection{Conditions for Concentration}\label{app:conditions for concentration}
	In addition to the aforementioned conditions, we further require the following Bernstein-type conditions to derive our concentration results in Proposition \ref{prop:approximation error} and Theorem \ref{thm:gen error}.   
	\par
Before illustrating these conditions, we need some definitions to simplify the notations. For $\bx\in\bbR^d$, let
\[
\bzeta_{\bx,t}=2\mE_{\bxi}\left[\nabla_{\bx}\cK(\cdot,\bx+\mathrm{i}\sqrt{1-t}\bxi)\right]\in\cH_{\cK,\bbR}^{d},
\]
and define the linear operator
\[
\bGamma_{\bx,t}(\bs)=\mE_{\bxi}\left[\cK(\cdot,\bx-\mathrm{i}\sqrt{1-t}\bxi)\bs(\bx+\mathrm{i}\sqrt{1-t}\bxi)\right]
\]
on $\cH_{\cK,\bbR}^{d}$. Here the derivative in $\bzeta_{\bx,t}$ acts on the displayed real argument $\bx$. The reproducing property and conjugation symmetry give, for $\bu,\bv\in\cH_{\cK,\bbR}^{d}$,
\begin{equation*}
\begin{aligned}
\langle\bu,\bGamma_{\bx,t}\bv\rangle_{\cH_{\cK}^{d}}
    &=\mE_{\bxi}\left[\bu(\bx+\mathrm{i}\sqrt{1-t}\bxi)^{\top}\bv(\bx+\mathrm{i}\sqrt{1-t}\bxi)\right],\\
\langle\bzeta_{\bx,t},\bu\rangle_{\cH_{\cK}^{d}}
    &=2\mE_{\bxi}\left[\tr\{\nabla\bu(\bx+\mathrm{i}\sqrt{1-t}\bxi)\}\right].
\end{aligned}
\end{equation*}
All expectations in these definitions are Bochner integrals. A sufficient condition for their existence, uniformly in $t$, is $8\lambda_{\max}(\bbSig^{1/2}\bbH\bbSig^{1/2})<1$, which holds for the recommended $\bbH=(8\bbSig+\epsilon\bI)^{-1}$ with fixed $\epsilon>0$.
	\begin{assumption}\label{ass:subexp}
For every $b>0$ and $t\in[0,1]$, there exists a real vector-valued function $\bbf_{t,b}\in(L^2_{\bX_t})^d$ satisfying
\begin{equation}\label{eq:score bias}
\nabla\log p_t(\bx)=\mE_{\bX_t}[\cK(\bx,\bX_t)\bbf_{t,b}(\bX_t)]+\bkappa_{t,b}(\bx),
\end{equation}
with $\mE_{\bX_t}[\|\bbf_{t,b}(\bX_t)\|^2]\leq b$ and
$\sup_{t\in[0,1]}\mE_{\bX_t}[\|\bkappa_{t,b}(\bX_t)\|^2]\leq\kappa(b)$, where $\kappa(b)\to0$ as $b\to\infty$. Moreover,
\begingroup\small
\begin{equation*}
\begin{aligned}
&\mE_{\bX_t^{\prime},\bX_1^{\prime}}\left[\left\|
\bbf_{t,b}(\bX_t^{\prime})\cK_t(\bX_t,\bX_1^{\prime})
-\nabla\log p_t(\bX_t)+\bkappa_{t,b}(\bX_t)
\right\|_{L^2_{\bX_t}}^k\right]\leq\frac{k!\sigma_0(b)^2 H_0(b)^{k-2}}{2},\qquad k\geq2,
\end{aligned}
\end{equation*}
\endgroup
for finite constants $\sigma_0(b)$ and $H_0(b)$, uniform in $t$, where $\{\bX_t^{\prime}:t\in[0,1]\}$ is an independent copy of the forward process $\{\bX_t:t\in[0,1]\}$, and $\|\cdot\|_{L^2_{\bX_t}}$ is the $L^2(P_{\bX_t})$ norm defined in Section \ref{sec:notations}, evaluated at an independent $\bX_t$.
	\end{assumption}
	This assumption implicitly requires that the original data $\bX_{0}$ do not lie on a lower-dimensional manifold in $\bbR^{d}$, since such a distribution results in $\|\nabla\log{p_{t}}(\bX_{t})\|_{L_{\bX_{t}}^{2}}\to \infty$ when $t\to 0$ \citep{chen2022sampling,chen2023score,oko2023diffusion}. Intuitively, $p_{t}(\bx)$ decreases from a positive value to zero as $t\to 0$ for those $\bx$ not in the support of $\bX_{0}$. 

    Notably, for the Gaussian kernel used in this paper, we distinguish the complex RKHS from its closed real Hilbert subspace
	\[
	\cH_{\cK,\bbR}
	=\{f\in\cH_{\cK}:f(\overline{\bz})=\overline{f(\bz)}\text{ for all }\bz\in\bbC^d\},
	\qquad \cH_{\cK,\bbR}^{d}=(\cH_{\cK,\bbR})^{d}.
	\]
	These functions are real-valued on $\bbR^d$. On this subspace, we retain $\langle\cdot,\cdot\rangle_{\cH_{\cK}^{d}}$ and $\|\cdot\|_{\cH_{\cK}^{d}}$ for the inherited inner product and norm. The score models have this conjugation symmetry; when their RKHS norm is finite, they belong to $\cH_{\cK,\bbR}^{d}$. 
    % The operator arguments in the Supplementary Material use this real subspace. The same subscript convention applies to the kernel $\cK_t$ defined below.
    
	\begin{assumption}\label{ass:bernstein}
		For $t\in [0, 1]$, it holds that 
		\begin{equation*}\mE_{\bX_{1}}\left[\left\|\bzeta_{\bX_{1},t} - \mE_{\bX_{1}}[\bzeta_{\bX_{1},t}]\right\|_{\cH_{\cK}^{d}}^{k}\right] \leq \frac{k!\sigma_{1}^{2}H_{1}^{k - 2}}{2}
		\end{equation*}
		and
		\begin{equation*}
			\mE_{\bX_{1}}\left[\left\|\bGamma_{\bX_{1},t} - \mE_{\bX_{1}}[\bGamma_{\bX_{1},t}]\right\|_{\rm HS}^{k}\right] \leq \frac{k!\sigma_{2}^{2}H_{2}^{k - 2}}{2}
		\end{equation*}
		for any $k\geq2$, and finite constants $\sigma_1,H_1,\sigma_2,H_2$ uniform in $t$ and independent of $n,K,b$; they may depend on $d$ and the fixed kernel. 
	\end{assumption}
	Assumptions \ref{ass:subexp} and \ref{ass:bernstein} are similar to the conditions required for the concentration results in \citet{de2005learning,smale2003estimating,smale2007learning,zhou2020nonparametric} and describe the smoothness of our kernel function. Such conditions are commonly used in the RKHS literature. The $\sigma_{1}^{2}$ and $\sigma^{2}_{2}$ here can be explained as the ``variance'' in Bernstein's condition \citep{duchi2016lecture}, which can scale exponentially with $d$ due to the structure of the kernel function. Thus, as discussed in the main text, our bound in Theorem \ref{thm:gen error} can have exponential dependence on $d$ in the worst case.

	\section{Discussion on Reverse-Time SDE \eqref{eq:reverse time sde}}\label{app:discussion on reverse time sde}
	\begin{proposition}\label{pro:density}
		Under Assumption \ref{ass:regularity}, let $\bX_t$ solve \eqref{eq:sode} with initial law $P_{\bX_0}$, and let $\tX_t$ solve \eqref{eq:reverse time sde} with $\tX_0\sim P_{\bX_1}$. Then $P_{\tX_t}=P_{\bX_{1-t}}$ for every $t\in[0,1]$.    
	\end{proposition}
	\begin{proof}
		Write $\sigma_{ij}=(\bbSig)_{ij}$. By the Fokker--Planck equation \citep{oksendal2013stochastic}, we know that the density of SDE \eqref{eq:sode} follows the PDE 
		\begin{equation*}
			\frac{\partial}{\partial{t}}p_{t}(\bx) = \frac{1}{2}\sum_{i, j=1}^{d} \sigma_{ij}\frac{\partial^{2}p_{t}(\bx)}{\partial x_i\partial x_j} = \nabla\cdot\left(p_{t}\frac{\bbSig\nabla\log{p_{t}(\bx)}}{2}\right).
		\end{equation*}
		Thus we can conclude that $\bX_{t}$ has the same density as the following probability-flow ODE $\bx_{t}$ such that 
		\begin{equation*}
			d\bx_{t} = -\frac{1}{2}\bbSig\nabla\log{p_{t}(\bx_{t})}dt 
		\end{equation*}
		with $\bx_{0} = \bX_{0}$. Define $\tilde{\bx}_t=\bx_{1-t}$. This reverse-time process satisfies
		\begin{equation}\label{eq:ode}
			d\tilde{\bx}_{t} = \frac{1}{2}\bbSig\nabla\log{p_{1 - t}(\tilde{\bx}_{t})}dt,
		\end{equation}
		with initial law $P_{\bX_1}$ and density $p_{1-t}$. Its continuity equation can be written as  
		\begin{equation*}
			\begin{aligned}
				\frac{\partial}{\partial{t}}p_{1 - t}(\bx) &= -\nabla\cdot\left(p_{1 - t}(\bx)\frac{\bbSig\nabla\log{p_{1 - t}(\bx)}}{2}\right) \\
				&= -\nabla\cdot\left(p_{1 - t}(\bx)\bbSig\nabla\log{p_{1 - t}(\bx)}\right) + \frac{1}{2}\sum_{i, j=1}^{d} \sigma_{ij}\frac{\partial^{2}p_{1 - t}(\bx)}{\partial x_i\partial x_j},
			\end{aligned}
		\end{equation*}
		which is exactly the Fokker--Planck equation of reverse SDE \eqref{eq:reverse time sde}. Thus, we can conclude that $\bX_t\overset{d}{=}\bx_t\overset{d}{=}\tilde{\bx}_{1-t}\overset{d}{=}\tX_{1-t}$, which proves the conclusion in Section \ref{subsec:Estimating Score Function}.   
		It indicates that $P_{\bX_{1 - t}} = P_{\tX_{t}}$ due to the Fokker-Planck equation \citep{oksendal2013stochastic}. 
	\end{proof}
	%\tweedie*
	%\begin{proof}
	%	First, we have 
	%	\begin{equation*}\label{eq:conditional expectation}
		%		
		%			\mE\left[\frac{\bx_{t} - \bx_{0}}{t}\mid \bx_{t}\right] = \frac{\bx_{t}}{t} - \frac{1}{t}\mE[\bx_{0}\mid \bx_{t}]. 
		%	\end{equation*}
	%	On the other hand, $\bx_{t}\mid \bx_{0}\sim \cN(\bx_{0} + t\bmu, t^{2}\Sigma)$. Then, it holds 
	%	\begin{equation*}\label{eq:score equivalence}
		%		
		%			\nabla_{\bx}\log{P_{t}(\bx_{t}\mid \bx_{0})} = -\frac{1}{t^{2}}\Sigma^{-1}[\bx_{t} - (\bx_{0} + t\bmu)].
		%	\end{equation*}
	%	Thus, 
	%	\begin{equation*}
		%		
		%		\begin{aligned}
			%			\mE[\bx_{0}\mid \bx_{t}] & = \int_{\bbR^{d}}\bx_{0}P_{0\mid t}(\bx_{0}\mid \bx_{t})d\bx_{0} \\
			%			& = \int_{\bbR^{d}}\bx_{0}\frac{p_{t\mid 0}(\bx_{t}\mid \bx_{0})P_{0}(\bx_{0})}{P_{t}(\bx_{t})}d\bx_{0} \\
			%			& = \int_{\bbR^{d}}\left[(\bx_{t} - t\bmu) + t^{2}\Sigma\nabla_{\bx}\log{p_{t\mid 0}(\bx_{t}\mid \bx_{0})}\right]\frac{p_{t\mid 0}(\bx_{t}\mid \bx_{0})P_{0}(\bx_{0})}{P_{t}(\bx_{t})}d\bx_{0} \\ 
			%			& = (\bx_{t} - t\bmu) + t^{2}\int_{\bbR^{d}}\Sigma\nabla_{\bx}p_{t\mid 0}(\bx_{t}\mid \bx_{0})\frac{P_{0}(\bx_{0})}{P_{t}(\bx_{t})}d\bx_{0} \\
			%			& = (\bx_{t} - t\bmu) + t^{2}\Sigma\nabla_{\bx}\log{P_{t}(\bx_{t})}. 
			%		\end{aligned}
		%	\end{equation*}
	%	Plugging this into \eqref{eq:conditional expectation}, we prove the result. 
	%\end{proof}
	\section{Proofs of Results in Section \ref{sec:DEP}}\label{app:proofs in Section dep}
	
	%\scorefunction*
	%\begin{proof}
	%	By simple algebra, we have 
	%	\begin{equation*}
		%		
		%		\begin{aligned}
			%			&\mE_{\bx_{t}}\left[\left\|\bs_{\btheta}(\bx_{t}, t) - \nabla_{\bx}\log{p_{t}}(\bx_{t})\right\|^{2}\right]\\
			%			&= \mE_{\bx_{t}}\left[\left\|\bs_{\btheta}(\bx_{t}, t)\right\|^{2}\right] - 2\mE\left[\left\langle\bs_{\btheta}(\bx_{t}, t), \nabla_{\bx}\log{p_{t}}(\bx_{t})\right\rangle\right] + 2\mE\left[\|\nabla_{\bx}\log{p_{t}}(\bx_{t})\|^{2}\right].
			%		\end{aligned}
		%	\end{equation*}
	%	Note that, by taking integral by part,  
	%	\begin{equation*}
		%		
		%		\begin{aligned}
			%			\mE\left[\left\langle\bs_{\btheta}(\bx_{t}, t), \nabla_{\bx}\log{p_{t}}(\bx_{t})\right\rangle\right]
			%			&= \sum_{i=1}^{d}\int p_{t}(\bx_{t})\bs_{\btheta}(\bx_{t}, t)\frac{\partial}{\partial{\bx_{i}}}\log{p_{t}}(\bx_{t})d\bx_{t}\\
			%			&= -\mE\left[\tr\{\nabla_{\bx}\bs_{\btheta}(\bx_{t}, t)\}\right],
			%		\end{aligned}
		%	\end{equation*}
	%	and plugging this into the above equation, we get the conclusion. 
	%\end{proof}
	\subsection{Proofs of Results in Section \ref{subsec:Estimating Score Function}}\label{app:proof in estimating score function}
	{\bf Restatement of Corollary \ref{coro:debias}.}
	\emph{Let $\bxi \sim \cN(\bzero, \bbSig)$ be a random vector independent of $\bX_{0}$ and $\beps$. If $\bs_t(\cdot;\btheta_t)$ is analytic on $\bbC^d$, then we have  
\begingroup\small
\begin{equation*}
			\begin{aligned}
				& \mE_{\bX_{t}}\left[\left\|\bs_{t}(\bX_{t}; \btheta_{t})\right\|^{2}\right] + 2\tr\left\{\mE_{\bX_{t}}\left[\nabla\bs_{t}(\bX_{t}; \btheta_{t})\right]\right\} = \\
				& \mE_{\bX_{1}, \bxi}\left[\bs_t(\bX_{1}+\mathrm{i}\sqrt{1-t}\bxi;\btheta_t)^\top\bs_t(\bX_{1}+\mathrm{i}\sqrt{1-t}\bxi;\btheta_t)\right] + 2\tr\left\{\mE_{\bX_{1}, \bxi}\left[\nabla\bs_{t}(\bX_{1} + \mathrm{i}\sqrt{1 - t}\bxi; \btheta_{t})\right]\right\}, 
			\end{aligned}
		\end{equation*}
\endgroup
		where $\mathrm{i} = \sqrt{-1}$ is the imaginary unit.}
\begin{proof}
For a fixed $t\in[0,1]$, we may represent the forward-process pair jointly as $\bX_t=\bX_0+\sqrt{t}\bxi_1$ and $\bX_{1} = \bX_{0} + \sqrt{t}\bxi_{1} + \sqrt{1 - t}\bxi_{2}$, where $\bxi_1$ and $\bxi_2$ are independent $\cN(\bzero,\bbSig)$ vectors, independent of $\bX_0$ and $\bxi$. Using the analytic function $\bs_t(\bx;\btheta_t)^{\top}\bs_t(\bx;\btheta_t)$ as an example, we have
\begin{equation*}
\begin{aligned}
&\mE\left[\bs_t(\bX_1+\mathrm{i}\sqrt{1-t}\bxi;\btheta_t)^{\top}\bs_t(\bX_1+\mathrm{i}\sqrt{1-t}\bxi;\btheta_t)\right]\\
&\quad=\mE\left[\mE\left[\bs_t(\bX_1+\mathrm{i}\sqrt{1-t}\bxi;\btheta_t)^{\top}\bs_t(\bX_1+\mathrm{i}\sqrt{1-t}\bxi;\btheta_t)\mid\bX_0,\bxi_1\right]\right]\\
&\quad\overset{a}{=}\mE\left[\left\|\bs_t(\bX_0+\sqrt{t}\bxi_1;\btheta_t)\right\|^2\right]
=\mE\left[\left\|\bs_t(\bX_t;\btheta_t)\right\|^2\right],
\end{aligned}
\end{equation*}
where the equality marked by $a$ follows from Lemma \ref{lem: complex debias} and the independence among $\bX_0,\bxi,\bxi_1,\bxi_2$. This proves the desired conclusion. The proof similarly extends to $\tr\{\nabla\bs_t(\bx;\btheta_t)\}$.
\end{proof}
	
	\subsection{Proofs of Results in Section \ref{sec:kernel-based function}}\label{app:proofs in kernel based function}

	{\bf Restatement of Proposition \ref{prop:approximation error}.}
\emph{Under Assumptions \ref{ass:lip continuity} and \ref{ass:subexp} in Supplementary Material, fix $b>0$ and $t\in[0,1]$, and define
\begin{equation*}
    \delta_{0m}(\eta;b)=2\left(\frac{H_0(b)}{m}+\frac{\sigma_0(b)}{\sqrt m}\right)\log\left(\frac{2}{\eta}\right),
\end{equation*}
where $\sigma_0(b)$ and $H_0(b)$ are the constants in Assumption \ref{ass:subexp}, uniform in $t$. Then,
\begin{equation*}
\begin{aligned}
    &\inf_{\bs_t(\cdot;\btheta_t)\in\cH_{\cK_t,m}^{d}}
    \mE_{\bX_t}\left[\left\|\bs_t(\bX_t;\btheta_t)-\nabla\log p_t(\bX_t)\right\|^2\right]
    \leq 2\delta_{0m}(\eta;b)^2+\frac{2\sigma_0(b)^2}{m}+2\kappa(b)
\end{aligned}
\end{equation*}
holds with probability at least $1-\eta$ ($0<\eta<1$), where the randomness comes from $\cH_{\cK_t,m}^{d}$.}
\begin{proof}
Fix $b>0$ and $t\in[0,1]$. By Lemma \ref{lem: complex debias}, conditional on $\bX_t^{\prime}$, and \eqref{eq:score bias},
\begin{equation*}
\mE_{\bX_t^{\prime},\bX_1^{\prime}}[\bbf_{t,b}(\bX_t^{\prime})\cK_t(\bx,\bX_1^{\prime})]
=\nabla\log p_t(\bx)-\bkappa_{t,b}(\bx).
\end{equation*}
The functions indexed by the independent forward-process copies $(\bX_t^{(i)},\bX_1^{(i)})$ are therefore centered after subtracting this mean. Applying the Hilbert-space Bernstein inequality (Proposition 23 of \citealp{bauer2007regularization}) in $(L^2_{\bX_t})^d$ under Assumption \ref{ass:subexp} gives
\begin{equation*}
\left\|\frac1m\sum_{i=1}^m\bbf_{t,b}(\bX_t^{(i)})\cK_t(\bX_t,\bX_1^{(i)})
-\nabla\log p_t(\bX_t)+\bkappa_{t,b}(\bX_t)\right\|_{L^2_{\bX_t}}
\leq\delta_{0m}(\eta;b)
\end{equation*}
with probability at least $1-\eta$. The triangle inequality and the bound on the remainder then imply
\begin{equation*}
\begin{aligned}
&\left\|\frac1m\sum_{i=1}^m\bbf_{t,b}(\bX_t^{(i)})\cK_t(\bX_t,\bX_1^{(i)})
-\nabla\log p_t(\bX_t)\right\|_{L^2_{\bX_t}}^2\\
&\qquad\leq2\delta_{0m}(\eta;b)^2+2\kappa(b)
\leq2\delta_{0m}(\eta;b)^2+\frac{2\sigma_0(b)^2}{m}+2\kappa(b).
\end{aligned}
\end{equation*}
The displayed kernel expansion belongs to $\cH_{\cK_t,m}^{d}$, so the infimum satisfies the same bound. This event depends only on the observed kernel centers and has probability at least $1-\eta$.
\end{proof}

	Next, we prove Proposition \ref{pro:explicit formulation}. We first derive \eqref{eq:basis s}. From the definition of $\cK_{t}(\bx_{1}, \bx_{2})$ in \eqref{eq:basis s} and basic algebra, we have 
\begingroup\small\thinmuskip=2mu
\begin{equation*}
		\begin{aligned}
			& \cK_{t}(\bx_{1}, \bx_{2}) 
			= \int_{\bbR^{d}}(2\pi)^{-d/2}|\bbSig|^{-1/2}\exp\left\{-(\bx_{1} - \bx_{2} + \mathrm{i}\sqrt{1 - t}\bxi)^{\top}\bbH(\bx_{1} - \bx_{2} + \mathrm{i}\sqrt{1 - t}\bxi) - \bxi^{\top}\bbOmega\bxi\right\}d\bxi \\
			& = \int_{\bbR^{d}}(2\pi)^{-d/2}|\bbSig|^{-1/2}\exp\left\{-\left[\bxi + \mathrm{i}\sqrt{1 - t}\bbOmega_{t}^{-1}\bbH(\bx_{1} - \bx_{2})\right]^{\top}\bbOmega_{t}\left[\bxi + \mathrm{i}\sqrt{1 - t}\bbOmega_{t}^{-1}\bbH(\bx_{1} - \bx_{2})\right]\right\} \\
			& \cdot \exp\left\{-(1 - t)(\bx_{1} - \bx_{2})^{\top}\bbH\bbOmega_{t}^{-1}\bbH(\bx_{1} - \bx_{2}) - (\bx_{1} - \bx_{2})^{\top}\bbH(\bx_{1} - \bx_{2})\right\}d\bxi \\
			& = \sqrt{|\bbOmega_{t}|^{-1}|\bbOmega|}\exp\left\{-(1 - t)(\bx_{1} - \bx_{2})^{\top}\bbH\bbOmega_{t}^{-1}\bbH(\bx_{1} - \bx_{2}) - (\bx_{1} - \bx_{2})^{\top}\bbH(\bx_{1} - \bx_{2})\right\}.
		\end{aligned}
	\end{equation*}
\endgroup
	Due to the matrix equality $(1 - t)\bbH\bbOmega^{-1}_{t}\bbH + \bbH = \bbH_{t}$, we obtain \eqref{eq:basis s}. Similar to the derivation of \eqref{eq:basis s}, we can prove Proposition \ref{pro:explicit formulation}. 
	
	\noindent
	{\bf Restatement of Proposition \ref{pro:explicit formulation}.}
	\emph{Condition on the first $m$ observations used as kernel centers, fix $\btheta_t$, and choose $\bbH$ independently of the remaining $n-m$ observations. For $\bs_t$ in \eqref{eq: kernel model}, provided the matrices $\bbOmega_t$, $\bbOmega_t^{(1)}$, and $\bbOmega_t^{(2)}$ below are positive definite, an unbiased estimator of \eqref{eq:new complex objective} is 
		\begin{equation*}
			\begin{aligned}
				\tr\left\{\btheta_{t}^{\top}\BK_{t}^{(1)} + \btheta_{t}^{\top}\BK_{t}^{(2)}\btheta_{t}\right\},
			\end{aligned}
		\end{equation*}
		where $\BK_{t}^{(1)}$ is an $m\times d$ matrix whose $(i, l)$-th element is $2(n - m)^{-1}\sum_{k = m + 1}^{n} \cK_{t}^{(1, l)}(\bX_{1}^{(k)}; \bX_{1}^{(i)})$, $\BK_{t}^{(2)}$ is an $m\times m$ matrix whose $(i,j)$-th element is $(n - m)^{-1}\sum_{k = m + 1}^{n}\cK_{t}^{(2)}(\bX_{1}^{(k)}; \bX_{1}^{(i)}, \bX_{1}^{(j)})$,  
		\begin{equation*}
			\begin{aligned}
				\cK_{t}^{(1, l)}(\bx; \bx_{1}) 
				& = -2\frac{|\bbOmega|}{\sqrt{|\bbOmega_{t}||\bbOmega_{t}^{(1)}|}}\be_{l}^{\top}\bbH_{t}^{(1)}(\bx - \bx_{1})\exp\left\{ - (\bx - \bx_{1})^{\top}\bbH_{t}^{(1)}(\bx - \bx_{1})\right\},
			\end{aligned}
		\end{equation*}
		and 
		\begin{equation*}
			\begin{aligned}
				\cK_{t}^{(2)}(\bx; \bx_{1}, \bx_{2}) 
				& = \frac{|\bbOmega|^{\frac{3}{2}}}{|\bbOmega_{t}||\bbOmega_{t}^{(2)}|^{\frac{1}{2}}}\exp\left\{ - (\bx_{1} + \bx_{2} - 2\bx)^{\top}\bbH_{t}^{(2)}(\bx_{1} + \bx_{2} - 2\bx)\right\}\\
				&\quad \times 
				\exp\left\{ - (\bx_{1} - \bx)^{\top}\bbH_{t}(\bx_{1} - \bx)\right\}
				\times \exp\left\{ - (\bx_{2} - \bx)^{\top}\bbH_{t}(\bx_{2} - \bx)\right\},
			\end{aligned}
		\end{equation*}
		for any $\bx,\bx_1,\bx_2\in\bbR^d$. Here $\be_{l}$ is the $l$-th basis vector in the $d$-dimensional space, $\bbOmega_{t}^{(1)} = \bbOmega - (1 - t)\bbH_{t}$,
		$\bbH_{t}^{(1)} = (1 - t)\bbH_{t}\bbOmega_{t}^{(1)-1}\bbH_{t} + \bbH_{t}$, $\bbOmega_{t}^{(2)} = \bbOmega - 2(1 - t)\bbH_{t}$ and $\bbH_{t}^{(2)} = (1 - t)\bbH_{t}\bbOmega_{t}^{(2)-1}\bbH_{t}$. 
	}
	\vspace{-0.4in}
	\begin{proof}
		Condition on the kernel centers and fix $\btheta_t$; all expectations below use a fresh evaluation point $\bX_1$ and independent $\bxi$. According to the formulation \eqref{eq:new complex objective}, it suffices to find the unbiased estimates for 
		\begin{equation*}
			\mE\left[\bs_{t}(\bX_{1}+\mathrm{i}\sqrt{1-t}\bxi; \btheta_{t})^{\top}\bs_{t}(\bX_{1}+\mathrm{i}\sqrt{1-t}\bxi; \btheta_{t})\right]\quad \text{and}\quad
			\mE\left[\nabla\bs_{t}(\bX_{1}+\mathrm{i}\sqrt{1-t}\bxi; \btheta_{t})\right].
		\end{equation*}
		We first compute 
		\begin{equation*}
			\begin{aligned}
				&\mE\left[\bs_{t}(\bX_{1}+\mathrm{i}\sqrt{1-t}\bxi; \btheta_{t})^{\top}\bs_{t}(\bX_{1}+\mathrm{i}\sqrt{1-t}\bxi; \btheta_{t})\right]\\
				& = \sum\limits_{l=1}^{d}\btheta_{t, l}^{\top}\mE\left[\bcK_{t}(\bX_{1} + \mathrm{i}\sqrt{1 - t}\bxi; \sD_{m})\bcK_{t}(\bX_{1} + \mathrm{i}\sqrt{1 - t}\bxi; \sD_{m})^{\top}\right]\btheta_{t, l} \\
				& = \tr\left\{\btheta_{t}^{\top}\mE\left[\bcK_{t}(\bX_{1} + \mathrm{i}\sqrt{1 - t}\bxi;\sD_{m})\bcK_{t}(\bX_{1} + \mathrm{i}\sqrt{1 - t}\bxi;\sD_{m})^{\top}\right]\btheta_{t}\right\},
			\end{aligned}
		\end{equation*}
		where $\sD_{m} = \{\bX_{1}^{(i)}\}_{i = 1}^{m}$ and $\bcK_{t}(\bx; \sD_{m})$ is the vector $(\cK_{t}(\bx, \bX_{1}^{(1)}), \cdots, \cK_{t}(\bx, \bX_{1}^{(m)}))^{\top}$. Direct calculation shows that
		\begin{equation*}
			\begin{aligned}
				\mE[\cK_{t}(\bX_{1} + \mathrm{i}\sqrt{1 - t}\bxi, \bX_{1}^{(i)})\cK_{t}(\bX_{1} + \mathrm{i}\sqrt{1 - t}\bxi, \bX_{1}^{(j)})] 
				& = \mE_{\bX_{1}}[\cK_{t}^{(2)}(\bX_{1}; \bX_{1}^{(i)}, \bX_{1}^{(j)})],
			\end{aligned}
		\end{equation*}
		for $1\leq i,j\leq m$, which can be unbiasedly estimated by $(n - m)^{-1}\sum_{k = m + 1}^{n}\cK_{t}^{(2)}(\bX_{1}^{(k)}; \bX_{1}^{(i)}, \bX_{1}^{(j)})$. Furthermore, we have
\begin{equation*}
			\begin{aligned}
				2\tr\left\{\mE\left[\nabla\bs_{t}(\bX_{1} + \mathrm{i}\sqrt{1 - t}\bxi; \btheta_{t})\right]\right\} = 2 \tr\left[\btheta_{t}^{\top}\mE\left\{\nabla \bcK_{t}(\bX_{1}+\mathrm{i}\sqrt{1-t}\bxi; \sD_{m})\right\}\right].
			\end{aligned}
		\end{equation*}
		Straightforward calculation shows that
\begin{equation*}
			\mE\left[\frac{\partial}{\partial x_{l}} \cK_{t}(\bX_{1} + \mathrm{i}\sqrt{1 - t}\bxi, \bX_{1}^{(i)})\right] = \mE\left[\cK_{t}^{(1, l)}(\bX_{1}; \bX_{1}^{(i)})\right],
		\end{equation*}
		Thus twice this expectation, as required by the score-matching objective, is unbiasedly estimated by
		\begin{equation*}
			2(n - m)^{-1}\sum_{k = m + 1}^{n} \cK_{t}^{(1, l)}(\bX_{1}^{(k)}; \bX_{1}^{(i)}).
		\end{equation*}
		This completes the proof.
	\end{proof}
	
	\section{Proofs of Results in Section \ref{sec:kl divergence error}}\label{app:proofs in section sampling with reverse-sde}
	%\equivalence*
	%\begin{proof}
	%	Let us first consider the reverse-time ODE 
	%	\begin{equation*}
		%		
		%		\frac{d\tX_{t}}{dt} = \bv_{\bx}(\tX_{t}, 1 - t) = \tv_{\bx}(\tX_{t}, t) = -(\bmu + (1 - t)\Sigma\nabla_{\bx}\log{P_{1 - t}}(\tX_{t})), 
		%	\end{equation*}
	%	where we use $\tv_{\bx}(\tX_{t}, t)$ to denote $-\bv_{\bx}(\tX_{t}, 1 - t)$. Then according to the Fokker-Planck equation, the density of $\tX_{t}$ satisfies the following PDE
	%	\begin{equation*}
		%		
		%		\begin{aligned}
			%			\frac{\partial}{\partial{t}}P_{\tX_{t}}(t, \tX_{t}) & = -\nabla\cdot(\tv_{\bx}(\tX_{t}, t)P_{\tX_{t}}(t, \tX_{t})) \\
			%			& = -\nabla\cdot\left[-\left(\bmu + (1 - t)\Sigma\nabla_{\bx}\log{P_{1 - t}}(\tX_{t})\right) + (1 - t)\Sigma\nabla_{\bx}\log{P_{1 - t}}(\tX_{t}) - (1 - t)\Sigma\nabla_{\bx}\log{P_{1 - t}}(\tX_{t})\right] \\
			%			& = -\nabla\cdot\left[-\left(\bmu + 2(1 - t)\Sigma\nabla_{\bx}\log{P_{1 - t}}(\tX_{t})\right)\right] + (1 - t)\Sigma\nabla_{\bx}\cdot\nabla_{\bx}\log{P_{1 - t}}(\tX_{t}) \\
			%			& = -\nabla\cdot\left[-\left(\bmu + 2(1 - t)\Sigma\nabla_{\bx}\log{P_{1 - t}}(\tX_{t})\right)\right] + (1 - t)\Sigma\Delta_{\bx}\cdot P_{\tX_{t}}(\tX_{t}).
			%		\end{aligned}
		%	\end{equation*}
	%	Thus, again from the results of Fokker-Planck equation, the density of $\tX_{t}$ is equivalent to the SDE
	%	\begin{equation*}\label{eq:reverse sde}
		%		
		%		\begin{dcases}
			%			& d\tX_{t} = -\left(\bmu + 2(1 - t)\Sigma\nabla_{\bx}\log{P_{1 - t}}(\tX_{t})\right)dt + \sqrt{2(1 - t)}\Sigma^{\frac{1}{2}}dW_{t}; \\
			%			& P_{\tX_{0}} = P_{1}, P_{\tX_{1}} = P_{0}.
			%		\end{dcases}
		%	\end{equation*}
	%	The we prove our result. 
	%\end{proof}
	In this section, we prove Theorem \ref{thm:sampling error} using arguments based on \citep{lee2022convergence,chen2022sampling,chen2023probability}. To do so, we need several lemmas. First, we rewrite the sampling process of data $\hX_{t}$ as the following SDE 
	\begingroup\small
\begin{equation}\label{eq:discrete reverse sde}
		\begin{dcases}
			& d\hX_{t} = \bbSig\bs_{1 - k/K}\left(\hX_{k/K}; \htheta_{1 - k/K}\right)dt + \bbSig^{\frac{1}{2}}d\widetilde{\bW}_{t}; \qquad \frac{k}{K} \leq t < \frac{k + 1}{K}; \qquad k = 0, \cdots, K - 2; \\
			& P_{\hX_{0}} = P_{\bX_{1}}.
		\end{dcases}
	\end{equation}
\endgroup
	Using this SDE, we can bound the gap between $P_{\hX_{1}}$ and the desired $P_{\bX_{0}} = P_{\tX_{1}}$ by applying Girsanov's theorem \citep{oksendal2013stochastic}. Here, please note that the denoising process is early stopped on $\hX_{1 - 1/K}\approx \bX_{1/K}$ as in \citep{chen2023probability}. The endpoint discrepancy between $P_{\bX_{1/K}}$ and $P_{\bX_0}$ is bounded separately below.
	Throughout this section, all expectations and generated laws are conditional on the fitted score functions, and $\mE_{\tX}$ denotes expectation under the exact reverse process. On each interval $[k/K,(k+1)/K)$, we use the original fitted score $\bs_{1-k/K}(\cdot;\htheta_{1-k/K})$ from \eqref{eq:euler solving sde}. In the change of measure below, this drift is evaluated along the reference path at $\tX_{k/K}$. The following lemma is a direct consequence of Girsanov's theorem \citep{oksendal2013stochastic}.
	\begin{lemma}[Girsanov's Theorem]\label{lem:girsanov}
		Let $P_{\tX_{[0, 1 - 1/K]}}$ and $P_{\hX_{[0, 1 - 1/K]}}$ be the probability measures on the path space $\cC([0, 1 - 1/K]; \bbR^{d})$ of SDE \eqref{eq:reverse time sde} and \eqref{eq:discrete reverse sde}, respectively. Then, if Novikov's condition imposed in Section \ref{app:regularity conditions} 
\begin{equation}\label{eq:novikov}
\mE_{\tX}\left[\exp\left\{\frac12\sum_{k=0}^{K-2}\int_{k/K}^{(k+1)/K}
\left\|\bbSig^{1/2}\left(\nabla\log p_{1-t}(\tX_t)-\bs_{1-k/K}(\tX_{k/K};\htheta_{1-k/K})\right)\right\|^2dt\right\}\right]<\infty
\end{equation}
		is satisfied, we have that  
\begingroup\small
\begin{equation*}
			\begin{aligned}
				\frac{dP_{\hX_{[0, 1 - 1/K]}}}{dP_{\tX_{[0, 1 - 1/K]}}} &= \exp\left\{-\sum_{k=0}^{K-2}\int_{k/K}^{(k+1)/K}\left\{\bbSig^{\frac{1}{2}}\left(\nabla_{\bx}\log{p_{1 - t}}(\tX_{t}) - \bs_{1-k/K}(\tX_{k/K};\htheta_{1-k/K})\right)\right\}^{\top}d\widetilde{\bW}_{t}\right. \\
				&\qquad\left. - \frac{1}{2}\sum_{k=0}^{K-2}\int_{k/K}^{(k+1)/K}\left\|\bbSig^{\frac{1}{2}}\left(\nabla_{\bx}\log{p_{1 - t}}(\tX_{t}) - \bs_{1-k/K}(\tX_{k/K};\htheta_{1-k/K})\right)\right\|^{2}dt\right\}, 
			\end{aligned}
		\end{equation*}
\endgroup
		where $\widetilde{\bW}_{t}$ is a $d$-dimensional standard Brownian motion under $P_{\tX_{[0, 1 - 1/K]}}$.  
	\end{lemma}
	By applying this lemma, we obtain the following direct consequence. 
	\begin{lemma}
		If Novikov's condition \eqref{eq:novikov} is satisfied, then 
\begin{equation*}
			\begin{aligned}
&D_{KL}\left(P_{\tX_{1-1/K}}\parallel P_{\hX_{1-1/K}}\right)\\
&\quad\leq\frac{\lambda_{\max}(\bbSig)}{2}\sum_{k=0}^{K-2}\mE_{\tX}\left[\int_{k/K}^{(k+1)/K}
\left\|\nabla\log p_{1-t}(\tX_t)-\bs_{1-k/K}(\tX_{k/K};\htheta_{1-k/K})\right\|^2dt\right].
\end{aligned}
		\end{equation*}
	\end{lemma}
	\begin{proof}
		According to the data-processing inequality of KL-divergence \citep{cover1999elements}, we have 
\begin{equation*}
			\begin{aligned}
				%D_{KL}\left(P_{\tX_{1}\mid \tX_{0}=\tX_{0}^{(i)}}\parallel P_{\hX_{1}\mid \hX_{0}=\tX_{0}^{(i)}}\right) \\
				%&= D_{KL}\left(P_{\tX_{1}\mid \tX_{0}}\parallel P_{\hX_{1}\mid \hX_{0}}\right) + D_{KL}\left(P_{\tX_{0}}\parallel P_{\hX_{0}}\right) \\
				&D_{KL}\left(P_{\tX_{1 - 1 / K}}\parallel P_{\hX_{1 - 1 / K}}\right)
				\leq D_{KL}\left(P_{\tX_{[0, 1 - 1/K]}} \parallel P_{\hX_{[0, 1 - 1/K]}}\right)\\
				& = \int\log{\frac{dP_{\tX_{[0, 1 - 1/K]}}}{dP_{\hX_{[0, 1 - 1/K]}}}}dP_{\tX_{[0, 1 - 1/K]}} \\
				& = \mE_{\tX}\left[\frac{1}{2}\sum_{k=0}^{K-2}\int_{k/K}^{(k+1)/K}\left\|\bbSig^{\frac{1}{2}}\left(\nabla_{\bx}\log{p_{1 - t}(\tX_{t})} - \bs_{1-k/K}(\tX_{k/K};\htheta_{1-k/K})\right)\right\|^{2}dt\right] \\
				& \leq \frac{1}{2}\mE_{\tX}\left[\sum_{k=0}^{K-2}\int_{k/K}^{(k+1)/K}\lambda_{\max}(\bbSig)\left\|\nabla_{\bx}\log{p_{1 - t}}(\tX_{t}) - \bs_{1-k/K}(\tX_{k/K};\htheta_{1-k/K})\right\|^{2}\,dt\right],
			\end{aligned}	
		\end{equation*}
		where the second equality follows from Lemma \ref{lem:girsanov} and the property of the It\^{o} integral \citep{oksendal2013stochastic}. 	
	\end{proof}
	The lemma bounds the gap between $P_{\hX_{1-1/K}}$ and $P_{\bX_{1/K}}$ in terms of the difference between our model $\bs_t(\bx; \htheta_{t})$ and the ground truth score function. Next, let us characterize this difference. 
	By simple algebra, we have 
\begin{equation}\label{eq:upper bound decomp}
		\begin{aligned}
			& \sum_{k=0}^{K - 2}\mE_{\tX}\left[\int_{\frac{k}{K}}^{\frac{k + 1}{K}}\lambda_{\max}(\bbSig)\left\|\nabla_{\bx}\log{p_{1 - t}(\tX_{t})} - \bs_{1-k/K}(\tX_{k/K};\htheta_{1-k/K})\right\|^{2}\,dt\right] \\
			& \leq 4\sum_{k=0}^{K - 2}\mE_{\tX}\left[\int_{\frac{k}{K}}^{\frac{k + 1}{K}}\lambda_{\max}(\bbSig)\left\|\nabla_{\bx}\log{p_{1 - t}}(\tX_{t}) - \nabla_{\bx}\log{p_{1 - t}}\left(\tX_{k/K}\right)\right\|^{2}\,dt\right] \\
			& + 4\sum_{k=0}^{K - 2}\mE_{\tX}\left[\int_{\frac{k}{K}}^{\frac{k + 1}{K}}\lambda_{\max}(\bbSig)\left\|\nabla_{\bx}\log{p_{1 - t}}\left(\tX_{k/K}\right) - \nabla_{\bx}\log{p_{1 - k/K}}\left(\tX_{k/K}\right)\right\|^{2}\,dt\right] \\
			& + 4\sum_{k=0}^{K - 2}\mE_{\tX}\left[\int_{\frac{k}{K}}^{\frac{k + 1}{K}}\lambda_{\max}(\bbSig)\left\|\bs_{1-k/K}\left(\tX_{k/K};\htheta_{1-k/K}\right) - \nabla_{\bx}\log{p_{1 - k/K}}\left(\tX_{k/K}\right)\right\|^{2}\,dt\right].
		\end{aligned}
	\end{equation}
	Next, we bound the three terms in the above inequality separately to characterize the final sampling error. First, we use the Lipschitz continuity assumption \ref{ass:lip continuity} to bound the first term. 
	\vspace{-0.1in}
	\begin{lemma}\label{lem:score lip bound}
		If $\tX_{t}$ follows \eqref{eq:reverse time sde}, and $k / K \leq t < (k + 1) / K$, we have 
\begin{equation*}
			\begin{aligned}			 
				\mE\left[\left\|\nabla_{\bx}\log{p_{1 - t}(\tX_{t})} - \nabla_{\bx}\log{p_{1 - t}}\left(\tX_{k/K}\right)\right\|^{2}\right] & \leq \frac{2L^{2}\tr\{\bbSig\}}{K} + \frac{2L^{3}d\lambda_{\max}^{2}(\bbSig)}{K^{2}}
			\end{aligned}
		\end{equation*}
	\end{lemma}
	\begin{proof}
		By the Lipschitz continuity assumption \ref{ass:lip continuity}, we have 
\begin{equation}\label{eq:score difference bound}
			\begin{aligned}
				& \mE_{\tX}\left[\left\|\nabla_{\bx}\log{p_{1 - t}(\tX_{t})} - \nabla_{\bx}\log{p_{1 - t}}\left(\tX_{k/K}\right)\right\|^{2}\right] \leq L^{2}\mE\left[\left\|\tX_{t} - \tX_{k/K}\right\|^{2}\right] \\ 
				& \leq 2L^{2}\mE\left[\left\|\int_{\frac{k}{K}}^{t}\bbSig\nabla_{\bx}\log{p_{1 - s}}(\tX_{s})ds\right\|^{2} + \left\|\int_{\frac{k}{K}}^{t}\bbSig^{\frac{1}{2}}d\widetilde{\bW}_{s}\right\|^{2}\right].
			\end{aligned}
		\end{equation}
		Then we have 
		\begin{equation}\label{eq:ito norm bound}
			\mE\left[\left\|\int_{\frac{k}{K}}^{t}\bbSig^{\frac{1}{2}}d\widetilde{\bW}_{s}\right\|^{2}\right] = \tr\{\bbSig\}\int_{\frac{k}{K}}^{t}ds = \tr\{\bbSig\}\left(t - \frac{k}{K}\right).  
		\end{equation}
		On the other hand, by Schwarz's inequality  
		\begin{equation}\label{eq:score norm bound}
			\begin{aligned}
				\mE& \left[\left\|\int_{\frac{k}{K}}^{t}\bbSig\nabla_{\bx}\log{p_{1 - s}}(\tX_{s})ds\right\|^{2}\right] \leq \lambda_{\max}^{2}(\bbSig)\left(t - \frac{k}{K}\right)\mE\left[\int_{\frac{k}{K}}^{t}\left\|\nabla_{\bx}\log{p_{1 - s}}(\tX_{s})\right\|^{2}ds\right]. 
			\end{aligned}
		\end{equation}
		By integration by parts, 
		\begin{equation*}
			\begin{aligned}
				\mE\left[\left\|\nabla_{\bx}\log{p_{1 - s}}(\tX_{s})\right\|^{2}\right] &= \int_{\bbR^{d}}p_{1 - s}(\bx)\nabla_{\bx}\log{p_{1 - s}}(\bx)^{\top}\nabla_{\bx}\log{p_{1 - s}}(\bx)d\bx \\
				& = -\int_{\bbR^{d}}p_{1 - s}(\bx)\tr\{\nabla_{\bx}^{2}\log{p_{1 - s}}(\bx)\}d\bx \leq Ld, 
			\end{aligned}
		\end{equation*}
		where the last inequality follows from the Lipschitz continuity assumption \ref{ass:lip continuity}. Plugging this into \eqref{eq:score norm bound} and combining it with \eqref{eq:ito norm bound} and \eqref{eq:score difference bound}, we obtain the conclusion. 
	\end{proof}
We then use the following lemma to bound the second term on the right-hand side of inequality \eqref{eq:upper bound decomp}.
\begin{lemma}\label{lem:score t bound}
Under Assumption \ref{ass:lip continuity}, let $\tX_t$ follow \eqref{eq:reverse time sde}.
For $k=0,\ldots,K-2$ and $k/K\leq t<(k+1)/K$,
\begin{equation*}
\mE\left[\left\|\nabla\log p_{1-t}(\tX_{k/K})-\nabla\log p_{1-k/K}(\tX_{k/K})\right\|^2\right]
\leq\frac{dL^2\lambda_{\max}(\bbSig)}{K}.
\end{equation*}
\vspace{-0.3in}
\end{lemma}
\begin{proof}
Differentiating the Gaussian convolution and integrating by parts gives
\begin{equation*}
\nabla_{\bx}\log p_{1-k/K}(\bx)
=\mE\left[\nabla_{\bx}\log p_{1-t}(\bX_{1-t})\mid\bX_{1-k/K}=\bx\right].
\end{equation*}
Since $\tX_{k/K}\overset{d}{=}\bX_{1-k/K}$, conditional Jensen's inequality and Assumption \ref{ass:lip continuity} imply
\begin{equation*}
\begin{aligned}
&\mE\left[\left\|\nabla_{\bx}\log p_{1-t}(\tX_{k/K})-\nabla_{\bx}\log p_{1-k/K}(\tX_{k/K})\right\|^2\right]\\
&\quad\leq\mE\left[\left\|\nabla_{\bx}\log p_{1-t}(\bX_{1-k/K})-\nabla_{\bx}\log p_{1-t}(\bX_{1-t})\right\|^2\right]\\
&\quad\leq L^2\mE\left[\left\|\bX_{1-k/K}-\bX_{1-t}\right\|^2\right]
=L^2\left(t-\frac{k}{K}\right)\tr\{\bbSig\}
\leq\frac{dL^2\lambda_{\max}(\bbSig)}{K},
\end{aligned}
\end{equation*}
where the equality follows from the forward process \eqref{eq:sode}. This proves the conclusion.
\end{proof}

	With all these lemmas, we are ready to prove Theorem \ref{thm:sampling error}.
	
	\noindent{\bf Theorem \ref{thm:sampling error} and an auxiliary KL bound.}
	\emph{Suppose that $\hX_{0} \sim P_{\bX_{1}}$ is independent of the fitted score functions $\bs_t(\cdot;\htheta_t)$, and that $\hX_{1}$ is generated from the initial point $\hX_{0}$ following the diffusion updates in \eqref{eq:euler solving sde}.
		For $K\geq2$, under Assumptions \ref{ass:regularity}, \ref{ass:lip continuity} and \ref{ass:subexp} in Supplementary Material, assume that for all $t\in \{1/K, 2/K, \cdots, 1\}$,  
		\begin{equation*}
			\mE_{\bX_{t}}\left[\left\|\bs_{t}(\bX_{t}; \htheta_{t}) - \nabla\log{p_{t}}(\bX_{t})\right\|^{2}\right] \leq \delta.
		\end{equation*}
        Then given the fitted score functions,
\begingroup\small
\begin{equation}\label{eq:early stop KL}
			D_{KL}\left(P_{\bX_{1 / K}} \parallel P_{\hX_{1 - 1/K}}\right) \leq \cO\left(\frac{L^{2}\lambda_{\max}(\bbSig)\{2\tr\{\bbSig\}+d\lambda_{\max}(\bbSig)\}}{K} + \frac{dL^{3}\lambda_{\max}^{3}(\bbSig)}{K^{2}} + \lambda_{\max}(\bbSig)\delta\right),
	\end{equation}
\endgroup
		and 
		\begin{equation*}
			\TV(P_{\bX_0},P_{\hX_1})=\TV(P_{\bX_{0}}, P_{\hX_{1 - 1/K}}) \leq \cO\left(\sqrt{\frac{d}{K} + \delta}\right).
		\end{equation*}}
	\vspace{-0.3in}
	\begin{proof}
		The theorem is directly obtained by combining \eqref{eq:upper bound decomp}, Lemmas \ref{lem:score lip bound}, and \ref{lem:score t bound}. Since $P_{\tX_{t}} = P_{\bX_{1 - t}}$, plugging Lemmas \ref{lem:score lip bound}, \ref{lem:score t bound}, and the condition $\mE\left[\left\|\bs_t(\bX_t;\htheta_t) - \nabla\log p_t(\bX_t)\right\|^{2}\right] \leq \delta$ into \eqref{eq:upper bound decomp} gives
\begingroup\small
\begin{equation*}
			\begin{aligned}
				D_{KL}\left(P_{\tX_{1 - 1/K}}\parallel P_{\hX_{1 - 1/K}}\right) & \leq 4\sum_{k=0}^{K - 2}\int_{\frac{k}{K}}^{\frac{k + 1}{K}}\lambda_{\max}(\bbSig)\left(\frac{2L^{2}\tr\{\bbSig\}}{K} + \frac{2dL^{3}\lambda_{\max}^{2}(\bbSig)}{K^{2}}\right)dt \\
				& + 4\sum_{k=0}^{K - 2}\int_{\frac{k}{K}}^{\frac{k + 1}{K}}\lambda_{\max}(\bbSig)\frac{dL^{2}\lambda_{\max}(\bbSig)}{K}dt 
				+ 4\sum_{k=0}^{K - 2}\int_{\frac{k}{K}}^{\frac{k + 1}{K}}\lambda_{\max}(\bbSig)\delta dt \\
				& = \cO\left(\frac{L^{2}\lambda_{\max}(\bbSig)\{2\tr\{\bbSig\}+d\lambda_{\max}(\bbSig)\}}{K} + \frac{dL^{3}\lambda_{\max}^{3}(\bbSig)}{K^{2}} + \lambda_{\max}(\bbSig)\delta\right), 
			\end{aligned}			
		\end{equation*}
\endgroup
		which proves \eqref{eq:early stop KL} by noting that $\tX_{1 - 1/K} \sim P_{\bX_{1/K}}$. On the other hand, by Pinsker's inequality, $\tr\{\bbSig\} = \cO(d)$ and \eqref{eq:early stop KL}, we have 
		\begin{equation}\label{eq:early stop tv}
			\TV(P_{\bX_{1/K}}, P_{\hX_{1 - 1/K}}) \leq \sqrt{\frac{1}{2}D_{KL}\left(P_{\bX_{1 / K}} \parallel P_{\hX_{1 - 1/K}}\right)} \leq \cO\left(\sqrt{\frac{d}{K} + \delta} \right). 
		\end{equation}
		Finally, by noting that $P_{\bX_{1/K}} = P_{\bX_{0}} * \cN(\bzero, (1/K)\bbSig)$, take $\beps\sim\cN(\bzero,\bbSig)$. Jensen's inequality and the lower Hessian bound $\nabla_{\bx}^{2}\log p_0(\bx)\succeq-L\bI_d$ implied by Assumption \ref{ass:lip continuity} give
\begin{equation*}
\begin{aligned}
\log p_{1/K}(\bx)
&=\log\mE_{\beps}[p_0(\bx-\beps/\sqrt K)]
\geq\mE_{\beps}[\log p_0(\bx-\beps/\sqrt K)]\\
&\geq\log p_0(\bx)-\frac{\nabla\log p_0(\bx)^{\top}\mE[\beps]}{\sqrt K}
-\frac{L\mE[\|\beps\|^2]}{2K}
=\log p_0(\bx)-\frac{L\tr\{\bbSig\}}{2K}.
\end{aligned}
\end{equation*}
Consequently, Pinsker's inequality yields
\begin{equation*}
\TV(P_{\bX_0},P_{\bX_{1/K}})
\leq\sqrt{\frac12D_{KL}(P_{\bX_0}\parallel P_{\bX_{1/K}})}
\leq\frac12\sqrt{\frac{L\tr\{\bbSig\}}{K}}.
\end{equation*}
		Combining this with \eqref{eq:early stop tv} and applying triangle inequality of TV distance, we prove our conclusion. 
	\end{proof}
	\section{Proof of Results in Section \ref{sec:generalization analysis}}\label{app:generalization proof}
	{\bf Restatement of Theorem \ref{thm:gen error}.}
\emph{Fix $b>0$, $0<\eta<1/4$, and $K\geq8$, and suppose $\lambda>2\delta_{2n}(\eta)\log K$.
Under Assumption \ref{ass:moment norm} and Assumptions \ref{ass:lip continuity}, \ref{ass:subexp}, and \ref{ass:bernstein} in Supplementary Material,
\begin{equation*}
\begin{aligned}
&\mE_{\bX_t}\left[\left\|\bs_t(\bX_t;\htheta_t)-\nabla\log p_t(\bX_t)\right\|^2\right]\\
&\leq 2\lambda\{\cB(b)+\log(K\eta^{-1})\}
    +\frac{4\sigma_0(b)^2}{m}+4\delta_{0m}(\eta;b)^2\log^2 K+4\kappa(b)\\
&\quad+\frac{4\{\delta_{1n}(\eta)+2\delta_{2n}(\eta)\sqrt{C_{\eta,m,\lambda}(b)}\}^{2}\log^2 K}{\lambda}
\end{aligned}
\end{equation*}
holds simultaneously for all $t\in\{1/K,\ldots,1\}$ with probability at least $1-4\eta$, where the randomness comes from the training sample.}
\begin{proof}
		Let $\bar{\btheta}_{t,b} = \frac{1}{m}(\bbf_{t,b}(\bX_{t}^{(1)}), \dots, \bbf_{t,b}(\bX_{t}^{(m)}))^{\top}$ where $\bbf_{t,b}$ is defined above \eqref{eq: approximate mean}. Then, we have
		\begin{equation}\label{eq:bound bar}
			\mE_{\bX_{t}}\left[\left\|\bs_t(\bX_{t}; \bar{\btheta}_{t,b}) - \nabla\log{p_{t}}(\bX_{t})\right\|^{2}\right] \leq \frac{2\sigma_0(b)^2}{m} + 2\delta_{0m}(\eta;b)^{2}\log^2 K + 2\kappa(b)
		\end{equation}
		with probability at least $1-\eta$, simultaneously over $t\in\{1/K,\cdots,1\}$, by the proof of Proposition \ref{prop:approximation error}, a union bound, and $\delta_{0m}(\eta/K;b)\leq\delta_{0m}(\eta;b)\log K$.
		Let $\widehat{\bzeta}_{t} = (n - m)^{-1}\sum_{i = m + 1}^{n}\bzeta_{\bX_{1}^{(i)}, t}$, $\bzeta_{t} = \mE[\bzeta_{\bX_{1}, t}]$, $\widehat{\bGamma}_{t} = (n - m)^{-1}\sum_{i = m + 1}^{n}\bGamma_{\bX_{1}^{(i)}, t}$ and $\bGamma_{t} = \mE[\bGamma_{\bX_{1}, t}]$ (see definitions in Appendix \ref{app:conditions for concentration}). All candidate score models in this proof lie in $\cH_{\cK,\bbR}^{d}$; the true score need not lie in this space. Under Assumption \ref{ass:bernstein} for the operators defined above, we have
		\begin{equation}\label{eq: bound zeta}
			\sup_{t\in\{1/K,\cdots, 1\}}\left\|\widehat{\bzeta}_{t} - \bzeta_{t}\right\|_{\cH_{\cK}^{d}} \leq 
			2\left(\frac{H_{1}}{n - m} + \frac{\sigma_{1}}{\sqrt{n - m}}\right)\log\frac{2K}{\eta} \leq \delta_{1n}(\eta)\log K
		\end{equation}
		and
		\begin{equation}\label{eq: bound Gamma}
			\sup_{t\in\{1/K,\cdots, 1\}}\left\|\widehat{\bGamma}_{t} - \bGamma_{t}\right\|_{\rm HS} \leq 
			2\left(\frac{H_{2}}{n - m} + \frac{\sigma_{2}}{\sqrt{n - m}}\right)\log\frac{2K}{\eta} \leq \delta_{2n}(\eta)\log K
		\end{equation}
		with probability at least $1 - 2\eta$ according to Proposition 23 in \cite{bauer2007regularization} and a union bound over the $K$ grid points. Subsequently, we write $\bs_t(\cdot; \btheta_{t})$ as $\bs_{\btheta_{t}}$ for simplicity. Notice that the objective function can be written as 
		\begin{equation*}
			\widehat{\cL}_{t}(\btheta_{t}) = \left\langle \bs_{\btheta_{t}}, (\widehat{\bGamma}_{t} + \lambda\mathrm{Id}) \bs_{\btheta_{t}}   \right\rangle_{\cH_{\cK}^{d}} + \left\langle \widehat{\bzeta}_{t}, \bs_{\btheta_{t}}   \right\rangle_{\cH_{\cK}^{d}},
		\end{equation*}
		where $\mathrm{Id}$ is the identity operator on $\cH_{\cK,\bbR}^{d}$. Let
		\begin{equation*}
			\cL_{t}(\btheta_{t}) = \left\langle \bs_{\btheta_{t}}, (\bGamma_{t} + \lambda\mathrm{Id}) \bs_{\btheta_{t}}   \right\rangle_{\cH_{\cK}^{d}} + \left\langle \bzeta_{t}, \bs_{\btheta_{t}}   \right\rangle_{\cH_{\cK}^{d}},
		\end{equation*}
		and let $\bar{\btheta}_{t,\lambda}$ minimize $\cL_t$ over the same finite kernel span. This population minimizer does not depend on $b$. By complex debiasing and score matching,
		\begin{equation*}
			\begin{aligned}
\cL_t(\btheta_t)&=\mE_{\bX_t}\left[\left\|\bs_t(\bX_t;\btheta_t)-\nabla\log p_t(\bX_t)\right\|^2\right]
    +\lambda\|\bs_{\btheta_t}\|_{\cH_{\cK}^{d}}^2
-\mE_{\bX_t}\left[\|\nabla\log p_t(\bX_t)\|^2\right].
\end{aligned}
		\end{equation*}
        Notice that $0\leq\cK_t^{(0)}(\bx_1,\bx_2)\leq|\bbOmega_t^{(0)}|^{-1/2}|\bbOmega^{(0)}|^{1/2}$ and $\bbf_{t,b}(\bx_1)^{\top}\bbf_{t,b}(\bx_2)\leq\frac12(\|\bbf_{t,b}(\bx_1)\|^2+\|\bbf_{t,b}(\bx_2)\|^2)$ for any $\bx_1,\bx_2\in\bbR^d$. We have
\begin{equation*}
\mE\left[\exp\left\{m^{-2}\sum_{i,j=1}^m\bbf_{t,b}(\bX_t^{(i)})^{\top}\bbf_{t,b}(\bX_t^{(j)})\cK_t^{(0)}(\bX_1^{(i)},\bX_1^{(j)})\right\}\right]\leq\exp\{\cB(b)\}
\end{equation*}
according to Assumption \ref{ass:moment norm} and Jensen's inequality. Then, by Markov's inequality and a union bound, we have
		\begin{equation*}
			\begin{aligned}
				\sup_{t\in\{1/K, \cdots, 1\}}\left\|\bs_{\bar{\btheta}_{t,b}}\right\|_{\cH_{\cK}^{d}}^{2} 
				&\leq \cB(b) + \log\left(K\eta^{-1}\right), 
			\end{aligned}
		\end{equation*} 
		with probability at least $1 - \eta$. Combining this with \eqref{eq:bound bar},
		we have
		\begin{equation}\label{eq: error s_bar}
			\begin{aligned}
&
				\mE_{\bX_{t}}\left[\left\|\bs_t(\bX_{t}; \bar{\btheta}_{t,\lambda}) - \nabla\log{p_{t}}(\bX_{t})\right\|^{2}\right]\\
				&\leq \mE_{\bX_{t}}\left[\left\|\bs_t(\bX_{t}; \bar{\btheta}_{t,b}) - \nabla\log{p_{t}}(\bX_{t})\right\|^{2}\right] + \lambda\left\|\bs_{\bar{\btheta}_{t,b}}\right\|_{\cH_{\cK}^{d}}^{2}\\
				&\leq \frac{2\sigma_0(b)^2}{m} + 2\delta_{0m}(\eta;b)^{2}\log^2 K + 2\kappa(b) + \lambda\{\cB(b)+\log(K\eta^{-1})\}     
			\end{aligned}
		\end{equation}
		for all $t\in\{1/K, \cdots, 1\}$, and
		\begin{equation}\label{eq: bound norm s_bar}
			\begin{aligned}
				\left\|\bs_{\bar{\btheta}_{t,\lambda}}\right\|_{\cH_{\cK}^{d}}^{2}
				&\leq \lambda^{-1}\mE_{\bX_{t}}\left[\left\|\bs_t(\bX_{t}; \bar{\btheta}_{t,b}) - \nabla\log{p_{t}}(\bX_{t})\right\|^{2}\right] + \left\|\bs_{\bar{\btheta}_{t,b}}\right\|_{\cH_{\cK}^{d}}^{2}\\
				&\leq \lambda^{-1}\left\{\frac{2\sigma_0(b)^2}{m} + 2\delta_{0m}(\eta;b)^{2}\log^2 K + 2\kappa(b)\right\} + \cB(b) + \log\left(K\eta^{-1}\right)
				=  C_{\eta,m,\lambda}(b)
			\end{aligned}
		\end{equation}
		on the joint approximation and norm event, which has probability at least $1-2\eta$,
		by the definition of $\bar{\btheta}_{t, \lambda}$. Next, we bound the gap between $\bs_{\htheta_{t}}$ and $\bs_{\bar{\btheta}_{t, \lambda}}$.
		By \eqref{eq: bound Gamma} and $\lambda>2\delta_{2n}(\eta)\log K$, the empirical quadratic form is coercive on the finite kernel span. Thus $\htheta_t$ minimizes $\widehat{\cL}_t(\btheta_t)$, and 
		\begin{equation}\label{eq:basic inequality}
			\left\langle \bs_{\htheta_{t}}, (\widehat{\bGamma}_{t} + \lambda\mathrm{Id}) \bs_{\htheta_{t}}   \right\rangle_{\cH_{\cK}^{d}} + \left\langle \widehat{\bzeta}_{t}, \bs_{\htheta_{t}}   \right\rangle_{\cH_{\cK}^{d}} \leq  \left\langle \bs_{\bar{\btheta}_{t, \lambda}}, (\widehat{\bGamma}_{t} + \lambda\mathrm{Id}) \bs_{\bar{\btheta}_{t, \lambda}}   \right\rangle_{\cH_{\cK}^{d}} + \left\langle \widehat{\bzeta}_{t}, \bs_{\bar{\btheta}_{t, \lambda}}   \right\rangle_{\cH_{\cK}^{d}}.
		\end{equation}
		The first-order condition for the population minimizer on this linear span gives 
		\begin{equation*}
			\left\langle \bs_{\htheta_{t}} - \bs_{\bar{\btheta}_{t,\lambda}}, 2(\bGamma_{t} + \lambda\mathrm{Id}) \bs_{\bar{\btheta}_{t, \lambda}} + \bzeta_{t}\right\rangle_{\cH_{\cK}^{d}} = 0
		\end{equation*}
		because the difference of the two fitted models belongs to the same span. Combining this with \eqref{eq:basic inequality}, we have
		\begingroup\small
\begin{equation*}
			\begin{aligned}
				&2\left\langle \bs_{\htheta_{t}} 
				- \bs_{\bar{\btheta}_{t, \lambda}}, (\widehat{\bGamma}_{t} - \bGamma_{t})\bs_{\bar{\btheta}_{t, \lambda}}   \right\rangle_{\cH_{\cK}^{d}} - \left\langle \bzeta_{t},  \bs_{\htheta_{t}} 
				- \bs_{\bar{\btheta}_{t, \lambda}}   \right\rangle_{\cH_{\cK}^{d}} +
				\left\langle \bs_{\htheta_{t}} 
				- \bs_{\bar{\btheta}_{t, \lambda}}, (\widehat{\bGamma}_{t} + \lambda\mathrm{Id})(\bs_{\htheta_{t}} 
				- \bs_{\bar{\btheta}_{t, \lambda}})   \right\rangle_{\cH_{\cK}^{d}}\\
				& =
				\left\langle \bs_{\htheta_{t}}, (\widehat{\bGamma}_{t} + \lambda\mathrm{Id}) \bs_{\htheta_{t}}   \right\rangle_{\cH_{\cK}^{d}} - \langle \bs_{\bar{\btheta}_{t, \lambda}}, (\widehat{\bGamma}_{t} + \lambda\mathrm{Id}) \bs_{\bar{\btheta}_{t, \lambda}}   \rangle_{\cH_{\cK}^{d}} \\
				&\leq         
				- \left\langle \widehat{\bzeta}_{t}, \bs_{\htheta_{t}} 
				- \bs_{\bar{\btheta}_{t, \lambda}}   \right\rangle_{\cH_{\cK}^{d}}.
			\end{aligned}
		\end{equation*}
\endgroup
		This implies that
		\begin{equation}\label{eq:quadratic bound}
\begin{aligned}
&\left\langle\bs_{\htheta_t}-\bs_{\bar{\btheta}_{t,\lambda}},
\bGamma_t(\bs_{\htheta_t}-\bs_{\bar{\btheta}_{t,\lambda}})\right\rangle_{\cH_{\cK}^{d}}
+\{\lambda-\delta_{2n}(\eta)\log K\}\|\bs_{\htheta_t}-\bs_{\bar{\btheta}_{t,\lambda}}\|_{\cH_{\cK}^{d}}^2\\
&\quad\leq\left\langle\bs_{\htheta_t}-\bs_{\bar{\btheta}_{t,\lambda}},
(\widehat{\bGamma}_t+\lambda\mathrm{Id})(\bs_{\htheta_t}-\bs_{\bar{\btheta}_{t,\lambda}})\right\rangle_{\cH_{\cK}^{d}}\\
&\quad\leq-2\left\langle\bs_{\htheta_t}-\bs_{\bar{\btheta}_{t,\lambda}},
(\widehat{\bGamma}_t-\bGamma_t)\bs_{\bar{\btheta}_{t,\lambda}}\right\rangle_{\cH_{\cK}^{d}}
-\left\langle\widehat{\bzeta}_t-\bzeta_t,\bs_{\htheta_t}-\bs_{\bar{\btheta}_{t,\lambda}}\right\rangle_{\cH_{\cK}^{d}}\\
&\quad\leq\{\delta_{1n}(\eta)+2\delta_{2n}(\eta)\sqrt{C_{\eta,m,\lambda}(b)}\}\log K
\|\bs_{\htheta_t}-\bs_{\bar{\btheta}_{t,\lambda}}\|_{\cH_{\cK}^{d}}.
\end{aligned}
\end{equation}
		according to \eqref{eq: bound zeta}, \eqref{eq: bound Gamma} and \eqref{eq: bound norm s_bar} with probability at least $1 - 4\eta$. Notice that 
		\begin{equation*}
			\left\langle \bs_{\htheta_{t}} 
			- \bs_{\bar{\btheta}_{t, \lambda}}, \bGamma_{t}(\bs_{\htheta_{t}} 
			- \bs_{\bar{\btheta}_{t, \lambda}})   \right\rangle_{\cH_{\cK}^{d}} = \mE_{\bX_{t}}\left[\left\|\bs_t(\bX_{t}; \htheta_{t}) - \bs_t(\bX_{t}; \bar{\btheta}_{t,\lambda})\right\|^{2}\right] \geq 0.
		\end{equation*}
		Thus,
		\begingroup\small
\begin{equation*}
			\begin{aligned}
				\left\|\bs_{\htheta_{t}} 
				- \bs_{\bar{\btheta}_{t, \lambda}}\right\|_{\cH_{\cK}^{d}}
				&\leq \frac{\{\delta_{1n}(\eta)+2\delta_{2n}(\eta)\sqrt{C_{\eta,m,\lambda}(b)}\}\log K}{\lambda - \delta_{2n}(\eta)\log K}\leq 
				\frac{2\{\delta_{1n}(\eta)+2\delta_{2n}(\eta)\sqrt{C_{\eta,m,\lambda}(b)}\}\log K}{\lambda}.
			\end{aligned}
		\end{equation*}
\endgroup
		
		Then, \eqref{eq:quadratic bound} implies that
		\begin{equation*}
			\mE_{\bX_{t}}\left[\left\|\bs_t(\bX_{t}; \htheta_{t}) - \bs_t(\bX_{t}; \bar{\btheta}_{t,\lambda})\right\|^{2}\right] \leq \frac{2\{\delta_{1n}(\eta)+2\delta_{2n}(\eta)\sqrt{C_{\eta,m,\lambda}(b)}\}^{2}\log^2 K}{\lambda}.
		\end{equation*}
		for all $t\in\{1/K, \cdots, 1\}$. This together with \eqref{eq: error s_bar} and \eqref{eq:quadratic bound} implies that
		\begin{align*}
				&\mE_{\bX_{t}}\left[\left\|\bs_t(\bX_{t}; \htheta_{t}) - \nabla\log{p_{t}}(\bX_{t})\right\|^{2}\right]\\
				& \leq \mE_{\bX_{t}}\left[\left(\left\|\bs_t(\bX_{t}; \bar{\btheta}_{t,\lambda}) - \nabla\log{p_{t}}(\bX_{t})\right\| + \left\|\bs_t(\bX_{t}; \htheta_{t}) - \bs_t(\bX_{t}; \bar{\btheta}_{t,\lambda})\right\|\right)^{2}\right]\\
				&\leq 2\mE_{\bX_{t}}\left[\left\|\bs_t(\bX_{t}; \bar{\btheta}_{t,\lambda}) - \nabla\log{p_{t}}(\bX_{t})\right\|^{2}\right] +
				2\mE_{\bX_{t}}\left[\left\|\bs_t(\bX_{t}; \htheta_{t}) - \bs_t(\bX_{t}; \bar{\btheta}_{t,\lambda})\right\|^{2}\right]
				\\
				&\leq  2\lambda(\cB(b) + \log\left(K\eta^{-1}\right)) + \frac{4\sigma_0(b)^2}{m} + 4\delta_{0m}(\eta;b)^{2}\log^2 K +  4\kappa(b)\\
				&\quad+\frac{4\{\delta_{1n}(\eta)+2\delta_{2n}(\eta)\sqrt{C_{\eta,m,\lambda}(b)}\}^{2}\log^2 K}{\lambda}
		\end{align*}
		with probability at least $1 - 4\eta$, for all $t\in\{1/K, \cdots, 1\}$, 
        which completes the proof. The final error bound depends only on the observed training sample, so the joint event above also gives the stated probability under that sample.
	\end{proof}
	
\begin{proof}[Corollary \ref{cor:convergence rate}]
With the choices of $m$ and $\lambda$ in the corollary, the definitions give $\delta_{1n}(\eta)\log K=\cO(\lambda)$ and $\delta_{2n}(\eta)\log K=\cO(\lambda)$. Moreover,
\[
\delta_{0m}(\eta;b)^2\log^2 K
=\cO\left(\left\{\frac{\sigma_0(b)^2}{\sqrt n}+\frac{H_0(b)^2}{n}\right\}\log^2(2/\eta)\log^2 K\right).
\]
Also,
\begin{equation*}
\begin{aligned}
\frac{\delta_{2n}(\eta)^2 C_{\eta,m,\lambda}(b)\log^2 K}{\lambda}
&=\cO\bigg(\frac{\sigma_0(b)^2}{m}+\delta_{0m}(\eta;b)^2\log^2 K+\kappa(b)
+\lambda\{\cB(b)+\log(K\eta^{-1})\}\bigg),
\end{aligned}
\end{equation*}
while $\delta_{1n}(\eta)^2\log^2 K/\lambda=\cO(\lambda)$. Substitution into Theorem \ref{thm:gen error} gives the uniform squared score-error bound stated before the corollary. Theorem \ref{thm:sampling error}, conditional on the fitted scores, then gives the asserted TV bound on the same event of probability at least $1-4\eta$. The constants in these bounds are independent of $b$, which may therefore depend on $n$.

Under the additional conditions \eqref{eq:b growth conditions}, for $b\geq1$ the squared score-error bound simplifies to
\begin{equation*}
\begin{aligned}
&\sup_{t\in\{1/K,\ldots,1\}}\mE_{\bX_t}\left[\left\|\bs_t(\bX_t;\htheta_t)-\nabla\log p_t(\bX_t)\right\|^2\right]
\leq\cO\left(b^{-\rho}+b^{\gamma}\frac{\log^2(2/\eta)\log^2 K}{\sqrt n}\right).
\end{aligned}
\end{equation*}
Here we use $\log(K/\eta)\leq\log(2/\eta)\log K$, $b^{\gamma}\geq1$, and $n^{-1}\leq n^{-1/2}$. The choice of $b$ in \eqref{eq:optimal b} is at least one under the stated condition and balances the two terms at
\[
\left\{\frac{\log^2(2/\eta)\log^2 K}{\sqrt n}\right\}^{\rho/(\rho+\gamma)}.
\]
Applying Theorem \ref{thm:sampling error} and $\sqrt{x+y}\leq\sqrt x+\sqrt y$ for $x,y\geq0$ gives
\[
\TV(P_{\bX_0},P_{\hX_1})
\leq\cO\left(\sqrt{\frac{d}{K}}
+\left\{\frac{\log^2(2/\eta)\log^2 K}{\sqrt n}\right\}^{\rho/\{2(\rho+\gamma)\}}\right).
\]
For fixed $d,\eta$ and $K$ chosen in \eqref{eq:optimal b}, we have $\log K\asymp\log n$, so the condition $\log^2(2/\eta)\log^2 K\leq\sqrt n$ holds for sufficiently large $n$. Moreover,
\[
\sqrt{\frac{d}{K}}
\asymp\left\{\frac{\log^2(2/\eta)\log^2 K}{\sqrt n}\right\}^{\rho/\{2(\rho+\gamma)\}}
\asymp n^{-\rho/\{4(\rho+\gamma)\}}(\log n)^{\rho/(\rho+\gamma)}.
\]
This proves the rate in Corollary \ref{cor:convergence rate}. Both $b$ and $K$ are deterministic, so the failure probability remains $4\eta$.
\end{proof}
	\subsection{Application to Empirical Bayes Estimation}
	\label{subsec: empirical bayes}
	
	Our denoising framework also provides a natural tool for empirical Bayes estimation. 
	The connection follows from the observation that the measurement-error model considered in this paper is mathematically equivalent to a nonparametric normal-means empirical Bayes model. 
	Concretely, suppose that for $i=1,\ldots,n$ we observe a noisy estimate $\bY_i\in\mathbb R^d$ of an unobserved latent effect $\bmu_i\in\mathbb R^d$:
	\[
	\bY_i=\bmu_i+\beps_i,\qquad 
	\beps_i\sim \cN(\bzero,\bbSig),\qquad 
	\bmu_i\sim G,
	\]
	where $G$ is an unknown prior distribution. 
	In this formulation, $\bY_i$ corresponds to the error-contaminated observation $\bX_1^{(i)}$, $\bmu_i$ corresponds to the error-free latent variable $\bX_0^{(i)}$, and $G=P_{\bX_0}$ is the unknown distribution to be learned from the noisy observations. 
	This model arises in many large-scale inference problems, including estimation of many treatment effects \citep{azevedo2019empirical}, gene-level \citep{smyth2004linear} or variant-level effects \citep{thompson2015empirical}, hospital profiling \citep{thomas1994empirical}, and small-area estimation \citep{fay1979estimates}. 
	In these applications, the covariance matrix $\bbSig$ is often known or consistently estimated from individual-level data, making the normal-means formulation particularly compatible with the setting studied in this paper.
	
	The proposed RKHS-based score estimation method can therefore be used to estimate the empirical Bayes prior distribution $G$ without specifying a parametric mixture model. 
	After training the score functions $\{\bs_t(\cdot;\htheta_t):t\in\{2/K,\ldots,1\}\}$ from the noisy observations $\{\bY_i\}_{i=1}^n$, one can generate samples from the estimated prior distribution by running the empirical reverse SDE in Section~\ref{sec:sampling with reverse-sde} with the initial values drawn from the observed empirical distribution of $\{\bY_i\}_{i=1}^n$. 
	The resulting denoised samples approximate draws from $G$, and hence can be used to estimate prior-level quantities such as
	\[
	\int h(\bmu)\,dG(\bmu)
	\]
	for a user-specified functional $h$. 
	For example, $h$ may represent a coordinate-wise moment, a tail probability, a sign probability, or a nonlinear summary of the latent effects.
	
	More importantly, the same reverse-SDE construction yields posterior simulation for individual empirical Bayes estimation. 
	For a fixed observation $\bY_i= \by_i$, the exact reverse-time diffusion initialized at $\by_i$ has terminal distribution equal to the posterior distribution of the latent effect,
	\[
	P_{\bmu_i\mid \bY_i=\by_i} = P_{\bX_0\mid \bX_1=\by_i}.
	\]
	Thus, by running the empirical reverse SDE independently $R$ times with the same initialization $\by_i$, we obtain approximate posterior draws\footnote{With exact scores and continuous-time sampling, early stopping targets $P_{\bX_{1/K}\mid\bX_1=\by_i}$. Accuracy relative to $P_{\bX_0\mid\bX_1=\by_i}$ requires control of the conditional estimation, discretization, and early-stopping errors, which is not supplied by the marginal TV bound.} 
	\[
	\widehat\bmu_i^{(r)}=\hX_1^{(i,r)},\qquad r=1,\ldots,R,\quad R\geq2,
	\]
	where $\hX_0^{(i,r)}$ are initialized from $\by_i$ and, consistently with the early-stopped sampler, $\hX_1^{(i,r)}=\hX_{1-1/K}^{(i,r)}$ and $k=0,\ldots,K-2$,
	\[
	\hX_{(k+1)/K}^{(i,r)}
	=
	\hX_{k/K}^{(i,r)}
	+
	\frac{1}{K}\bbSig 
	\bs_{1 - k/K}\!\left(\hX_{k/K}^{(i,r)};\htheta_{1-k/K}\right)
	+
	\sqrt{\frac{1}{K}}\bbSig^{1/2}\beps_k^{(i,r)},
	\qquad 
	\beps_k^{(i,r)}\overset{\rm iid}{\sim} \cN(\bzero, \bI_d).
	\]
	These posterior draws can be summarized to obtain empirical Bayes point estimates and uncertainty measures:
	\[
	\widehat\bmu_i^{\rm EB}
	=
	\frac{1}{R}\sum_{r=1}^R \widehat\bmu_i^{(r)},
	\qquad
	\widehat \bV_i^{\rm EB}
	=
	\frac{1}{R-1}
	\sum_{r=1}^R
	\left(\widehat\bmu_i^{(r)}-\widehat\bmu_i^{\rm EB}\right)
	\left(\widehat\bmu_i^{(r)}-\widehat\bmu_i^{\rm EB}\right)^\top.
	\]
	For any decision set $A\subset\mathbb R^d$, the posterior probability can be estimated by
	\[
	\bbP(\bmu_i\in A\mid \bY_i)
	\approx
	\frac{1}{R}\sum_{r=1}^R 
	\mathbf 1\left\{\widehat\bmu_i^{(r)}\in A\right\}.
	\]
	
	In addition to our denoising-based method, the posterior mean or variance also admits a direct score-based representation through Tweedie's formula \citep{efron2011tweedie}. 
	To see this, let $p_1$ denote the marginal density of $\bY = \bX_1 = \bmu + \beps$. 
	Because    
	\[p_1(\by) = \int \phi_{\bbSig}(\by-\bmu) \,dG(\bmu),\]
	where $\phi_{\bbSig}$ is the density of $\cN(\bzero,\bbSig)$, we have
	\begingroup\small
\[
	\nabla \log p_1(\by)
	=
	\int -\bbSig^{-1}(\by - \bmu)\phi_{\bbSig}(\by - \bmu)\,dG(\bmu)/
	\int \phi_{\bbSig}(\by - \bmu)\,dG(\bmu)
	=
	\bbSig^{-1}\left\{\mE(\bmu\mid \bY = \by) - \by\right\}.\]
\endgroup
	Therefore,
	\[\mE(\bmu_i\mid \bY_i)
	=
	\bY_i + \bbSig\nabla\log p_1(\bY_i).\]
	Replacing the unknown score $\nabla\log p_1$ with the proposed RKHS score estimator gives the closed-form empirical Bayes shrinkage estimator
	\[\widehat\bmu_i^{\rm Tw}
	=
	\bY_i + \bbSig \bs_1(\bY_i;\htheta_1).\]
	This estimator is especially convenient when only posterior means are required (the conditional variance can be similarly obtained). 
	In contrast, the reverse-SDE posterior sampler provides a more flexible route for estimating credible regions, nonlinear posterior summaries, and decision-specific posterior probabilities.
	
	Compared with classical empirical Bayes methods based on parametric normal mixtures or explicit deconvolution of the marginal density, the proposed approach directly estimates the score functions needed for posterior denoising and does not require specifying a finite-dimensional prior family. 
	It is therefore well suited for multivariate empirical Bayes problems where the latent effects may have a complex non-Gaussian distribution. 
	For an initialization independent of the fitted scores, the theoretical results in Section~\ref{sec:theoretical analysis} imply that the marginal distribution of the denoised samples approximates the latent prior distribution $G$ under the stated regularity conditions. Reusing the fitting observations follows the practical procedure in Algorithm \ref{alg:alg1}, for which the independence requirement is not guaranteed. 
	A full theory for conditional posterior approximation is beyond the scope of this paper, but the above construction provides a practical empirical Bayes interpretation of the proposed denoising diffusion procedure.
	
	\section{Discussion on the Sampling Process}
	\begin{algorithm}[t!]
		\caption{Denoising with Diffusion Model}
		\label{alg:alg1}
		\textbf{Input:} Observations $\{\bX_1^{(i)}\}_{i=1}^n$, covariance $\bbSig$, admissible bandwidth $\bbH$, grid size $K\geq2$, and tuning parameters $1\leq m<n$ and $\lambda>0$.
		\begin{algorithmic}[1]
			% \State  {Sample $n$ i.i.d. samples $\{\bxi^{i}\}_{i=1}^{n}$ with $\bxi^{i}\sim \cN(\bmu, \bbSig)$.}
			\State  {\textbf{\emph{Training}}: 
				\State {Use $\{\bX_1^{(i)}\}_{i=1}^m$ as kernel centers and $\{\bX_1^{(i)}\}_{i=m+1}^n$ for the empirical objective.}
                \State {For $t=1-k/K$, $k=0,\ldots,K-2$, compute $\BK_t^{(0)}$, $\BK_t^{(1)}$, $\BK_t^{(2)}$ and $\htheta_t$ using \eqref{eq:closed form solution}.}
				%	\Comment{\emph{Empirical version of \eqref{eq:complex objective}}.}
				\State  {\textbf{\emph{Sampling Denoised Data}}}:
				\qquad \For    {$i=1,\cdots, n$}
				\State  {Initialize $\hX_{0}^{(i)} = \bX_{1}^{(i)}$}
				\For    {$k=0,\cdots, K - 2$}
				\State {Update $\hX_{(k+1)/K}^{(i)}$ by \eqref{eq:euler solving sde}, with Gaussian increments independent across $i$ and $k$.}
				%	 \Comment{\emph{Numerical solver of \eqref{eq:ODE integral}}.}
				\EndFor 
				\EndFor\\
				\Return {Samples $\{\hX_1^{(i)}\}_{i=1}^n$, where $\hX_1^{(i)}=\hX_{1-1/K}^{(i)}$.}
			}
		\end{algorithmic}
	\end{algorithm}
	Algorithm \ref{alg:alg1} gives the practical procedure that reuses the observed data for fitting and generation. It takes $K-1$ updates on a grid of size $K$; only scores at $t\in\{2/K,\ldots,1\}$ are used. The theorems control the slightly larger grid $\{1/K,\ldots,1\}$. The independent-generation modification and the corresponding changes to the training averages are described in Section \ref{sec:kl divergence error}. As discussed in Section \ref{sec:sampling with reverse-sde}, the sampling process numerically solves a given SDE. Solving the reverse-time SDE \eqref{eq:reverse time sde} has been well studied in the existing diffusion-model literature. One standard pipeline \citep{song2020score,song2020denoising} transforms the target reverse SDE \eqref{eq:reverse time sde} into a marginally equivalent (for exact scores) ordinary differential equation (ODE) \eqref{eq:ode} via the Fokker--Planck equation \citep{oksendal2013stochastic}, as in Proposition \ref{pro:density}. Theoretically, the samples obtained by solving the derived ODE may have smaller discretization error than those obtained by \eqref{eq:euler solving sde}, owing to the numerical advantages of ODE solvers compared with SDE solvers \citep{platen1999introduction}. 
	\par
	A broad literature \citep{bao2022analytic,lu2022dpm,xue2023sa,zhao2024unipc} improves the Euler--Maruyama method \eqref{eq:euler solving sde} by using enhanced numerical solvers for the reverse-time SDE \eqref{eq:reverse time sde} or its equivalent ODE. These improved numerical methods are useful in the deep learning context because the score function model $\bs_t(\cdot; \btheta_{t})$ is usually a deep neural network \citep{lecun2015deep,ronneberger2015u}. Improved numerical methods significantly reduce the number of function evaluations (NFEs) and the required computational cost. However, in this paper, we use a kernel-based model $\bs_t(\cdot; \btheta_{t})$ to approximate the score function. Thus, model evaluation is not expensive, so decreasing the NFEs does not bring substantial improvement in practice.  
	
	\section{Ablation Study}\label{app:ablation study}
	In this section, we follow the settings in Section \ref{sec:simulation} to explore the effect of hyperparameters on our proposed method. Concretely, we study the regularization parameter $\lambda$, time-grid size $K$, number of kernel centers $m$, anisotropic bandwidth scale $h$ in $\bbH=h\widehat{\bbSig}_{\bX_1}^{-1}$, and number of samples $n$. We report the scaled joint sliced Wasserstein distance $\mathrm{SW}_{1}^{\mathrm{sc}}=\sqrt{d}\,\mathrm{SW}_{1}$, with $d=d_{\bZ}+1$, and MSE under the three distributions of $\bZ_{0}$ used in Section \ref{sec:simulation}. Unless otherwise specified, we set $d_{\bZ}=3$, $\sigma = 0.5$, $n=1000$, $m=64$, $\lambda=n^{-1/2}$, $K=50$, and $h=0.5$.
	\par
	The results are summarized as follows: 
	\begin{enumerate}
		\item For the regularization parameter $\lambda$ in Figure \ref{fig:lam}, we find that a proper $\lambda$ guarantees the best performance of our method. This is a standard phenomenon for regularization parameters \citep{smale2003estimating} when applying RKHS-based methods. Notice that SW does not necessarily have the same trend as MSE, since the former focuses on distributional recovery while the latter depends on the quality of $\{(\hZ_{0}^{(i)}, \hat Y_{0}^{(i)})\}_{i=1}^{n}$ for downstream prediction. 
		\item For the denoising steps $K$ in Figure \ref{fig:K}, we find that increasing the number of denoising steps does not necessarily result in improved performance. We speculate that this occurs because the estimation error of the score function $\nabla\log p_t$ varies with $t$. Thus, taking more steps with an inaccurately estimated score function may induce larger error. A similar phenomenon is also observed for diffusion models in AI \citep{lu2022dpm,bao2022analytic,xue2023sa}.  
		\item For the number of kernel centers $m$ in Figure \ref{fig:m}, we find that increasing $m$ from a small value improves performance and that the method remains stable for a wide range of $m$. This observation is consistent with our theoretical justification after Theorem \ref{thm:gen error}. 
		\item For the bandwidth scale $h$ in Figure \ref{fig:h}, we observe that the proposed anisotropic bandwidth choice is stable over a broad range of values around the default $h=0.5$.
		\item For the number of samples $n$ in Figure \ref{fig:n}, we observe that increasing the number of samples generally improves the performance of our method, as suggested in Theorem \ref{thm:gen error}. 
	\end{enumerate}
	
	\begingroup
	\captionsetup{type=figure,hypcap=false,aboveskip=3pt,belowskip=0pt}
	\setlength{\parskip}{0pt}
	
	\noindent\begin{minipage}{\textwidth}
		\centering
		\includegraphics[width=0.95\textwidth]{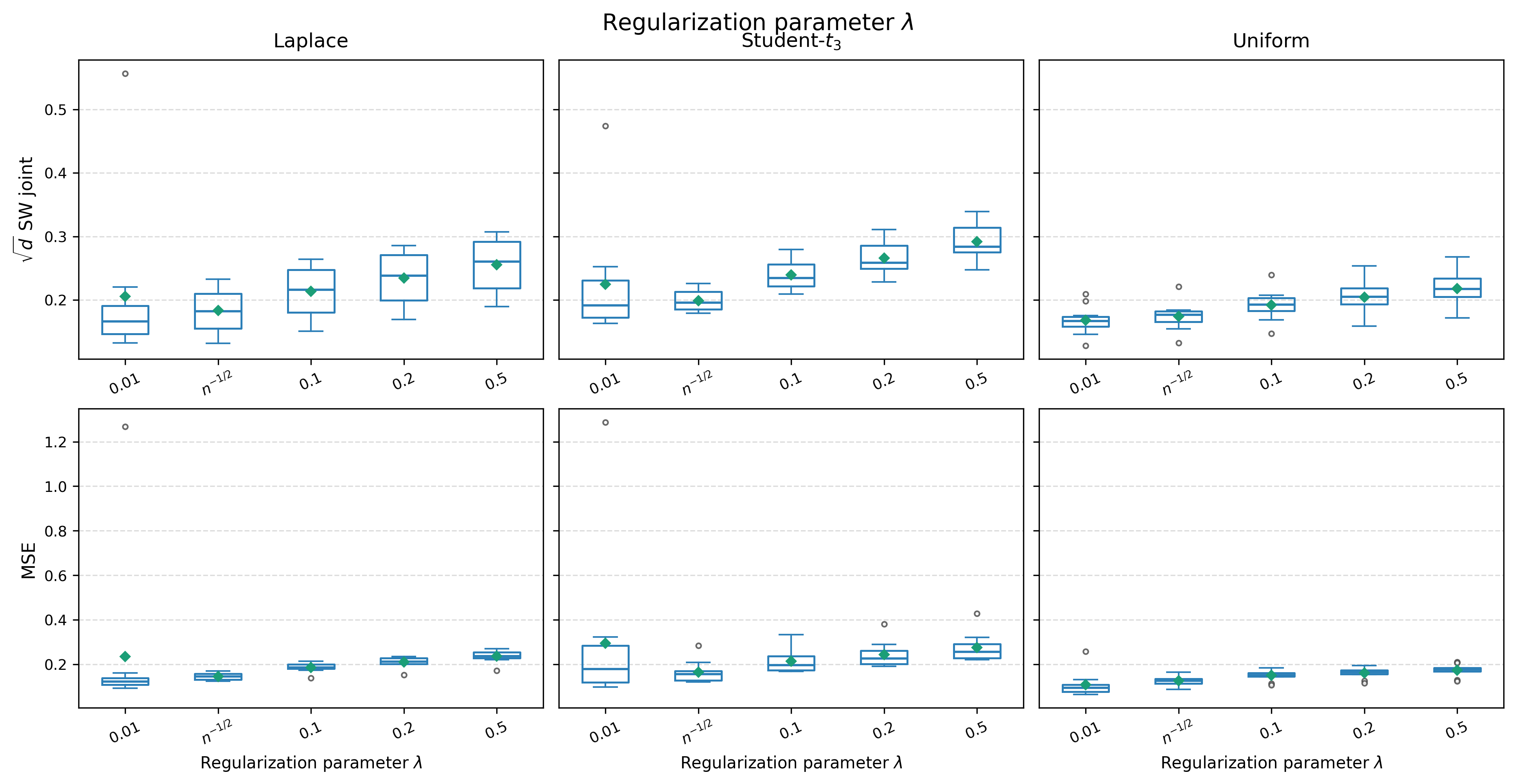}
		\captionof{figure}{The scaled joint sliced Wasserstein distance and MSE over different regularization parameters $\lambda$.}
		\label{fig:lam}
	\end{minipage}\par\vspace{0.35\baselineskip}
	
	\noindent\begin{minipage}{\textwidth}
		\centering
		\includegraphics[width=0.95\textwidth]{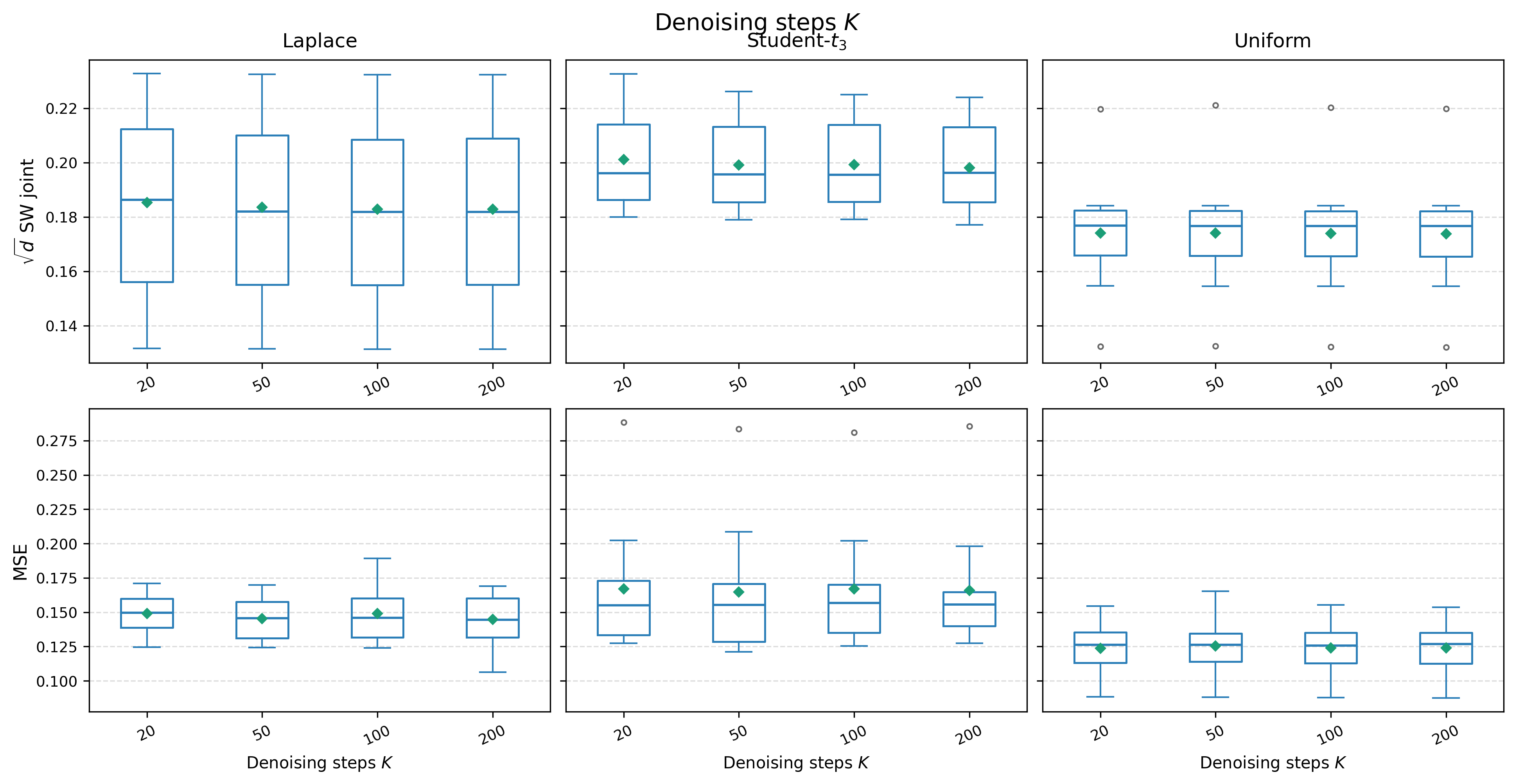}
		\captionof{figure}{The scaled joint sliced Wasserstein distance and MSE over different denoising steps $K$.}
		\label{fig:K}
	\end{minipage}\par\vspace{0.35\baselineskip}
	
	\noindent\begin{minipage}{\textwidth}
		\centering
		\includegraphics[width=0.95\textwidth]{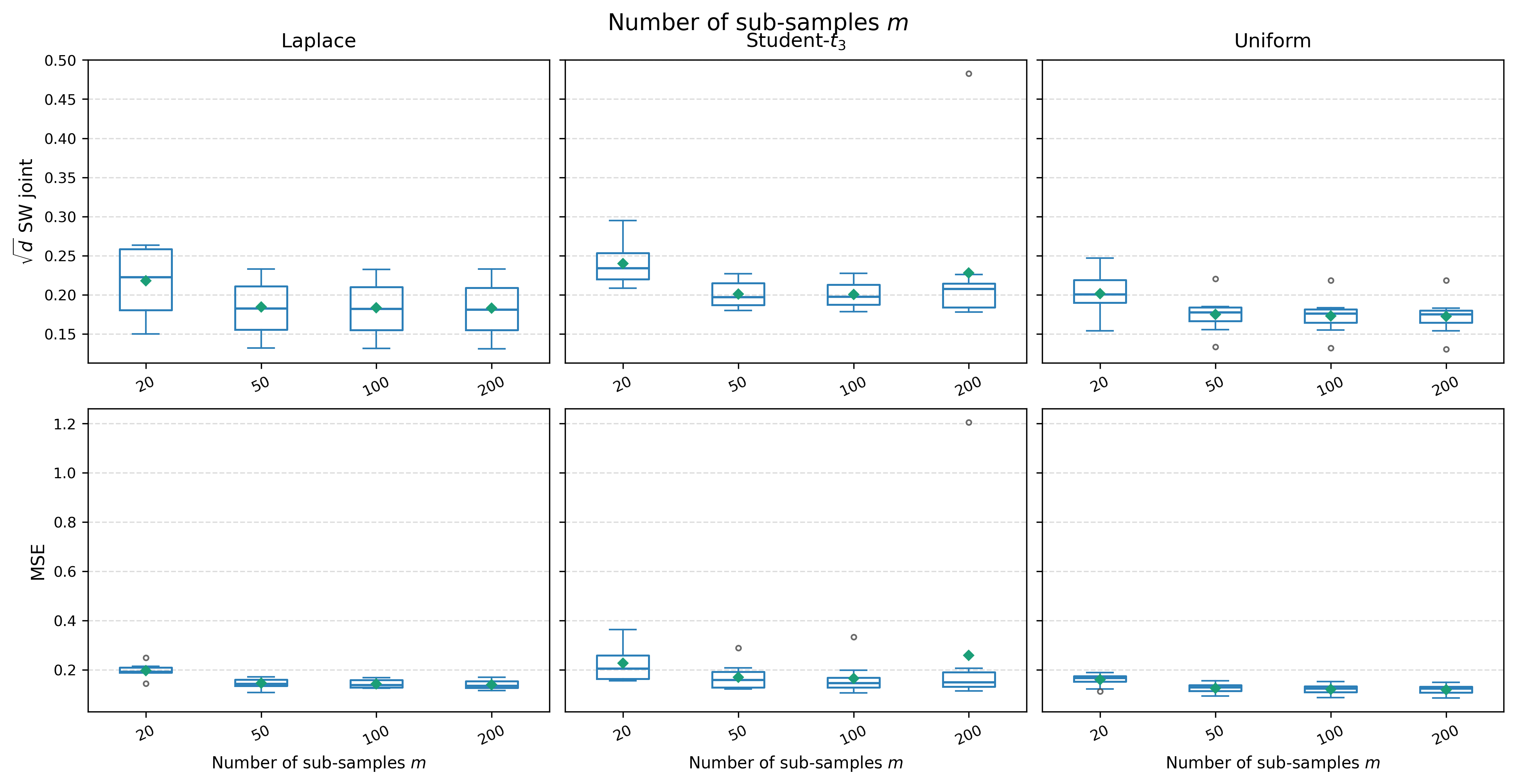}
		\captionof{figure}{The scaled joint sliced Wasserstein distance and MSE over different numbers of kernel centers $m$.}
		\label{fig:m}
	\end{minipage}\par\vspace{0.35\baselineskip}
	
	\noindent\begin{minipage}{\textwidth}
		\centering
		\includegraphics[width=0.95\textwidth]{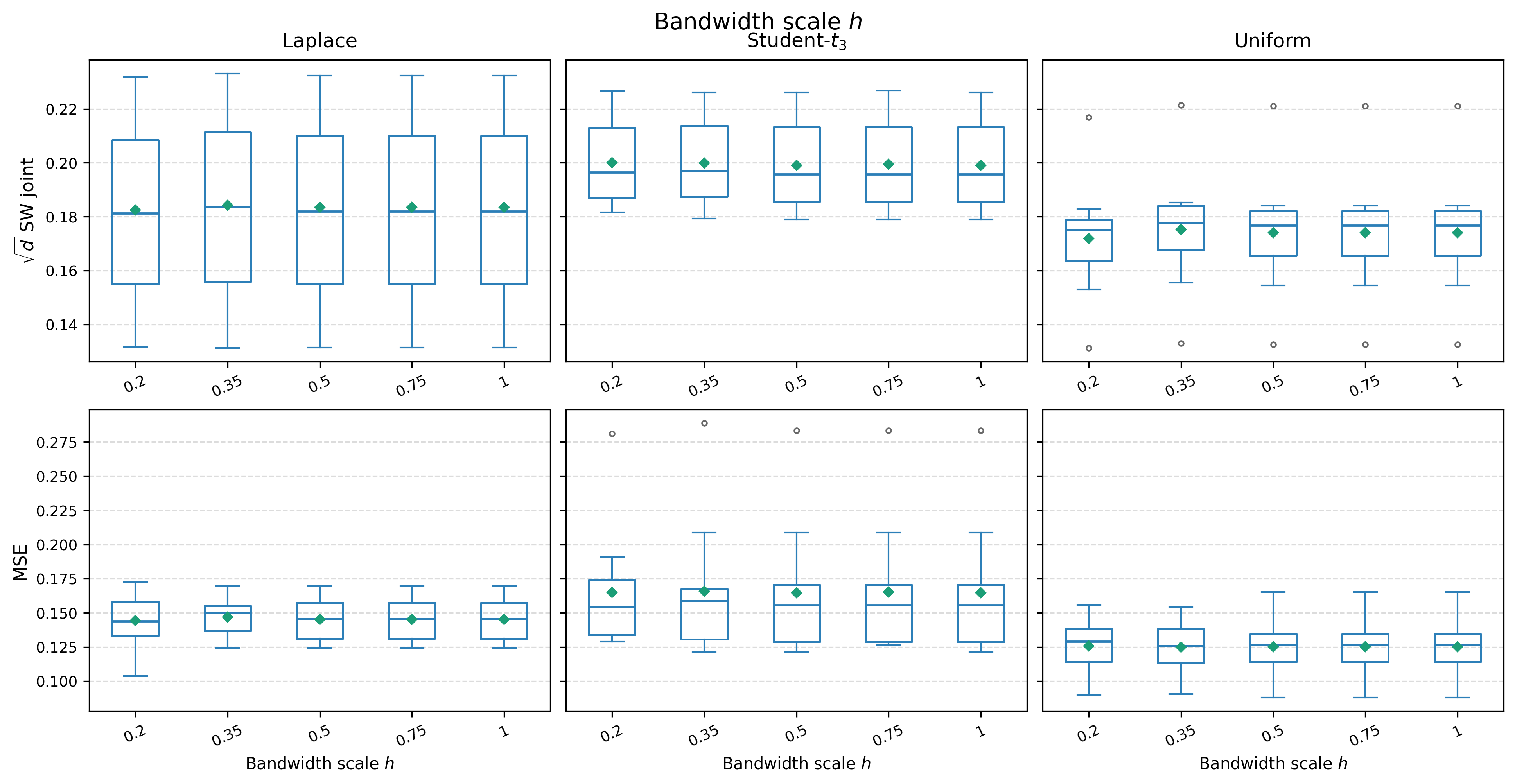}
		\captionof{figure}{The scaled joint sliced Wasserstein distance and MSE over different anisotropic bandwidth scales $h$.}
		\label{fig:h}
	\end{minipage}\par\vspace{0.35\baselineskip}
	
	\noindent\begin{minipage}{\textwidth}
		\centering
		\includegraphics[width=0.95\textwidth]{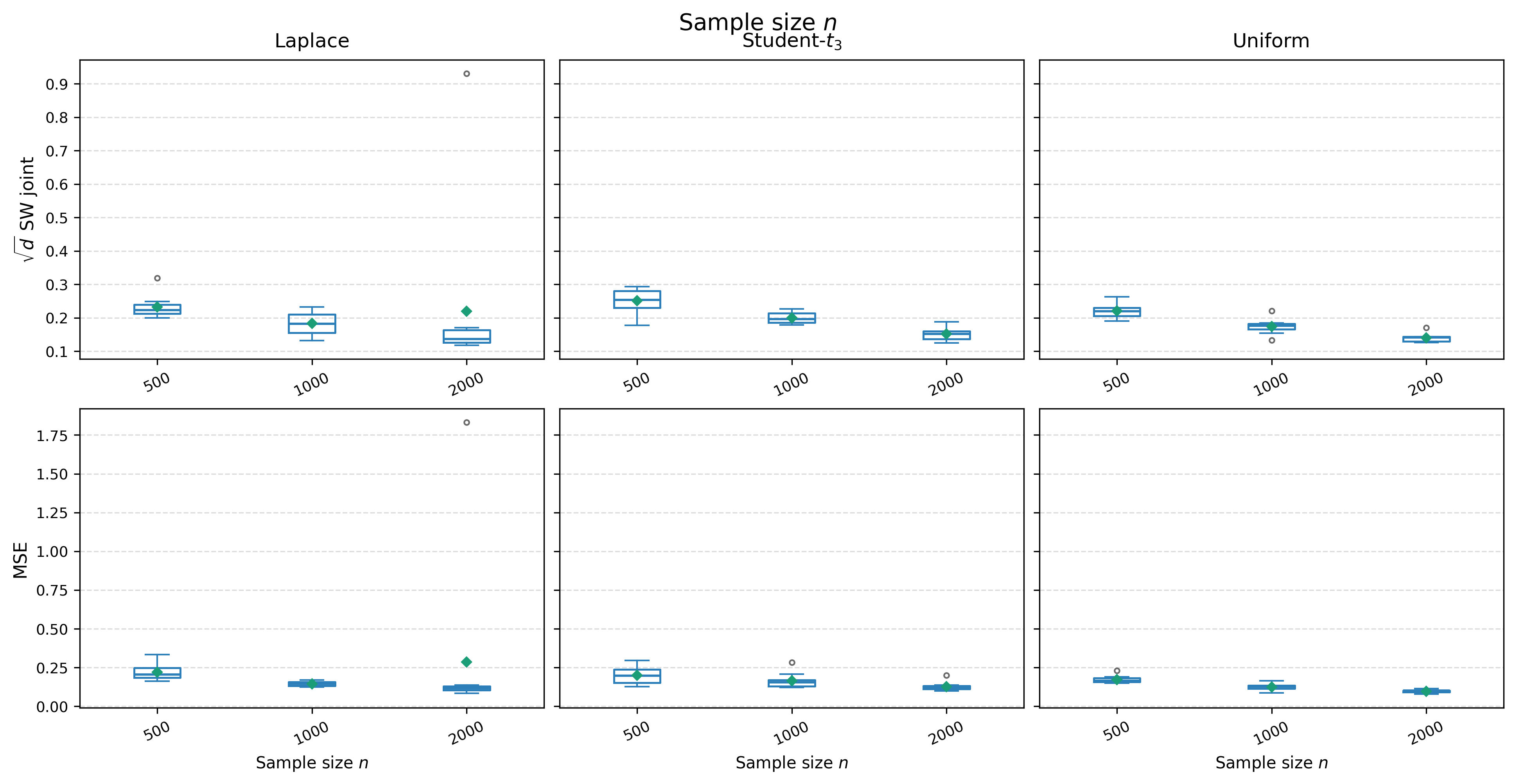}
		\captionof{figure}{The scaled joint sliced Wasserstein distance and MSE over different numbers of samples $n$.}
		\label{fig:n}
	\end{minipage}
	\endgroup
	
	\FloatBarrier
	
	\section{Unknown Noise Variance}\label{app:unknown-sigma}
	The main part of this paper assumes that the Gaussian measurement-error covariance is known. We now consider the practically important case where the error distribution is still Gaussian but its variance is unknown. In the simulation setting of Section \ref{sec:simulation}, the error is isotropic, $\bbSig=\sigma^{2}\bI$, and we estimate $\sigma$ from the contaminated covariates before running the denoising step with $\widehat{\bbSig}=\widehat{\sigma}^{2}\bI$.
	
	\paragraph{Original characteristic-function estimator and its applicability.}
	Write the contaminated covariates as $\bZ_1=\bZ_0+\bepsilon_{\bZ}$, where $\bepsilon_{\bZ}\sim\cN(\bzero,\sigma^2\bI_{d_{\bZ}})$ is independent of $\bZ_0$. Let $Z_{1,j}$ and $Z_{0,j}$ denote their $j$th scalar coordinates. Then
	\[
	\phi_{Z_{1,j}}(u)=\phi_{Z_{0,j}}(u)\exp\{-\sigma^{2}u^{2}/2\}.
	\]
	Hence
	\[
	\log|\phi_{Z_{1,j}}(u)|
	=
	\log|\phi_{Z_{0,j}}(u)|-\frac{\sigma^{2}}{2}u^{2}.
	\]
	The characteristic-function method of \citet{butucea2005minimax} is useful when the Gaussian factor dominates the high-frequency decay of the contaminated characteristic function. In particular, the latent characteristic function should not have a quadratic exponential decay that is indistinguishable from Gaussian noise, and it should not have persistent zeros in the frequency range used for estimation. A convenient finite-sample implementation assumes an ordinary-smooth nuisance tail and fits
	\[
	\log |\widehat{\phi}_{Z_{1,j}}(u)|
	\approx
	a_j-\frac{\sigma^{2}}{2}u^{2}-p_j\log u,
	\qquad
	\widehat{\phi}_{Z_{1,j}}(u)
	=
	\frac{1}{n}\sum_{i=1}^{n}\exp\{\mathrm{i}uZ_{1,j}^{(i)}\},
	\]
	over a moderate-frequency grid.\footnote{Notably, estimating $\sigma$ by this characteristic-function method results in a poor $\cO(1/\log{n})$ convergence rate.} We take the median of the coordinatewise estimates and truncate it at a small positive lower bound.
	
	The three latent covariate distributions used in our simulations satisfy these requirements to different degrees. For the Laplace design, the marginal characteristic function has ordinary-smooth polynomial decay and no zeros, so the original slope regression is well matched to the theory. For the Student-$t_3$ design, the tail of the characteristic function contains an exponential term of order $\exp(-c|u|)$; this is still lower order than the Gaussian $\exp(-\sigma^2u^2/2)$ term, so $\sigma$ is identifiable asymptotically, but the ordinary-smooth regression misspecifies the finite-sample nuisance tail. For the Uniform design, the characteristic function has the form $\sin(au)/(au)$ up to scaling. Its envelope is ordinary-smooth, but it has infinitely many zeros, which violates the non-vanishing regularity needed for a stable moderate-frequency slope fit.
	
	\paragraph{Results from the original plug-in scheme.}
	Figure \ref{fig:unknown-sigma-original} reports the estimated noise standard deviation and the corresponding denoising performance for $d_{\bZ}=3$ and $n=1000$. The denoising methods use the same plug-in estimate $\widehat{\sigma}$; Error-Contaminated ignores the measurement error, and Error-Free is the oracle benchmark. To reduce purely numerical instability from the plug-in step, the diffusion implementation selects the ODE trajectory whose covariance trace is closest to the plug-in target $\operatorname{tr}\{\widehat{\operatorname{Cov}}(\bX_1)\}-\operatorname{tr}(\widehat{\bbSig})$, subject to a mild displacement constraint. This plug-in ODE variant, including its data-dependent stopping rule, differs from the fixed-endpoint Euler--Maruyama sampler analyzed in Theorem \ref{thm:sampling error}; that theorem does not provide an error bound for this variant. 
	% This selection rule only uses contaminated data and $\widehat{\sigma}$.
	
	\begin{figure}[!tbp]
		\centering
		\includegraphics[height=0.74\textheight,keepaspectratio]{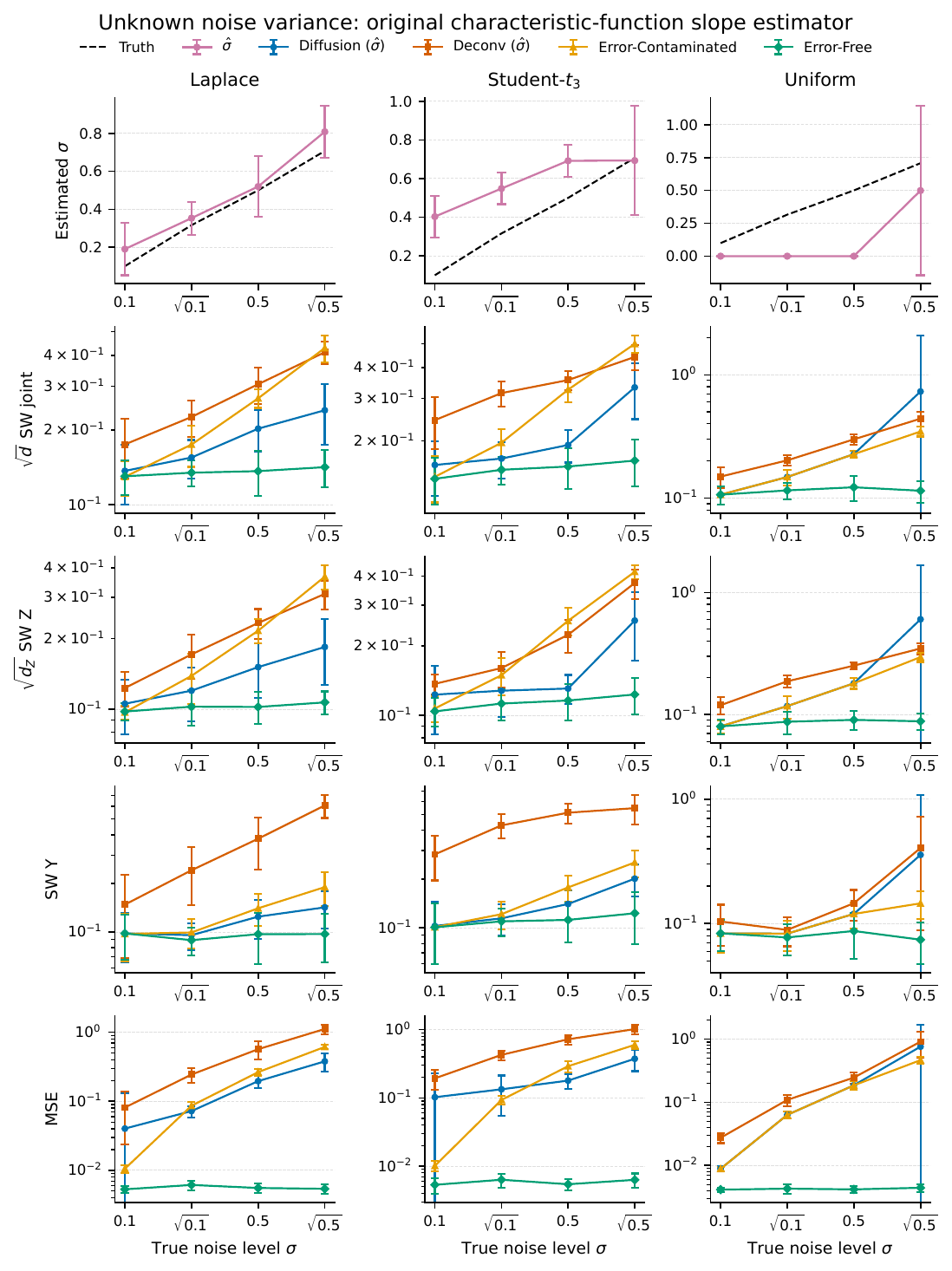}
		\caption{Unknown noise variance with the original characteristic-function slope estimator. The first row reports $\widehat{\sigma}$ against the true $\sigma$. The remaining rows report $\sqrt{d}$-scaled joint sliced Wasserstein distance, $\sqrt{d_{\bZ}}$-scaled sliced Wasserstein distance for $\bZ$, sliced Wasserstein distance for $Y$, and MSE. Values are means over 10 repetitions and are not multiplied by 1000; error bars are one standard deviation, and the performance axes are on the logarithmic scale.}
		\label{fig:unknown-sigma-original}
	\end{figure}
	
	The numerical patterns agree with the preceding regularity discussion. The Laplace design is the most favorable case for the original estimator, because its ordinary-smooth characteristic-function tail is exactly the nuisance structure used in the regression. The Student-$t_3$ design tends to overestimate small and moderate noise levels: part of the linear exponential decay of the latent characteristic function is incorrectly attributed to the Gaussian error. This overestimation can induce over-denoising and larger MSE. The Uniform design is the least reliable case. The zeros and oscillations of its characteristic function make the slope estimate highly sensitive to the chosen frequency window, and the plug-in procedure often behaves like an error-contaminated analysis when the estimate is truncated near zero. Thus, the results indicate that the denoising methods are not suitable for all clean data distributions when the noise variance is unknown. Fortunately, our \texttt{Diffusion} method outperforms the \texttt{Deconv} method in most cases. 
	
	\paragraph{Distribution-aware refinement.}
	If the rough form of the latent $\bZ_0$ distribution is known, the plug-in step can be improved by modelling the nuisance characteristic-function tail more explicitly. In our implementation, we first take an upper-envelope summary of $\log|\widehat{\phi}_{Z_{1,j}}(u)|$ over frequency bins to reduce the effect of empirical dips and oscillatory zeros. We then fit
	\[
	\log |\widehat{\phi}_{Z_{1,j}}(u)|
	\approx
	a_j-\frac{\sigma^{2}}{2}u^{2}-b_j|u|-p_j\log u.
	\]
	The additional $|u|$ term is intended to absorb the first-order exponential tail that appears, for example, in Student-$t$ characteristic functions, while the upper-envelope step is designed to avoid fitting directly through local zeros. This refinement is not fully model-free: it uses qualitative information about the clean data. If stronger prior information, replicate measurements, or calibration data are available, those sources should be preferred for estimating $\bbSig$.
	
	\begin{figure}[t!]
		\centering
		\includegraphics[height=0.74\textheight,keepaspectratio]{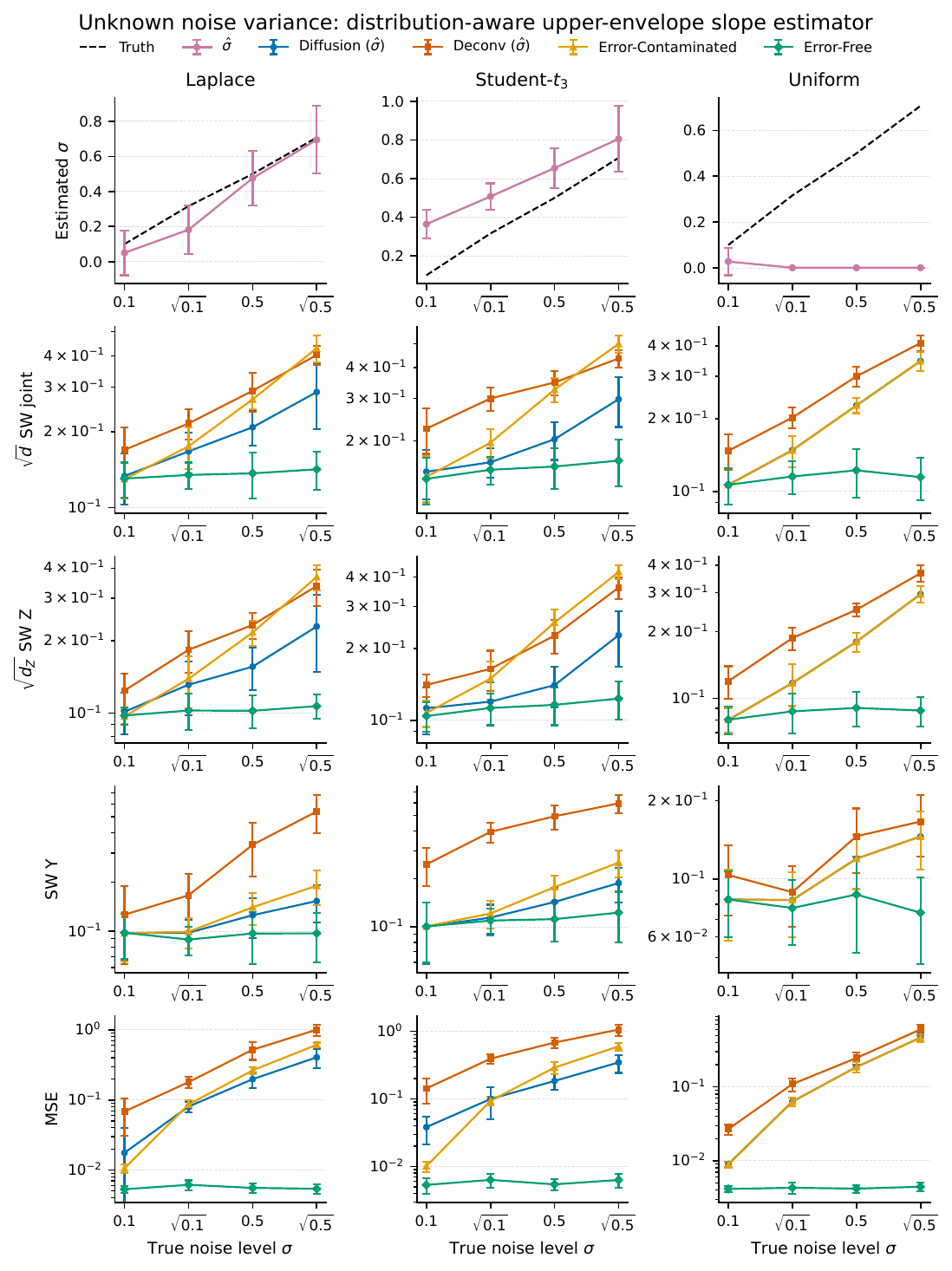}
		\caption{Unknown noise variance with the distribution-aware upper-envelope slope estimator. The plotting convention is the same as in Figure \ref{fig:unknown-sigma-original}.}
		\label{fig:unknown-sigma-refined}
	\end{figure}
	
	Figure \ref{fig:unknown-sigma-refined} shows that the refined estimator behaves as expected. It is more robust when the latent characteristic function has non-polynomial lower-order decay, and it reduces some of the large finite-sample fluctuations that occur when the original slope regression attributes the signal tail to noise. However, the refinement does not solve the fundamental difficulty of compactly supported or strongly oscillatory clean-data distributions. In the Uniform case, the upper envelope mitigates local zeros but does not fully restore identifiability from a short frequency window, so the resulting denoising performance remains close to the error-contaminated baseline for some noise levels.
	
	\paragraph{Role of the response variable.}
	The preceding variance estimation is based on $\bZ$ rather than $Y$. This is deliberate. In our simulations, $Y_0=f_{\bbeta_{\mathrm{truth}}}(\bZ_0)$ is generated by a nonlinear neural network, so its marginal distribution is an unknown push-forward distribution depending jointly on $f_{\bbeta_{\mathrm{truth}}}$ and the distribution of $\bZ_0$. Unlike the covariate distributions considered above, the characteristic function of $Y_0$ does not have a simple known ordinary-smooth, Student-type, or compact-support form. Estimating $\sigma$ from $Y$ would therefore introduce an additional and hard-to-control modelling error.
	
	Nevertheless, the response variable has a direct effect on the final performance metrics. The joint sliced Wasserstein distance measures the distribution of $(\bZ,Y)$, and the MSE is computed after training a regression function on the denoised pairs. If $\widehat{\sigma}$ is too small, residual measurement error remains in both $\bZ$ and $Y$; if it is too large, the denoising step can oversmooth the joint sample and distort the functional relation between $\bZ$ and $Y$. Thus the accuracy of $\widehat{\sigma}$ is only one part of the final result: the nonlinear and unknown distribution of $Y$ determines how noise-estimation error is amplified in the downstream regression task.
	
	\section{Learning in Reproducing Kernel Hilbert Spaces: Some Background and Justification of Tuning Parameter Impact}\label{app:rkhs}

	%\textbf{1.~Representer Theorem and Kernel Methods.} \\
	Many learning problems in RKHS can be formulated as minimizing an empirical loss plus a regularization term in a \emph{Reproducing Kernel Hilbert Space} (RKHS). Concretely, one often considers a problem of the form
	\begin{equation}
		\label{eq:tikhonov}
		\min_{f \in \cH_{\cK,\bbR}} 
		\left\{
		\frac{1}{n}\sum_{i=1}^{n} \left(Y^{(i)} -  f\left(\bX^{(i)}\right)\right)^{2} + 
		\lambda\|f\|_{\cH_{\cK}}^2\right\},
	\end{equation}
	\noindent where $\cH_{\cK,\bbR}$ is the real subspace defined in Section \ref{sec:notations}, with inherited norm from $\cH_{\cK}$ and kernel $\cK$ restricted to real inputs, $\{(\bX^{(i)},Y^{(i)})\}_{i=1}^n\subset\bbR^d\times\bbR$ are training data, and $\lambda > 0$ is the regularization parameter. The population score-matching objective in this paper has an analogous squared-error interpretation for each dimension of the target score function; that is, $\bX^{(i)}$ corresponds to $\bX_{t}^{(i)}$ and $Y^{(i)}$ corresponds to one specific dimension of $\nabla\log{p_{t}}(\bX_{t}^{(i)})$, where $\bX_{t}^{(i)}\sim P_{\bX_{t}}$. Although these data are unobservable, our method in Section \ref{sec:DEP} resolves this issue.   
	\par
	By the representer theorem, the minimizer $f$ of \eqref{eq:tikhonov} takes the form
	\begin{equation*}
		f(\bx) = \sum\limits_{i=1}^{n}\alpha_{i}\cK(\bx, \bX^{(i)}),
	\end{equation*}
	where one choice of the coefficient vector $\boldsymbol{\alpha} = (\alpha_{1}, \cdots, \alpha_{n})^{\top}$ is
	\begin{equation*}
		\boldsymbol{\alpha} = \left(\BK + n\lambda \bI_n\right)^{-1}\bY,
	\end{equation*}
	where $\BK$ is an $n\times n$ matrix whose $(i,j)$-th element is $\cK(\bX^{(i)}, \bX^{(j)})$, and $\bY = (Y^{(1)},\cdots, Y^{(n)})^{\top}$. The fitted-value vector is $\BK\boldsymbol{\alpha}=\BK(\BK+n\lambda\bI_n)^{-1}\bY$. The regularization parameter $\lambda>0$ makes the matrix $\BK+n\lambda\bI_n$ invertible, while increasing $\lambda$ moves the obtained $f$ away from the ground-truth target. Notably, this gap has an explicit dependence on $\lambda$, which usually depends on the smoothness of the target function, i.e., $\nabla\log{p_{t}}$ in this paper \citep{smale2003estimating,de2005learning}. 
Under Assumption \ref{ass:subexp}, the approximation error is controlled by $\kappa(b)$, while regularization and concentration depend on $\cB(b)$, $\sigma_0(b)$ and $H_0(b)$. Section \ref{sec:generalization analysis} balances these terms using $\lambda\asymp n^{-1/2}\log(2/\eta)\log K$ for a bound uniform over the time grid; the resulting rate requires retaining the approximation remainder.
	\par
	In this paper, we set the kernel function $\cK(\cdot, \cdot)$ as the Gaussian kernel 
	\begin{equation*}
		\cK(\bx_{1}, \bx_{2}) = \exp\left\{-(\bx_{1} - \bar{\bx}_{2})^{\top}\bbH(\bx_{1} - \bar{\bx}_{2})\right\},
	\end{equation*}
	in the complex coordinate space.
	As suggested above, instead of optimizing over infinitely many possible functions in \(\cH_{\cK,\bbR}\), one only needs to solve for finitely many coefficients, as we did in Section \ref{sec:kernel-based function} by restricting $\bs_t(\cdot; \btheta_t)\in \cH_{\cK_t,m}^{d}$. The Gaussian kernel is analytic and serves as a universal approximator. Proposition \ref{prop:approximation error} shows that functions in $\cH_{\cK_{t}, m}^{d}$ can approximate the target score function $\nabla\log{p_{t}}$ well. 
	\par
	Notably, for the model based on RKHS \citep{de2005learning,smale2003estimating,smale2007learning}, our theoretical analysis treats $\bbH$ as fixed rather than deriving an optimal bandwidth sequence as in classical kernel density estimation. The simulations instead use the empirical covariance-based choice described in Section \ref{sec:simulation}. To make such a selection, one would need to explicitly compute the dependence of $\cB(b)$ on $\bbH$, which is difficult; we leave this for future work.  
	\section{Sliced Wasserstein Distance}\label{app:sw-distance}
	In the main simulation and the ablation study, we use the sliced Wasserstein distance (SW) to evaluate the distributional gap between the generated denoised sample and an independent error-free test sample. Let $P$ and $Q$ be two distributions on $\mathbb{R}^{d}$. For a unit vector $\boldsymbol{\omega}\in\mathbb{S}^{d-1}$, denote by $P_{\boldsymbol{\omega}}$ and $Q_{\boldsymbol{\omega}}$ the one-dimensional distributions of $\boldsymbol{\omega}^{\top}\bX$ with $\bX\sim P$ and $\boldsymbol{\omega}^{\top}\bY$ with $\bY\sim Q$, respectively. The unscaled first-order sliced Wasserstein distance is
	\begin{equation*}
		\mathrm{SW}_{1}(P,Q)
		=
		\int_{\mathbb{S}^{d-1}} \sW_{1}(P_{\boldsymbol{\omega}},Q_{\boldsymbol{\omega}})\,d\boldsymbol{\omega},
	\end{equation*}
	where $\sW_{1}$ denotes the one-dimensional Wasserstein distance and $d\boldsymbol{\omega}$ is the uniform probability measure on the unit sphere. The scaled quantity reported in our simulations is
	\[
	\mathrm{SW}_{1}^{\mathrm{sc}}(P,Q)=\sqrt{d}\,\mathrm{SW}_{1}(P,Q).
	\]
	Thus, figure-axis labels of the form $\sqrt{d}\,\mathrm{SW}$ refer to $\mathrm{SW}_{1}^{\mathrm{sc}}$; the dimension factor is applied once.
	
	We estimate $\mathrm{SW}_{1}^{\mathrm{sc}}$ by Monte Carlo random projections. Specifically, we draw $N_{\mathrm{proj}}=256$ independent random directions $\boldsymbol{\omega}_{1},\ldots,\boldsymbol{\omega}_{N_{\mathrm{proj}}}$ from the unit sphere by normalizing standard Gaussian vectors. For each direction, we project the two empirical samples and compute the one-dimensional Wasserstein distance by integrating the absolute difference of their empirical quantile functions. This also accommodates unequal sample sizes; for equal sizes it reduces to matching the corresponding order statistics. When the two empirical samples are $\{\bx_{i}\}_{i=1}^{n}$ and $\{\by_{i}\}_{i=1}^{n}$, our estimator is
	\begin{equation*}
		\widehat{\mathrm{SW}}_{1}^{\mathrm{sc}}(P,Q)
		=
		\frac{\sqrt{d}}{N_{\mathrm{proj}}}\sum_{\ell=1}^{N_{\mathrm{proj}}}
		\frac{1}{n}\sum_{i=1}^{n}
		\left|
		x_{(i),\ell}-y_{(i),\ell}
		\right|,
	\end{equation*}
	where $x_{(i),\ell}$ and $y_{(i),\ell}$ are the $i$th order statistics of $\{\boldsymbol{\omega}_{\ell}^{\top}\bx_{i}\}_{i=1}^{n}$ and $\{\boldsymbol{\omega}_{\ell}^{\top}\by_{i}\}_{i=1}^{n}$, respectively. In the one-dimensional case, this reduces to the usual empirical $\sW_{1}$ distance.
	
	The factor $\sqrt{d}$ is a dimensional normalization. For any fixed vector $\bv\in\mathbb{R}^{d}$ and a random unit direction $\boldsymbol{\omega}$, $\mE_{\boldsymbol{\omega}}[(\boldsymbol{\omega}^{\top}\bv)^{2}]=\|\bv\|_{2}^{2}/d$. Therefore, a typical one-dimensional projection is smaller than the original Euclidean displacement by a factor of order $1/\sqrt{d}$. Multiplying by $\sqrt{d}$ compensates for the scale reduction induced by one-dimensional projection; it does not generally make distributional discrepancies dimension-invariant. For a fixed dimension, this scaling does not change the relative ranking of competing methods.   
	\section{Additional Figures in the Diabetes Clinical Case Study}\label{app:figure for diabetes}
	The variations in glucose time-in-range metrics for control and treatment groups at different noise levels are in Figure \ref{fig:ejemplo}. 
	\begin{figure}[!htbp]
		\centering
		\includegraphics[angle=270,width=0.48\textwidth]{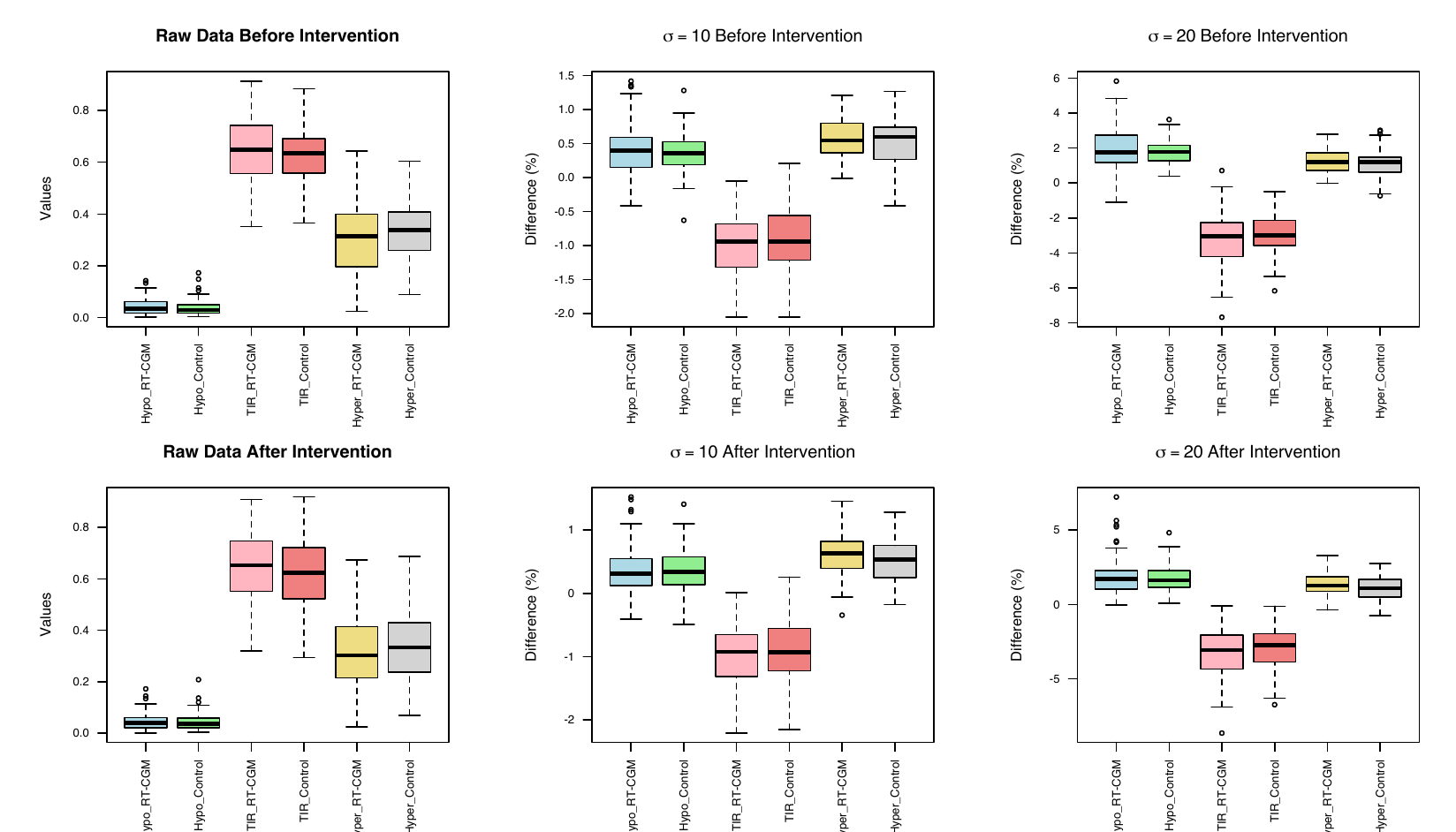}
		\caption{Variations in glucose time-in-range metrics for control and treatment groups at noise levels \(\sigma = 0, 10, 20\) before and after interventions.}
		\label{fig:ejemplo}
	\end{figure}
	\newpage

\clearpage
\bibliographystyle{asa}
\bibliography{ref,ref_marcos}
\end{document}